\documentclass[pra,twocolumn,showpacs,groupedaddress,notitlepage,superscriptaddress,floatfix,nofootinbib]{revtex4-2}
\usepackage[english]{babel}
\usepackage{amsmath}
\usepackage{mathtools}
\usepackage{amssymb}
\usepackage{amsthm}
\usepackage{braket}
\usepackage{xcolor}
\usepackage{graphicx}
\usepackage{slashed}
\usepackage{hyperref}
\hypersetup{
  colorlinks=true,
  linkcolor=cyan,
  citecolor=magenta,
  filecolor=magenta,
  urlcolor=cyan,
  runcolor=cyan}
\usepackage[capitalise]{cleveref}

\usepackage{listings}
\newtheorem{mytheorem}{Theorem}
\newtheorem{mylemma}{Lemma}
\newtheorem{mycorollary}{Corollary}
\newcommand{\abs}[1]{\left|#1\right|}
\newcommand{\floor}[1]{\left\lfloor#1\right\rfloor}
\newcommand{\ceil}[1]{\left\lceil#1\right\rceil}
\DeclareMathOperator{\syn}{syn}
\newcommand{\tmcEc}{{\tilde{\mathcal{E}}_c}}
\DeclareMathOperator{\wt}{wt}
\newcommand{\ketbra}[1]{|{#1}\rangle\!\langle#1|}
\newcommand{\mC}{\mathcal{C}}
\newcommand{\mL}{\mathcal{L}}

\newcommand{\mP}{\mathcal{P}}
\newcommand{\mS}{\mathcal{S}}
\newcommand{\tmP}{\tilde{\mP}}
\newcommand{\tG}{\tilde{G}}
\newcommand{\tGp}{\tG^{\perp}}
\newcommand{\pmS}{\pi(\mS)}
\newcommand{\pmSL}{\pi(\mS\mL)}
\DeclareMathOperator{\Tr}{Tr}
\newcommand{\pDD}{p_{\rm DD}}
\newcommand{\pQEC}{p_{\rm QEC}}
\newcommand{\pQED}{p_{\rm QED}}
\newcommand{\FDD}{F_{\rm DD}}
\newcommand{\FLDD}{F_{\rm LDD}}
\newcommand{\FQEC}{F_{\rm QEC}}
\newcommand{\FQED}{F_{\rm QED}}
\newcommand{\FHyb}{F_{\rm Hyb}}
\newcommand{\FHybD}{F_{\rm HybD}}
\newcommand{\all}{\textup{A}}
\newcommand{\C}{\textup{C}}
\newcommand{\D}{\textup{D}}
\newcommand{\symL}{\textup{L}}
\newcommand{\symS}{\textup{S}}
\newcommand{\St}{\textup{St}}
\newcommand{\StL}{\textup{StL}}
\newcommand{\uS}{\textup{\ensuremath{\not\!\mathrm{S}}}}
\newcommand{\uC}{\textup{\ensuremath{\not\!\mathrm{C}}}}
\newcommand{\SSt}{\textup{S-St}}
\newcommand{\uSSt}{\textup{\ensuremath{\not\!\mathrm{S}}-\textup{St}}}
\newcommand{\SL}{\textup{S-L}}
\newcommand{\uSL}{\textup{\ensuremath{\not\!\mathrm{S}}-L}}
\newcommand{\SStL}{\textup{S-StL}}
\newcommand{\uSStL}{\textup{\ensuremath{\not\!\mathrm{S}}-StL}}
\newcommand{\uSuC}{\textup{\ensuremath{\not\!\mathrm{S}}-\!\ensuremath{\not\!\mathrm{C}}}}
\newcommand{\uSC}{\textup{\ensuremath{\not\!\mathrm{S}}-C}}
\newcommand{\uSD}{\textup{\ensuremath{\not\!\mathrm{S}}-D}}
\newcommand{\SuC}{\textup{S-\!\ensuremath{\not\!\mathrm{C}}}}
\newcommand{\SC}{\textup{S-C}}
\newcommand{\SD}{\textup{S-D}}
\newcommand{\uCD}{\textup{\ensuremath{\not\!\mathrm{C}}-D}}
\newcommand{\uSuSt}{\textup{\ensuremath{\not\!\mathrm{S}}-\ensuremath{\slashed{\mathrm{St}}}}}
\newcommand{\SuSt}{\textup{S-\ensuremath{\slashed{\mathrm{St}}}}}
\newcommand{\uSt}{\textup{\ensuremath{\slashed{\mathrm{St}}}}}

\newcommand{\eref}{P_{\mathrm{ref}}}
\newcommand{\Id}{\pi(I)}

\crefname{mytheorem}{Theorem}{Theorems}
\Crefname{mytheorem}{Theorem}{Theorems}
\crefname{mylemma}{Lemma}{Lemmas}
\Crefname{mylemma}{Lemma}{Lemmas}
\crefname{mycorollary}{Corollary}{Corollaries}
\Crefname{mycorollary}{Corollary}{Corollaries}

\begin{document}

\title{Quantum Error Correction and Dynamical Decoupling: Better Together or Apart?}

\author{Victor Kasatkin}
\affiliation{Department of Electrical \& Computer Engineering, University of Southern California, Los Angeles, California 90089, USA}
\affiliation{Center for Quantum Information Science \& Technology, University of Southern California, Los Angeles, California 90089, USA}

\author{Mario Morford-Oberst}
\affiliation{Department of Electrical \& Computer Engineering, University of Southern California, Los Angeles, California 90089, USA}
\affiliation{Center for Quantum Information Science \& Technology, University of Southern California, Los Angeles, California 90089, USA}

\author{Arian Vezvaee}
\affiliation{Department of Electrical \& Computer Engineering, University of Southern California, Los Angeles, California 90089, USA}
\affiliation{Center for Quantum Information Science \& Technology, University of Southern California, Los Angeles, California 90089, USA}
\affiliation{Quantum Elements, Inc., 2829 Townsgate Road, Westlake Village, California 91361, USA}

\author{Daniel A. Lidar}
\affiliation{Department of Electrical \& Computer Engineering, University of Southern California, Los Angeles, California 90089, USA}
\affiliation{Center for Quantum Information Science \& Technology, University of Southern California, Los Angeles, California 90089, USA}
\affiliation{Quantum Elements, Inc., 2829 Townsgate Road, Westlake Village, California 91361, USA}
\affiliation{Department of Physics \& Astronomy, University of Southern California, Los Angeles, California 90089, USA}
\affiliation{Department of Chemistry, University of Southern California, Los Angeles, California 90089, USA}

\begin{abstract}
Quantum error correction/detection (QEC/QED) and dynamical decoupling (DD)
are widely used tools for protecting quantum information,
and a natural goal is to combine them to outperform either approach alone.
Such a benefit is not automatic:
physical DD pulses can conflict with an encoded subspace,
and correction/detection performance is determined by the errors that survive decoding or evade postselection,
which need not be those suppressed by DD.
We analyze a hybrid memory cycle in which DD is implemented logically (LDD)
using elements of the normalizer of an $[[n,k,d]]$ stabilizer code, and the LDD is followed by syndrome measurement and (in the correction setting) by recovery or (in the detection setting) by postselection on a trivial syndrome.
In an effective Pauli model with physical error probability $p$,
LDD suppression factor $\pDD$, and phenomenological recovery/readout imperfection rates $\pQEC$ and $\pQED$,
we derive closed-form entanglement-fidelity expressions for QEC-only, QED-only, LDD-only, physical DD (as a baseline), and the hybrid protocols LDD+QEC and LDD+QED, including conditional-fidelity and acceptance-probability formulas in the QED setting.
The formulas are expressed via a small set of code-dependent weight enumerator polynomials, making explicit the role of the recovery map, the postselection rule, and the LDD group.
For ideal recovery we obtain a necessary-and-sufficient criterion for when LDD+QEC outperforms QEC-only, and in the low-noise regime we give a simple sufficient design rule:
it is enough that LDD suppresses at least one minimum-weight uncorrectable Pauli error for the chosen recovery map;
stabilizer-equivalent choices of LDD generators can be used to enforce this condition.
For the detection setting we derive the corresponding conditional-fidelity and acceptance-probability formulas,
together with an exact comparison criterion under ideal readout for when LDD+QED improves upon QED-only.
We supplement our analysis with numerical results for the $[[7,1,3]]$ Steane code, a $[[13,1,3]]$ code, the ``perfect'' $[[5,1,3]]$ code, and the $[[4,2,2]]$ error-detecting code, mapping regions of hybrid-protocol advantage in parameter space beyond the small-$p$ regime.
Our work illustrates the need for co-design of the code, the recovery/detection rule, and the logical decoupling group, and clarifies the conditions under which the hybrid LDD+QEC and LDD+QED protocols are advantageous.
\end{abstract}

\maketitle

\section{Introduction}
\label{sec:introduction}

Protecting quantum information against noise is essential for long-lived quantum memories
and for fault-tolerant quantum computation~\cite{Gaitan:book,Lidar-Brun:book,Terhal:2015dq}.
Two broad families of protection methods are syndrome-based coding and dynamical decoupling (DD).
In syndrome-based protocols, one either actively corrects errors using quantum error correction (QEC)
or detects them and postselects on trivial syndromes, i.e., quantum error detection
(QED)~\cite{Shor1995PRA,Steane:96a,Vaidman:1996vs,Knill:1997kx,Calderbank:98,Knill:2000dq,Kribs:2005:180501,poulin_stabilizer_2005,Bacon:05}.
DD, by contrast, uses open-loop control to average out certain error
mechanisms~\cite{Viola:98,Duan:98e,Zanardi1999fk,Vitali:99,ShiokawaLidar:02,Khodjasteh:2005xu,Uhrig:2007qf,West:2010:130501,Wang:10,Genov2017PRL,Yi:2026aa}.

It is natural to try to combine QEC/QED and DD to improve overall performance:
DD is commonly expected to help against slowly varying/coherent or low-frequency error components,
whereas QEC/QED is often formulated for stochastic effective noise.
Thus, one might expect that, given a fixed hardware noise environment and a fixed code,
adding DD would typically improve the performance of a QEC or QED memory cycle
(assuming additional errors introduced by DD pulse imperfections are negligible).
It is, however, a priori unclear whether this is always the case,
or whether there are situations in which combining QEC or QED with DD provides little benefit,
or even degrades the performance.
The combination of QEC and DD has been studied in terms of resource overhead and the fault-tolerance accuracy threshold in a hybrid DD-QEC setting,
where DD pulses are implemented at the physical level~\cite{Ng:2011dn}.
An alternative approach is to apply DD as logical operations of a code
(encoded DD~\cite{Viola:01a}), in particular as logical dynamical decoupling
(LDD), in which the decoupling operations are chosen to be Pauli strings preserving
the codespace (e.g., by using elements of the code normalizer) while averaging away
Pauli errors that do not commute with a specified logical decoupling
group~\cite{PhysRevLett.100.160506,Paz-Silva:2013tt,vezvaee2025demonstrationhighfidelityentangledlogical}.
It is an open question whether hybrid LDD+QEC and LDD+QED protocols are always better than either approach in isolation.

In this work, we analyze hybrid strategies in which LDD---or more generally any decoupling construction captured by a nontrivial Pauli subgroup $\tG\subset\tmP_n$ in our effective model---is implemented on an encoded block and is followed by stabilizer syndrome processing.
In the correction setting this means syndrome measurement and recovery;
in the detection setting it means syndrome measurement and postselection on a trivial reported syndrome.
More specifically, we focus on a single ``encode-wait-syndrome-processing'' cycle for an $[[n,k,d]]$ stabilizer code and compare four strategies in each setting.
In the correction setting these are QEC-only, LDD-only, physical DD on $k$ unencoded qubits (as a baseline), and the hybrid protocol LDD+QEC.
In the detection setting they are QED-only, LDD-only, physical DD, and the hybrid protocol LDD+QED.
Instead of attempting a microscopic description of the noise during the wait (memory) interval,
we adapt a phenomenological effective model for the combined effect of the physical
noise and DD during the whole wait interval.
In this model, the noise is described by an effective Pauli channel
parameterized by a physical Pauli error probability $p$,
and DD suppression is parameterized by $\pDD \in [0, 1]$,
which rescales the contribution of errors that do not commute with the chosen decoupling group (followed by renormalization).
For example, $\pDD = 1$ corresponds to no DD-induced reshaping of the effective Pauli distribution,
while $\pDD < 1$ corresponds to progressively stronger reweighting of the DD-suppressed sector. The endpoint $\pDD = 0$ is the idealized limit in which the effective residual channel has no support on Pauli errors that anticommute with the chosen decoupling group, so that all remaining residual noise lies in the DD-unsuppressed sector.
It should be noted that, had we assumed a microscopic depolarizing noise model,
DD would not be able to suppress it because DD is ineffective against genuinely Markovian white noise.
This situation would correspond to $\pDD = 1$ in our effective phenomenological model.
Imperfections in recovery or syndrome readout/postselection are captured by phenomenological parameters $\pQEC \in [0, 1]$ and $\pQED \in [0, 1]$.

In the correction setting we quantify performance by the entanglement fidelity $F$,
equivalently the probability of no net logical fault after the protocol.
In the detection setting we track both the probability $P_A$ of accepting a run
and the logical fidelity conditioned on acceptance.
Our main goal is to extract checkable criteria for when LDD+QEC improves upon QEC-only and when LDD+QED improves upon QED-only, as well as to compare both hybrid protocols to LDD-only and physical DD.
To do so, we derive the dependence of the relevant performance metrics on the code, the recovery or postselection rule, and the LDD group.
Due to the generality of our approach, many results developed for LDD apply, in fact, to any nontrivial Pauli decoupling group
(even if it does not consist of logical operators).

We summarize our main results and takeaways as follows.
\begin{enumerate}
  \item \textit{Analytical performance formulas.}
  We derive closed-form expressions for the QEC entanglement fidelities and the QED conditional fidelities/acceptance probabilities of the strategies considered here
  in terms of a small set of weight-enumerator polynomials (WEPs)
  that count Pauli errors by weight within subsets defined by the stabilizer code,
  the recovery or detection rule, and the (L)DD group
  (see \cref{thm:F1.general,lem:QED.Hyb.general,tab:tags,tab:tags-QED}).
  These results quantify how performance depends on specific design choices.

  \item \textit{An exact criterion for hybrid-protocol advantage under perfect recovery.}
  In the idealized perfect-recovery setting ($\pQEC=0$),
  we give a necessary and sufficient condition for when LDD+QEC outperforms QEC-only
  (\cref{thm:Hyb.vs.QEC}):
  the hybrid protocol helps precisely when the fraction of uncorrectable errors
  is larger among errors suppressed by LDD than among errors it leaves unsuppressed.
  Intuitively, LDD must target the errors that survive decoding.

  \item \textit{A low-noise criterion for hybrid-protocol advantage.}
  In the limit $p\to 0$, we prove a simple sufficient condition guaranteeing that
  LDD+QEC outperforms QEC-only for all sufficiently small $p$ (for any fixed $\pDD<1$):
  it suffices that LDD suppresses at least one minimum-weight uncorrectable Pauli error
  (\cref{thm:Hyb.vs.QEC-asymptotics}).
  Moreover, when the sufficient condition fails
  (i.e., when $\beta>\alpha$, where $\alpha$ is the minimum weight of an uncorrectable error and $\beta$ is the minimum weight of a suppressed uncorrectable error),
  but there exists a minimum-weight uncorrectable error
  with nontrivial syndrome (a condition that holds, for example, whenever $d\ge 2$),
  dressing an LDD generator by a stabilizer can modify the suppressed sector
  (without changing the logical action)
  so as to enforce $\beta'=\alpha$ and thereby recover the sufficient-condition regime
  (\cref{thm:Hyb.vs.QEC-asymptotics}, part~2).

  \item \textit{Robustness to imperfect recovery.}
  If there is a strict hybrid-protocol advantage at $\pQEC=0$,
  then the advantage persists for sufficiently small $\pQEC>0$ by continuity
  (\cref{cor:Hyb.vs.QEC}).

  \item \textit{Extension to error detection.}
  We derive analogous WEP-based expressions for the QED setting,
  track both conditional fidelity and acceptance probability,
  introduce a natural partial order on pairs $(F,P_A)$,
  and obtain an exact comparison criterion for LDD+QED versus QED-only under ideal readout
  (\cref{ss:QED-partial-order,thm:QED.Hyb.vs.QED,tab:tags-QED}).

  \item \textit{Numerical corroboration.}
  We perform numerical case studies for the $[[7,1,3]]$, $[[5,1,3]]$, and $[[13,1,3]]$ codes in the QEC setting,
  and the $[[4,2,2]]$ code in the QED setting,
  mapping regions of advantage in the $(\pDD,\pQEC)$ or $(\pDD,\pQED)$ plane beyond the low-$p$ theorem regimes,
  and illustrating that the hybrid protocol's benefit depends on whether the chosen LDD group suppresses the error classes that dominate logical failure for the chosen recovery or postselection rule
  (\cref{sec:numerics}).
\end{enumerate}

The main take-home message of this work is the following:
DD and QEC/QED are ``better together'' when logical decoupling
is chosen to suppress the Pauli errors that dominate the relevant logical-failure mechanism for the chosen recovery or postselection rule,
and when the physical error rate is low enough for the encoded strategy to realize its asymptotic advantage.
Our WEP formulas and comparison theorems make this criterion explicit and checkable.

The remainder of the manuscript is organized as follows.
In \cref{sec:setup} we present the common model and define the correction and detection protocols we analyze.
In \cref{sec:theorem-2-ec} we give the main results for the error-correction setting.
In \cref{sec:QED} we give the corresponding expressions and comparison criteria for the error-detection setting.
In \cref{sec:numerics} we numerically evaluate the formulas for representative QEC and QED codes
and illustrate the resulting regions of advantage in the $(\pDD,\pQEC)$ and $(\pDD,\pQED)$ planes.
\cref{sec:summary} summarizes the conclusions,
with technical derivations and the decoding maps used in the numerics collected in the appendices.

\section{Setup}
\label{sec:setup}

We assume the following common setup for stabilizer-based QEC/QED and (L)DD.

\subsection{Stabilizer-code preliminaries}
\begin{itemize}
  \item There is a natural projection $\pi\colon \mP_n \to \tmP_n$
    which removes the global phase of a Pauli operator.
    We call elements of $\tmP_n$ ``phase-stripped Paulis.''
    Equivalently, $\tmP_n \cong \mP_n/\langle iI\rangle$
    and we may represent each class by its unique Hermitian representative
    of the form $P_1\otimes \dots \otimes P_n$ with $P_j \in \{I, X, Y, Z\}$.
    In examples we write such Pauli strings by concatenating the letters from
    $\{\texttt{I}, \texttt{X}, \texttt{Y}, \texttt{Z}\}$
    (i.e., omitting $\otimes$), e.g., $\texttt{XXIZZY}$.
  \item Two elements of $\mP_n$ either commute or anticommute.
    If $E_1', E_2' \in \tmP_n$, we say that $E_1'$ and $E_2'$ commute (resp. anticommute)
    if their representatives $E_1, E_2$ (such that $E_j' = \pi(E_j)$)
    commute (resp. anticommute), i.e.,
    \begin{equation}
      \label{eq:commute}
      E_1E_2E_1^{-1}E_2^{-1} = I
    \end{equation}
    (resp. $=-I$), and write $[E_1', E_2'] = 0$ (resp. $[E_1', E_2'] \neq 0$).
    In other words, when we use the terms ``commute'' and ``anticommute,''
    we always refer to multiplication in $\mP_n$ and not in $\tmP_n$
    (otherwise these notions would be trivial because $\tmP_n$ is abelian).
    Note that this is well-defined because another choice of $E_1$, $E_2$
    would differ by phase factors which cancel in \cref{eq:commute}.
  \item $1 \leq k \leq n$.
  \item $\mS$ is a subgroup of $\mP_n$ with $2^{n-k}$ elements that does not contain $-I$.
  \item From the above it follows that $\mS$ is abelian
    and all elements of $\mS$ square to the identity.
    Indeed, any Pauli with phase $\pm i$ squares to $-I$,
    so $\mS$ contains only Hermitian Paulis (phases $\pm 1$).
    If $P,Q\in\mS$ anticommute, then $(PQ)^2=-I\in\mS$, a contradiction.
    Therefore, $\mS$ is an elementary abelian $2$-group
    and can be viewed as a vector space over $\mathbb{F}_2$.
  \item Let $\mS^*$ be the dual of $\mS$,
    i.e., the set of $\mathbb{F}_2$-linear functionals $\mS\to\mathbb{F}_2$.
    We interpret $\mS^*$ as the space of syndrome outcomes
    (stabilizer measurement results, in $\{0,1\}$ form).
    For every $E \in \mP_n$ we let $\syn(E) \in \mS^*$ be the corresponding syndrome
    given by $\syn(E)(S) = 1$ iff $E$ anticommutes with $S$.
    Since commutation/anticommutation does not depend on phase factors,
    $\syn(E)$ depends on $E$ only through $\pi(E)$,
    and we define $\syn(E')=\syn(E)$ for $E'=\pi(E)$.
  \item For the QEC analysis, let $D\colon \mS^* \to \tmP_n$ be any map satisfying $D(0) = \pi(I)$
    and $\forall \sigma \in \mS^*\; \syn(D(\sigma)) = \sigma$.
    This map can be interpreted as a decoding (recovery) map:
    it assigns to each syndrome $\sigma \in \mS^*$
    a chosen Pauli recovery $D(\sigma)$ with that syndrome. In the QED analysis of \cref{sec:QED}, no such recovery is applied;
    instead we postselect on a trivial reported syndrome.
  \item Let $C(\mS)\subset\mP_n$ denote the (Pauli) centralizer of $\mS$,
    i.e., the set of all $P\in\mP_n$ commuting with all elements of $\mS$.
    We choose a subgroup $\mL\subseteq C(\mS)$ such that
    \begin{equation}
      C(\mS)=\mS \mL,
      \quad
      \pmS\cap \pi(\mL)=\{\pi(I)\}.
    \end{equation}
    Equivalently, the standard quotient map $q\colon C(\mS) \to C(\mS) / \mS$
    restricts to an isomorphism $q|_{\mL} \colon \mL \xrightarrow{\sim} C(\mS) / \mS$.
    $\pi(\mL)\cong \pi(C(\mS))/\pmS$ is the logical Pauli group
    (phase-stripped, modulo stabilizers). Note that $\abs{C(\mS)} = 2^{n+k+2}$, $\abs{\mL} = 4^{k+1}$, and $\abs{\pi(\mL)} = 4^k$.
\end{itemize}

Many results we present below do not impose any additional restrictions on the code.
For example, they are applicable even to codes with $d=1$
which fail to correct some distance-1 errors.
Some results require additional assumptions on the code,
in which case we state those assumptions explicitly.

\subsection{(L)DD}

We assume that $\tG \subset \tmP_n$ is a nontrivial decoupling group,
and we define $G=\pi^{-1}(\tG)$. While we operate at the level of a phenomenological
effective model and do not prescribe any specific implementation of the DD sequence,
one can think of DD as a sequence of Pauli pulses $(\hat g_1, \hat g_1^{-1}\hat g_2, \hat g_2^{-1}\hat g_3, \dots, \hat g_{N-2}^{-1}\hat g_{N-1}, \hat g_{N-1}^{-1})$ with $\hat g_j \in G$, interleaved with free-evolution segments,
such that the total time spent in a free-evolution segment corresponding to a given phase-stripped element $g' \in \tG$
is equal for all $g' \in \tG$. Since the pulses are representatives in $G\subset\mP_n$,
they are Pauli strings and can be implemented by simultaneous single-qubit Pauli pulses.

For $E \in \mP_n$ we say that $E$ anticommutes with $G$
if and only if (iff) $E$ anticommutes with at least one element of $G$.
In that case (since $\tG$ is an $\mathbb{F}_2$-vector space under multiplication),
$E$ anticommutes with exactly half of the elements of $G$.
Similarly, we say that $E \in \mP_n$ commutes with $G$ if it commutes with all its elements.

We define the set of phase-stripped Paulis commuting with $\tG$ (equivalently, with $G$)
as the symplectic orthogonal complement
\begin{equation}
  \label{eq:Gperp}
  \tGp \coloneqq \bigl\{E' \in \tmP_n \big| [E',g']=0 \ \ \forall g'\in \tG \bigr\}.
\end{equation}

The errors in \cref{eq:Gperp} are unsuppressed by DD using $\tG$ in our model, whereas errors in the complement $\tmP_n\setminus \tGp$ are suppressed. This corresponds to the standard DD averaging intuition: repeated conjugation by the DD group averages the relevant error generator toward the commutant of the group, at least to leading order. This motivates treating $\tGp$ as the DD-unsuppressed sector and $\tmP_n\setminus\tGp$ as the DD-suppressed sector.

\subsection{Protocol and error model}
\label{ss:protocol-and-error-model}
Throughout, we assume the following single-cycle protocol (a single memory cycle):
\begin{enumerate}
  \item \label{it:thm2-s1} A state is encoded into either the trivial code (for DD-phys) or an $[[n,k,d]]$ stabilizer code.
  \item \label{it:thm2-s2} The encoded state is stored for a finite time---the ``wait interval'' (or ``memory interval'')---during which physical noise acts and we may or may not apply (L)DD.
  \item \label{it:thm2-s3} At the end of the cycle, stabilizer syndrome information is processed, either by applying recovery (QEC) or by postselection on a trivial syndrome (QED).
\end{enumerate}

The ``wait interval'' should be understood as the effective noise-accumulation window between syndrome-processing steps: in a bare memory experiment it is simply the storage time, while in a periodically corrected/detected memory it represents an interval between consecutive rounds of syndrome extraction. Its duration sets the physical error scale $p$. This interval need not be interpreted as literal idling only; rather, it is the portion of the cycle whose accumulated effect we model as an effective Pauli channel, possibly modified by DD/LDD. By contrast, imperfections of the final syndrome-processing stage itself are modeled separately by the phenomenological parameters $\pQEC$ and $\pQED$.

We assume that the encoding is perfect, so the only imperfections are those accumulated during the wait interval and those occurring in the syndrome-processing step.
More specifically, we assume a simplified Pauli error model in which the state $\rho$ transforms as
\begin{equation}
  \rho\ \longmapsto\ \sum_{E'\in \tmP_n} \Pr(E')  E' \rho (E')^{\dagger},
\end{equation}
where $\Pr(\bullet)$ is a probability distribution on $\tmP_n$.
The right-hand side $E' \rho (E')^{\dagger}$ is defined as $E \rho E^\dagger$ for
any representative $E\in\mP_n$ such that $\pi(E)=E'$.
This is unambiguous because replacing $E$ with $e^{i\phi} E$
leaves $E \rho E^\dagger$ unchanged because the global phase cancels.

We introduce an effective ``DD inverse-strength'' parameter $\pDD\in[0,1]$, which scales the probabilities of DD-suppressed errors to capture the residual error strength after applying DD.
We assume that the DD sequence does not introduce additional error mechanisms beyond those already absorbed into this effective $\pDD$; this is best interpreted as either using pulse-robust DD sequences~\cite{Quiroz2013PRA,Genov2017PRL} or having independently characterized control imperfections and folding them into the fitted value of $\pDD$. The case
$\pDD=1$ corresponds to the absence of DD or to situations where DD does not result in any reshaping of the effective Pauli distribution (e.g., genuinely Markovian noise).
The endpoint $\pDD=0$ should be understood as the idealized limit in which, within our effective model, the residual channel assigns zero weight to Pauli errors that anticommute with $\tG$ (equivalently, all remaining residual noise lies in $\tGp$). We use this only as a formal benchmark, not as implying error-free evolution; in any realistic implementation one expects $1 > \pDD > 0$ due to finite pulse widths, calibration/control errors, and finite-bandwidth filtering~\cite{Suter:2016aa,CywinskiLutchynNaveDasSarma:2008aa,GreenSastrawanUysBiercuk:2013aa}.
In practice, small values of $\pDD$ would correspond to sufficiently fast, regular, and accurate DD applied to slowly varying non-Markovian noise.
The description we present here is a phenomenological model for the combined effect of DD during the whole wait interval; we do not attempt to derive it from any
specific Hamiltonian, Lindbladian, pulse schedule, or Magnus-expansion order.

With this in mind, we assume the following error model for the wait interval.
In the absence of DD (equivalently, $\pDD=1$, i.e., no suppression via DD),
the probability of error $E' \in \tmP_n$ is
\begin{equation}
\label{eq:error-model}
\Pr(E') = (1-p)^{n-w}\Bigl(\frac{p}{3}\Bigr)^{w},
\end{equation}
where $w = \wt(E')$ is the number of non-identity single-qubit factors in the canonical representative of $E'$.
When DD is applied, the probabilities of errors $E'$ that anticommute with $\tG$ (i.e., $E'\notin \tGp$) are multiplied by $\pDD$, and then all probabilities are renormalized so that they still sum to $1$.
We note that there are alternatives to this model of DD suppression, and discuss this in detail in \cref{app:alt-DD-model}.

After the wait interval, the QEC and QED branches differ only in how the syndrome information is used.

For the QEC setting, we assume that with probability $1 - \pQEC$ error correction applies the recovery $D(\sigma)$ corresponding to the measured syndrome $\sigma$, and with probability $\pQEC$ it applies a uniformly random Pauli recovery consistent with that syndrome, but only when a nontrivial syndrome is measured; if the measured syndrome is trivial, the state is left unchanged.
Conditioned on a fixed nonzero syndrome, any two syndrome-consistent recoveries differ by a zero-syndrome operator, i.e., by an element of $\pi(C(\mS))$, and hence (modulo $\pmS$) by a logical Pauli in $\pi(C(\mS))/\pmS\cong\pi(\mL)$.
Therefore, under our assumption that the recovery in the case of decoder failure
is chosen uniformly at random among syndrome-consistent recoveries,
the induced logical Pauli is uniform over the $4^k$ logical classes
and the probability of the trivial logical class is $4^{-k}$.
Thus, after this step the state lies in the codespace
and the remaining effective errors are logical Paulis.
We call this whole procedure
(regardless of whether $D(\sigma)$ or a random recovery was applied) a decoder.%
\footnote{We remark that in experimental QEC work,
  ``decoder'' often denotes the classical post-processing algorithm
  used to infer the logical correction from syndrome data under an assumed noise model.
  In contrast, in this work ``decoder'' means the effective {single-round}
  syndrome-to-recovery rule (equivalently, a Pauli-frame update rule)
  together with the associated recovery operation.}

For the QED setting, no Pauli recovery is applied.
Instead, the run is accepted iff the reported syndrome is trivial.
With probability $1-\pQED$ the syndrome is reported correctly,
and with probability $\pQED$ the reported outcome is completely random, uniformly distributed over the $2^{n-k}$ syndromes.
Hence a detected (nonzero-syndrome) error can be accepted only through readout failure.
When such a detected error is nevertheless accepted, we model the resulting logical action on the codespace as a uniformly random logical Pauli;
equivalently, the accepted output is maximally mixed on the $k$-qubit logical subsystem, so the trivial logical class occurs with probability $4^{-k}$.
This modeling assumption is phenomenological:
even though the actual state is different from a maximally mixed state
(indeed, it is outside of the code space),
it may have the same impact as a maximally
mixed state on the probability distribution
of an eventual outcome of an experiment in which it is used.

Under this model,
$\pDD=0$ corresponds to the idealized limit in which the effective residual channel assigns zero weight to all Pauli errors that anticommute with $\tG$ (equivalently, all remaining residual noise lies in $\tGp$),
$\pQEC=0$ corresponds to ideal recovery given the measured syndrome,
and $\pQED=0$ corresponds to ideal syndrome readout and postselection.
We treat the case when a QEC code is used but its recovery step is omitted as $\pQEC = 1$,
i.e., whenever a nontrivial syndrome is detected we apply
a uniformly random syndrome-consistent recovery.
Conditioned on detecting a nontrivial syndrome,
such a random recovery yields the correct logical frame with probability $4^{-k}$
(and hence induces a logical error with probability $1-4^{-k}$).
The analogous no-postselection LDD-only baseline in the QED setting is defined separately in \cref{sec:QED};
it is not obtained by setting $\pQED=1$, which would instead correspond to completely random syndrome readout.

\subsection{Faults \textit{vs} errors}

We distinguish between \emph{faults}, which are stochastic malfunction events in the implementation,
and \emph{errors}, which denote the resulting effective operators acting on the data.
In the effective-Pauli model used throughout this work,
the cumulative effect of all faults during the memory (wait) interval
is represented by a single random phase-stripped Pauli error $E'\in\tmP_n$;
its {weight} is $\wt(E')$, the number of qubits on which it acts nontrivially.

In the QEC setting, imperfections in syndrome extraction and/or recovery
are modeled at the event level by a recovery fault:
when a nontrivial syndrome is measured,
the ideal recovery $D(\sigma)$ is applied with probability $1-\pQEC$,
while with probability $\pQEC$ a uniformly random syndrome-consistent Pauli recovery is applied
(and if the measured syndrome is trivial, the state is left unchanged).

In the QED setting, imperfections in the accept/reject decision are modeled at the event level by a readout fault:
the true syndrome is reported correctly with probability $1-\pQED$ and is otherwise replaced by a uniformly random syndrome.
A readout fault can therefore cause a detected error to be accepted.

Given a fixed recovery map $D$ (QEC setting), we call a Pauli error $E'$ correctable if $D(\syn(E'))E'\in\pmS$, i.e., if ideal decoding returns the state to the codespace without inducing a nontrivial logical action.
Otherwise $E'$ is uncorrectable.

In the QED setting, rather than correctable versus uncorrectable, the natural distinction is between
undetected errors $E'\in\pmSL$ (zero syndrome) and detected errors $E'\notin\pmSL$ (nonzero syndrome).
Among the undetected errors, stabilizers $\pmS$ act trivially on the codespace,
while $\pmSL\setminus\pmS$ are nontrivial logical errors on accepted runs.

A logical error (or logical fault) refers to the induced action on the encoded subspace
(equivalently, the coset of the residual operator modulo the stabilizer,
an element of $\pi(C(\mS))/\pmS \simeq \pi(\mL)$).
A logical failure event occurs when the net operator after the full protocol has a nontrivial logical component;
in the QED setting this definition is understood conditionally on acceptance.

\subsection{Strategies}
We study two related four-strategy comparison problems.

In the QEC setting, we compare:
\begin{enumerate}
  \item \textit{DD-phys}: DD on $k$ physical qubits (no encoding or recovery);
  \item \textit{QEC-only}: QEC recovery alone (without LDD);
  \item \textit{LDD-only}: LDD alone on an encoded block, with the syndrome measured but deterministic recovery omitted (equivalently, $\pQEC=1$ as above);
  \item \textit{LDD+QEC}: Hybrid LDD followed by QEC recovery.
\end{enumerate}

In the QED setting, we compare:
\begin{enumerate}
  \item \textit{DD-phys}: DD on $k$ physical qubits (no encoding or postselection);
  \item \textit{QED-only}: syndrome measurement with postselection, but no LDD;
  \item \textit{LDD-only}: LDD on the encoded block, but without postselection (so $P_A=1$ by definition; the precise convention is given in \cref{sec:QED});
  \item \textit{LDD+QED}: Hybrid LDD followed by syndrome measurement and postselection.
\end{enumerate}

While LDD typically refers to DD with $\tG = \pi(\mL)$,
most arguments presented in this work do not depend on any specific choice of $\tG$
and apply equally well to any nontrivial subgroup $\tG \subset \tmP_n$.

\emph{A major goal of this work is to establish, in both the QEC and QED settings, the ordering of the four strategies.} The comparison is made assuming fixed values of relevant parameters: the number of encoded qubits $k$ and per-qubit physical error probability $p$ are the same for all four strategies, the code is fixed for strategies other than DD-phys, $\pQEC$ and $\pQED$ are fixed for strategies involving QEC or QED, and $\pDD$ and the size of the DD group are fixed for strategies involving DD. This is done to isolate the error-suppression question from a full overhead optimization problem. A realistic architecture-level cost analysis of LDD implementation is an important direction for future work.

We note that another commonly used strategy
is to directly embed DD-phys in the idle gaps of a quantum circuit.
This approach, which has been successfully used to boost algorithmic performance~%
\cite{pokharel2022demonstration,Mundada:2023aa,Baumer2024,singkanipa2025demonstration}
and QEC fidelities~\cite{Acharya:2025aa,Bluvstein:2026aa,vezvaee2025surfacecodescalingheavyhex},
has been rigorously analyzed
in the context of a hybrid DD fault-tolerant-QEC strategy~\cite{Ng:2011dn}. We do not address it here.

\subsection{Performance metrics}

In the QEC setting, after physical noise during the wait interval, (L)DD, and (optionally) QEC,
the effective CPTP map on the codespace $\mC (\simeq \mathbb C^{2^{k}})$ can be written as
\begin{equation}
\label{eq:logical-Pauli-channel}
\Lambda(\rho)=\sum_{L\in\pi(C(\mS)) / \pmS} p_{L} L\rho L^\dagger,
\quad \sum_{L\in\pi(C(\mS)) / \pmS} p_{L}=1.
\end{equation}
Here $\rho$ is a density matrix on $\mC$, $L$ is an equivalence class of logical operators defined up to phase and multiplication by stabilizers, and $p_{L}$ is the renormalized probability of the logical Pauli $L$, after any DD rescaling has been applied and the final probabilities have been rescaled to sum to one.
Note that $L\rho L^\dagger$ is well-defined.
That is, if $L_1,L_2 \in C(\mS)$ are two representatives of the same equivalence class $L$, then $L_1 \rho L_1^\dagger = L_2 \rho L_2^\dagger$.
Indeed, since $L_1, L_2$ belong to the same class, we have $L_2 = sSL_1$ where $s \in \mathbb{C}$, $\abs{s}=1$, and $S \in \mS$.
Since $\rho$ is supported on $\mC$, we have $S\rho = \rho S^\dagger = \rho$.
Since $L_1\in C(\mS)$, $S$ commutes with $L_1$.
Combining these observations, we obtain
\begin{equation}
  L_2 \rho L_2^\dagger = sSL_1 \rho L_1^\dagger S^\dagger s^* = L_1 \rho L_1^\dagger.
\end{equation}

We quantify the correction-setting performance using the same fidelity measure $F$:
\begin{align}
  F &\coloneqq \text{probability of no logical errors after the protocol} \notag \\
  &= p_{I}. \label{eq:F}
\end{align}
We show in \cref{app:ent-fid} that this is identical to the standard entanglement fidelity.
We denote the corresponding QEC fidelities as $\FDD$, $\FQEC$, $\FLDD$, $\FHyb$.

In the QED setting, the natural figure of merit is two-dimensional.
We track the acceptance probability $P_A$, i.e., the probability that the reported syndrome is trivial,
together with the logical fidelity $F$ conditioned on acceptance, i.e., the probability that an accepted run induces the trivial logical class.
We denote the conditional fidelities of QED-only and LDD+QED by $\FQED$ and $\FHybD$, respectively;
the corresponding acceptance probabilities are $P_{A,{\rm QED}}$ and $P_{A,{\rm Hyb}}$.
For DD-phys and LDD-only in the detection setting, $P_A=1$ by definition.

When comparing these quantities, we often focus on the limit $p \to 0$,
while keeping $\pDD > 0$ constant and taking either $\pQEC = 0$ or $\pQEC = o(1)$ in the correction setting,
or $\pQED = 0$ or $\pQED = o(1)$ in the detection setting.
This asymptotic regime isolates the leading dependence on the physical error rate $p$ while treating $\pDD$ as a fixed ``DD floor'' parameter, which need not vanish as $p\to 0$.
We also focus on $\pQEC=0$ or $\pQEC=o(1)$, and on $\pQED=0$ or $\pQED=o(1)$, to capture the idealized settings in which recovery or syndrome readout/postselection become increasingly reliable in the same small-$p$ limit.
Finally, we focus on $\pDD>0$ to avoid the degenerate perfect-decoupling limit; the endpoint $\pDD=0$ corresponds to
complete suppression (followed by renormalization) of all errors outside $\tGp$.
In particular, for DD-phys with $\tG=\tmP_k$ (so that $\tGp=\{\pi(I)\}$), this limit yields $F=1$ for any $p<1$.

We first derive explicit expressions for the QEC fidelities and then compare them.
For this comparison, we assume that $k$ in all cases is the same
and that the QEC code (and hence $n$) used in QEC-only, LDD-only, and LDD+QEC is the same.
For DD-phys, the code is trivial: $n=k$, $\mS = \{I\}$, $\mL = \mP_k$.
To make the comparison fair, we assume that the order of the DD group is the same in all three cases where DD is used: $4^k$.
In the LDD-only case, since $n>k$ and we need to project the $n$-qubit state into the codespace, we assume that the syndrome is measured and (upon detecting a nontrivial syndrome) a uniformly random syndrome-consistent recovery is applied so that the output state lies in the codespace.
The QED comparison uses the same encoded block and DD-group size, with postselection replacing recovery as described in \cref{sec:QED}.

\section{Analysis of the four strategies}
\label{sec:theorem-2-ec}

\subsection{DD-phys}
\label{ss:DD-phys}
In DD-phys, unlike in the other three strategies, we set $n = k$.
In this case, we have $k$ physical qubits and apply DD with a group that contains $4^k$ elements.
This means that the DD group is $\tG = \tmP_k$ and $\tGp = \{\pi(I)\}$.

\begin{mylemma}
  With no DD (i.e., $\pDD=1$), the probability of no error is $\FDD = (1-p)^k$.
With DD, the probability of no error is
  \begin{subequations}
  \begin{align}
    \label{eq:F2.explicit}
    \FDD &= \frac{(1-p)^{k}}{(1-p)^{k}+\pDD\bigl[1-(1-p)^{k}\bigr]} \\
    \label{eq:F2.explicit-2}
    &= 1 - k p \pDD + O(p^2)
  \end{align}
  \end{subequations}
  as $p \to 0$ (uniformly in $\pDD$).
\end{mylemma}

\begin{proof}
  Recall that without DD, we have \cref{eq:error-model}.
The probability of no error ($w=0$) is then $\FDD = \Pr[\pi(I)] = (1-p)^{k}$.

  In the DD-phys case, all Pauli errors not in $\tGp=\{\pi(I)\}$ anticommute with at least one element of $\tG$ and therefore have their probabilities multiplied by $\pDD\in[0,1]$.
Thus, the renormalized probability of no error, i.e., the fidelity $\FDD$, is
  \begin{equation}
  \FDD = \frac{\Pr[\pi(I)]}{\Pr[\pi(I)]  +  \pDD \bigl(1-\Pr[\pi(I)]\bigr)} ,
  \end{equation}
  which is \cref{eq:F2.explicit}.

  Factoring $(1-p)^{k}$ out of the denominator yields $\FDD = [1 + \pDD ((1-p)^{-k} - 1)]^{-1}$.
For $p\ll 1$ we have $(1-p)^{-k}=1+kp+O(p^{2})$.
Hence,
  \begin{equation}
  \FDD = \bigl[1+k  p  \pDD+O(p^{2})\bigr]^{-1}
  = 1-k  p  \pDD+O(p^{2}),
  \end{equation}
  which is \cref{eq:F2.explicit-2}.
\end{proof}

\subsection{Weight enumerators}

To derive the fidelities in the other three cases,
we need to correctly aggregate probabilities over the different classes of physical errors
(e.g., stabilizers vs. logical operators, suppressed vs. unsuppressed, correctable vs. uncorrectable).
We do this using weight enumerator polynomials~\cite{PhysRevLett.78.1600,Rains:1999vp}.
For any subset $A \subset \tmP_n$, its WEP is defined as
\begin{equation}
  \label{eq:W}
  W(A; z) \coloneqq \sum_{E'\in A} z^{\wt(E')}
  =\sum_{w=0}^{n} W_w(A) z^{w},
\end{equation}
where $W_w(A)$ is the number of elements of $A$ of weight $w$, i.e.,
$W_{w}(A)= \abs{\bigl\{E'\in A \,\big|\, \wt(E')=w\bigr\}}$.

For example, $\wt(\pi(I))=0$, hence $W(\{\pi(I)\}; z) = 1$; for a single qubit, $\wt(X)=\wt(Y)=\wt(Z)=1$, so $W(\tmP_1; z) = (1 + 3z)$; for $n$ qubits,
\begin{equation}
\label{eq:WPn}
W(\tmP_n; z) = \sum_{w=0}^n\binom{n}{w}(3z)^w = (1 + 3z)^n .
\end{equation}

In our depolarizing wait-interval model [\cref{eq:error-model}],
errors of the same weight have the same probability.
Writing
\begin{equation}
  (1-p)^{n-w}\Bigl(\frac{p}{3}\Bigr)^w = (1-p)^n \Bigl(\frac{p}{3-3p}\Bigr)^w,
\end{equation}
we see that for any subset $A\subset\tmP_n$,
\begin{equation}
  \sum_{E'\in A}\Pr(E') = (1-p)^n  W\left(A; z\right),
  \quad z \coloneqq \frac{p}{3-3p}.
\end{equation}
Hence, when forming conditional probabilities after DD renormalization,
it is sufficient to evaluate WEPs at $z=p/(3-3p)$; common prefactors cancel.

Similarly, $W\bigl(\tGp;z\bigr)$ is the WEP of errors unsuppressed by the decoupling group $\tG$, while $W\bigl(\tmP_n\setminus \tGp;z\bigr)$ is the WEP of errors suppressed (rescaled) by $\tG$.

To handle the protocols that use QEC, we introduce the set $\tmcEc$ of correctable physical errors, i.e., the errors that are corrected to the trivial logical class by the chosen decoding map $D$:
\begin{equation}
  \tmcEc \coloneqq \left\{E' \in \tmP_n \ \middle|\ D(\syn(E'))  E' \in \pmS\right\}.
\end{equation}
This set splits into (trivial) stabilizer errors $\pmS$ with zero syndrome, and nontrivial correctable errors $\tmcEc \setminus \pmS$ with nonzero syndrome.
Conversely, the set of uncorrectable errors is $\tmP_n \setminus \tmcEc$, which splits into
(i) nontrivial logical errors $\pmSL\setminus \pmS$ with zero syndrome, and
(ii) detectable but uncorrectable errors $\tmP_n \setminus \tmcEc \setminus \pmSL$ with nonzero syndrome.

Including DD, we obtain additional sets of physical significance.
For example, $(\tmcEc\cap\tGp)\setminus\pmS$ is the set of correctable errors that are unsuppressed and have nonzero syndrome, and $(\tmP_n\setminus\tGp) \setminus\pmSL$ is the set of suppressed errors that are detectable.
\cref{tab:tags} summarizes all the different error sets that appear in our fidelity expressions below.
It also introduces compact tag notation for the corresponding weight enumerators.

\begin{table}
  \footnotesize
  \begin{center}
    \begin{tabular}{|c|c|c|}
      \hline
      tag & weight enumerator & set description \\ \hline
      \all & $W\bigl(\tmP_n;z\bigr) = (1+3z)^n$ & all \\ \hline
      \St & $W\bigl(\pmS;z\bigr)$ & stabilizers \\ \hline
      \uSt & $W\bigl(\tmP_n\setminus \pmS;z\bigr)$ & non-stabilizers \\ \hline
      \symS & $W\bigl(\tmP_n\setminus \tGp;z\bigr)$ & suppressed \\ \hline
      \uS & $W\bigl(\tGp;z\bigr)$ & unsuppressed \\ \hline
      \C & $W\bigl(\tmcEc \setminus \pmS;z\bigr)$ & corrected (nontrivial syndrome) \\ \hline
      \uC & $W\bigl(\tmP_n\setminus \tmcEc;z\bigr)$ & uncorrected \\ \hline
      \D & $W\bigl(\tmP_n \setminus \pmSL;z\bigr)$ & detected (nonzero syndrome) \\ \hline
      \symL & $W\bigl(\pmSL \setminus \pmS;z\bigr)$ & nontrivial logical \\ \hline
      \uSSt & $W\bigl(\pmS \cap \tGp; z\bigr)$ & unsuppressed stabilizers \\ \hline
      \SSt & $W\bigl(\pmS \setminus \tGp; z\bigr)$ & suppressed stabilizers \\ \hline
      \uSuC & $ W\bigl(\tGp\setminus\tmcEc;z\bigr) $ &
      unsuppressed uncorrected \\ \hline
      \uSC & $ W\bigl((\tmcEc\cap\tGp)\setminus\pmS;z\bigr) $ &
      unsuppressed corrected \\ \hline
      \SuC  & $ W\bigl(\tmP_n\setminus\tGp\setminus\tmcEc;z\bigr) $ &
      suppressed uncorrected \\ \hline
      \SC  & $ W\bigl(\tmcEc\setminus\pmS\setminus\tGp;z\bigr) $ &
      suppressed corrected \\ \hline
      \uSD & $ W\bigl(\tGp\setminus\pmSL;z\bigr) $ &
      unsuppressed detected \\ \hline
      \SD  & $ W\bigl(\tmP_n\setminus\tGp\setminus\pmSL;z\bigr) $ &
      suppressed detected \\ \hline
      \uCD & $W\bigl(\tmP_n \setminus \tmcEc \setminus \pmSL;z\bigr)$ & uncorrected detected \\ \hline
      \SuSt & $W\bigl(\tmP_n \setminus\tGp \setminus \pmS; z\bigr)$ & suppressed non-stabilizers \\ \hline
      \uSuSt & $W\bigl(\tGp \setminus \pmS; z\bigr)$ & unsuppressed non-stabilizers \\ \hline
    \end{tabular}
  \end{center}
  \caption{List of all WEPs that appear in the fidelity expressions.
Each error set (middle column) is described in terms of its relation to DD and QEC (right column) and is assigned a corresponding acronym tag (left column) aligned with the set description.
We always evaluate each WEP at $z=p/(3-3p)$.}
  \label{tab:tags}
\end{table}

We proceed to reexpress the DD-phys fidelity in terms of WEPs.

\subsection{WEP expression for DD-phys}

\begin{mylemma}
  \label{lem:F2.general}
  In the DD-phys strategy, the fidelity can be written as
  \begin{equation}
    \label{eq:F2.general}
    \FDD
    =\frac{\uSSt + \pDD\SSt}{\uS+\pDD\symS},
  \end{equation}
  and for $\tG=\tmP_k$ (so that $\tGp=\{\pi(I)\}$) this reduces to \cref{eq:F2.explicit}.
\end{mylemma}

\begin{proof}
  For any $E'\in\tmP_n$ of weight $w=\wt(E')$, \cref{eq:error-model} can be rewritten as
  \begin{equation}
    \Pr(E') = (1-p)^{n} \Bigl(\frac{p}{3-3p}\Bigr)^{w}.
  \end{equation}
  Thus, evaluating all WEPs at $z=\frac{p}{3-3p}$,
  the common prefactor $(1-p)^{n}$ cancels between numerator and denominator after DD renormalization.

  In our DD model, errors in $\tGp$ are left unchanged
  (they are unsuppressed, with WEP denoted $\uS$ in \cref{tab:tags}),
  while errors in $\tmP_n \setminus \tGp$ have their probabilities rescaled by a factor $\pDD$ (suppressed, denoted $\symS$).
  Therefore, the total unnormalized probability mass after rescaling is proportional to
  \begin{equation}
  \label{eq:phys-DD-denom-again}
    W \bigl(\tGp;z\bigr) + \pDD  W \bigl(\tmP_n \setminus \tGp;z\bigr)
    = \uS+\pDD\symS,
  \end{equation}
  which gives the fidelity (normalization) denominator.

  The ``no-logical-error'' events are precisely stabilizer errors, i.e., elements of $\pmS$.
  These split into an unsuppressed subset $\pmS\cap\tGp$ (unsuppressed stabilizer, denoted $\uSSt$) and a suppressed subset $\pmS\setminus\tGp$ (suppressed stabilizer, denoted $\SSt$).
  Hence the corresponding unnormalized probability mass is proportional to
  \begin{equation}
    W\bigl(\pmS \cap \tGp; z\bigr) + \pDD  W\bigl(\pmS \setminus \tGp; z\bigr)
    = \uSSt + \pDD\SSt,
  \end{equation}
  which gives the fidelity numerator.
  Together, this yields \cref{eq:F2.general}.

  To recover \cref{eq:F2.explicit}, note that in the DD-phys case $n=k$, $\pmS=\{\pi(I)\}$, and for a $4^k$-element DD group we have $\tG=\tmP_k$ and thus $\tGp=\{\pi(I)\}$.
  Then $W(\{\pi(I)\};z)=1$, $W(\tmP_{k} \setminus \{\pi(I)\};z)=(1+3z)^{k}-1$, and $W(\pmS\setminus\tGp;z)=0$.
  Substituting into \cref{eq:F2.general} gives
  \begin{equation}
    \FDD = \frac{1}{1+\pDD\bigl[(1+3z)^{k}-1\bigr]},
  \end{equation}
  which becomes \cref{eq:F2.explicit} after setting $z=p/(3-3p)$ and simplifying.
\end{proof}

\subsection{LDD+QEC}
\label{ss:LDD+QEC}

We are now prepared to state our first main technical result.

\begin{mytheorem}[Hybrid LDD+QEC: infidelity under an effective Pauli noise model with DD suppression and decoder failure]
  \label{thm:F1.general}
  Consider an $[[n,k,d]]$ stabilizer code with stabilizer group $\mS$ and logical Pauli group $\mL$
  (so $\pmS\subseteq \pmSL\subseteq \tmP_n$).
  Fix a (logical) dynamical-decoupling group $\tG$
  and let $\tGp\subseteq \tmP_n$ denote the set of Pauli errors
  that are not suppressed by the chosen (L)DD procedure
  (with suppressed sector $\tmP_n\setminus \tGp$).

  Assume the following effective Pauli noise model.
  Writing $z \coloneqq p/(3-3p)$,
  each Pauli $E'\in\tmP_n$ of weight $\wt(E')$ contributes a factor $z^{\wt(E')}$.
  In addition, if $E'$ lies in the suppressed sector ($E'\notin \tGp$), its contribution is rescaled by a multiplicative
  suppression factor $\pDD$ (with $\pDD=1$ meaning no DD suppression).

  After the (L)DD step, we perform one round of QEC:
  we measure the stabilizer syndrome $\sigma$ and apply a recovery.
  If the measured syndrome is trivial, we apply no recovery.
  If the measured syndrome is nontrivial, then with probability $1-\pQEC$ we apply a fixed (ideal) recovery $D(\sigma)$,
  while with probability $\pQEC$ we instead apply a uniformly random Pauli recovery consistent with the measured syndrome.
  Under this decoder-failure model, conditioned on any fixed nonzero syndrome, the induced logical Pauli is uniform over the $4^k$
  logical classes, so the probability of the trivial logical class is $4^{-k}$.

  Let $\tmcEc\subseteq\tmP_n$ denote the set of Pauli errors that are corrected to the stabilizer (no logical fault) by the
  ideal decoder, i.e., those $E'$ such that $D(\syn(E'))E'\in \pmS$.

  Then the entanglement infidelity of the resulting logical channel under the hybrid LDD+QEC protocol is
  \begin{align}
    \label{eq:1-F1.general}
    1 - \FHyb = \frac{1}{\uS+\pDD\symS}\bigl(&[\uSuC + \pQEC (\uSC -4^{-k}\uSD)] \notag\\
      +\pDD&[\SuC + \pQEC(\SC-4^{-k}\SD)]\bigr)  .
  \end{align}
  All WEP tags are as in \cref{tab:tags}.
\end{mytheorem}

This result includes DD-phys, QEC-only, and LDD-only as special cases.
We give independent derivations of the latter two in \cref{ss:QEC-only,ss:LDD-only};
these derivations go into more detail for each specific case than the unified proof we give here.

\begin{proof}
Just as in the DD-phys case [\cref{eq:phys-DD-denom-again}], after DD rescaling, the (unnormalized) total probability weight is proportional to
  \begin{equation}
    W(\tGp;z) + \pDD W(\tmP_n\setminus \tGp;z) = \uS+\pDD\symS,
  \end{equation}
which yields the denominator in \cref{eq:1-F1.general} after renormalization.
However, the sets involved are different from those in the DD-phys case, since now $n>k$.

For the numerator, we count (via WEPs) the total weight of physical Pauli errors $E'$ that produce a {nontrivial} residual logical Pauli after the protocol.
This gives the infidelity (rather than the fidelity) in \cref{eq:1-F1.general}.

Recall that in our model, decoder failure occurs only when the measured syndrome is nontrivial.
Conditioned on any fixed nonzero syndrome, the (failed) decoder applies a uniformly random syndrome-consistent Pauli recovery, so the induced logical Pauli is uniform over the $4^k$ logical classes.
Hence, under decoder failure, the probability of landing in the trivial logical class is $4^{-k}$.
With this in mind, we now split errors into the two sectors used throughout, namely, unsuppressed vs. suppressed by LDD.

\begin{enumerate}
  \item \textit{Unsuppressed sector} ($E'\in\tGp$; no factor of $\pDD$).

  Using the notation of \cref{tab:tags},
  $\uSuC$ = WEP of unsuppressed {uncorrectable} errors;
  $\uSC$ = WEP of unsuppressed {correctable detected} errors (nonzero syndrome);
  $\uSD$ = WEP of {all} unsuppressed {detected} errors (nonzero syndrome),
  so in particular $\uSC\subset\uSD$.
  We now do a three-step inclusion-exclusion count:

  \begin{enumerate}
    \item[(i)] {Baseline: count all uncorrectable errors as logical faults.}
    Any error counted by $\uSuC$ produces a nontrivial logical Pauli under ideal decoding (or is an undetected logical operator, for which no recovery is ever applied), so we start with a baseline contribution $\uSuC$.

    \item[(ii)] {Add: correctable detected errors matter only if the decoder fails.}
    Errors counted by $\uSC$ produce no logical fault under ideal decoding, but when the decoder fails (probability $\pQEC$) we initially count them as contributing a logical fault.
    This gives an additive term $\pQEC \uSC$.

    \item[(iii)] {Subtract: when the decoder fails, a detected error is accidentally harmless with probability $4^{-k}$.}
    Under decoder failure, for any detected error (correctable or uncorrectable), the resulting logical Pauli is uniform, so with probability $4^{-k}$ we land in the trivial logical class and should not count a logical fault.
    Since $\uSD$ counts {all} detected errors in the unsuppressed sector, we subtract the overcount $\pQEC 4^{-k}\uSD$.
  \end{enumerate}

  Summing these three contributions, the unsuppressed-sector logical-fault weight is
$\uSuC + \pQEC\bigl(\uSC-4^{-k}\uSD\bigr)$.

  \item \textit{Suppressed sector} ($E'\in\tmP_n\setminus\tGp$; overall factor of $\pDD$).

  The identical bookkeeping applies with the suppressed-sector \cref{tab:tags} quantities $\SuC,\SC,\SD$,
  and we multiply by $\pDD$ because these errors are suppressed by LDD.
  Hence this sector contributes
    $\pDD\Bigl[\SuC  +  \pQEC\bigl(\SC-4^{-k}\SD\bigr)\Bigr]$.
\end{enumerate}

Summing the two sector contributions yields the stated numerator, and dividing by the normalization $\uS+\pDD\symS$ yields \cref{eq:1-F1.general}.
\end{proof}

\subsection{Reductions}
We now reduce LDD+QEC to QEC-only, LDD-only, and DD-phys.

\subsubsection{Reduction to QEC-only}

LDD+QEC reduces to QEC-only in the limit $\pDD=1$, when DD produces no rescaling.

Because every Pauli either lies in $\tGp$ or in its complement, each relevant set splits into unsuppressed and suppressed parts.
In particular,
$\uS+\symS=\all$,
$\uSuC+\SuC=\uC$,
$\uSC+\SC=\C$, and $\uSD+\SD=\D$.
Using \cref{eq:1-F1.general}, we therefore have
\begin{subequations}
  \label{eq:1-F4.general.b}
\begin{align}
  1-\FQEC &= 1 - \FHyb \bigr|_{\pDD=1} \\
  &= \frac{1}{\uS+\symS}\Bigl([\uSuC + \pQEC (\uSC -4^{-k}\uSD)]\notag\\
  &\quad\quad +[\SuC + \pQEC(\SC-4^{-k}\SD)]\Bigr) \\
  &= \frac{1}{\all}\Bigl(\uC+\pQEC\bigl(\C -4^{-k}\D\bigr)\Bigr).
\end{align}
\end{subequations}

\subsubsection{Reduction to LDD-only}
LDD+QEC reduces to LDD-only in the limit $\pQEC=1$ (decoder always applies a uniformly random syndrome-consistent recovery when a nontrivial syndrome is detected).

We first note that in each DD sector, the non-stabilizer errors split into the disjoint sets of those decoded to the trivial logical class ($\tmcEc\setminus\pmS$) and those not ($\tmP_n\setminus\tmcEc$), which yields the following identities.

Since
\begin{subequations}
\begin{align}
&(\tGp\setminus\tmcEc) \ \dot\cup\ \bigl[(\tmcEc\cap\tGp)\setminus\pmS\bigr]
=\tGp\setminus\pmS\\
&\quad\Rightarrow\quad \uSuC+\uSC=\uSuSt,
\end{align}
\end{subequations}
and
\begin{subequations}
\begin{align}
&\bigl((\tmP_n\setminus\tGp)\setminus\tmcEc\bigr) \ \dot\cup\ \bigl[(\tmcEc\setminus\pmS)\setminus\tGp\bigr]
=(\tmP_n\setminus\tGp)\setminus \pmS\\
&\quad\Rightarrow\quad \SuC + \SC = \SuSt,
\end{align}
\end{subequations}
we have, using \cref{eq:1-F1.general}:
\begin{subequations}
\label{eq:1-F3.general}
\begin{align}
  1-\FLDD &=1 - \FHyb \bigr|_{\pQEC=1} \\
  &= \frac{1}{\uS+\pDD\symS}\Bigl([\uSuC + \uSC -4^{-k}\uSD]\notag\\
  &\quad\quad+\pDD[\SuC + \SC-4^{-k}\SD]\Bigr)\\
  &= \frac{1}{\uS+\pDD\symS}\Bigl((\uSuSt -4^{-k}\uSD)\notag\\
  &\quad\quad+\pDD(\SuSt-4^{-k}\SD)\Bigr) .
\end{align}
\end{subequations}

\subsubsection{Reduction to DD-phys}
LDD-only reduces to DD-phys when there is no encoding/decoding ($n=k$ and $\pmS=\{\pi(I)\}$), and when the DD group has order $4^k$ (so $\tG=\tmP_k$ and $\tGp=\{\pi(I)\}$).
In this case $\pmSL=\tmP_k$, so there are no detected errors ($\D=\uSD=\SD=0$), and $\tmcEc=\pmS=\{\pi(I)\}$, so there are no correctable non-stabilizer errors ($\C=\uSC=\SC=0$).
All non-identity Paulis are uncorrectable, and all are suppressed:
$\uC=\uSt=(1+3z)^n-1$,
$\SuC=\SuSt=(1+3z)^n-1$,
$\uSuC=\uSuSt=0$.
Therefore, \cref{eq:1-F3.general} reduces to
\begin{subequations}
\label{eq:F2.general-again}
\begin{align}
  1-\FLDD\bigr|_{n=k} &= \frac{1}{\uS+\pDD\symS}\bigl(\uSuSt +\pDD \SuSt\bigr) \\
  &= 1-\FDD ,
\end{align}
\end{subequations}
in agreement with \cref{eq:F2.general}.

\subsection{Comparison of LDD+QEC and QEC-only}
We next show that in order to decide which of $\FHyb$ and $\FQEC$ is larger when $\pQEC=0$ (ideal recovery), it suffices to compare two sector-wise fractions, and the conclusion is independent of the value of $\pDD\in[0,1)$.

\begin{mytheorem}
  \label{thm:Hyb.vs.QEC}
  Assume $\pQEC = 0$ and $\pDD \in [0, 1)$.
Then, for any fixed $z>0$,
  \begin{equation}
    \label{eq:F1.vs.F4}
    \FHyb > \FQEC \iff \frac{\SuC}{\symS} > \frac{\uSuC}{\uS}.
  \end{equation}
\end{mytheorem}

We prove this using the following lemma:
\begin{mylemma}
  \label{lem:fraction-comparison}
  Let $A,B,C,D$ be real numbers with $C,D>0$, and define $f(t)=\frac{A+tB}{C+tD}$ for $t\in[0,1]$.
Then, for any $t\in[0,1)$,
  \begin{equation}
    f(t) > f(1) \iff \frac{A}{C} > \frac{B}{D}.
  \end{equation}
\end{mylemma}
\begin{proof}
  Expanding $(A+tB)(C+D) > (A+B)(C+tD)$ and canceling the common terms, we see that $f(t) > f(1)$ is equivalent to $(1-t)AD > (1-t)BC$, which is equivalent to $A/C > B/D$ (since $1-t>0$).
\end{proof}

\begin{proof}[Proof of \cref{thm:Hyb.vs.QEC}]
  With $\pQEC=0$, \cref{eq:1-F1.general} gives
  \begin{equation}
    1-\FHyb=\frac{\uSuC+\pDD\SuC}{\uS+\pDD\symS}.
  \end{equation}
  QEC-only corresponds to $\pDD=1$, i.e. $1-\FQEC=\frac{\uSuC+\SuC}{\uS+\symS}$.
  Applying \cref{lem:fraction-comparison} to $f(t)=\frac{\uSuC+t\SuC}{\uS+t\symS}$ yields
  \begin{equation}
    1-\FHyb < 1-\FQEC\ \iff\ \frac{\uSuC}{\uS}<\frac{\SuC}{\symS},
  \end{equation}
  which is equivalent to \cref{eq:F1.vs.F4}.
\end{proof}

We now come to our central QEC-LDD result: a sufficient condition for when the hybrid protocol outperforms QEC alone:

\begin{mycorollary}
  \label{cor:Hyb.vs.QEC}
  Fix $z>0$ and assume
  \begin{equation}
    \label{eq:cor-condition}
    \frac{\SuC}{\symS} > \frac{\uSuC}{\uS}.
  \end{equation}
  Then $\FHyb > \FQEC$ for any $\pDD \in [0,1)$ and all sufficiently small $\pQEC$
  (with the neighborhood size depending on $z$ and $\pDD$).
\end{mycorollary}

\begin{proof}
  Fix $z>0$ and $\pDD\in[0,1)$.
  Under the assumption \cref{eq:cor-condition}, \cref{thm:Hyb.vs.QEC} implies the strict inequality
  \begin{equation}
    \FHyb\big|_{\pQEC=0}  >  \FQEC\big|_{\pQEC=0}.
  \end{equation}
  Now view $\FHyb$ and $\FQEC$ as functions of the decoder-failure parameter $\pQEC$ (with $z,\pDD$ held fixed).
  By the closed-form expressions for $\FHyb$ and $\FQEC$, each is a rational function of $\pQEC$ with denominator strictly positive at $\pQEC=0$;
  in particular, both are continuous at $\pQEC=0$.
  Therefore, the difference
  \begin{equation}
    \Delta(\pQEC)\coloneqq \FHyb(\pQEC)-\FQEC(\pQEC)
  \end{equation}
  is continuous at $\pQEC=0$ and satisfies $\Delta(0)>0$, so there exists $\varepsilon>0$ (depending on $z$ and $\pDD$) such that
  $\Delta(\pQEC)>0$ for all $\pQEC\in[0,\varepsilon)$.
\end{proof}

To interpret \cref{eq:cor-condition}, recall that $\uS$ and $\symS$ are the total WEP weights of the unsuppressed and suppressed sectors, respectively, while
$\uSuC$ and $\SuC$ are the WEP weights of the uncorrectable subsets within the respective sectors under ideal recovery.
Thus, the ratio $\uSuC/\uS$ is the conditional probability (under the WEP weighting) that an error drawn from the
unsuppressed sector is uncorrectable, and $\SuC/\symS$ is the analogous conditional probability in the suppressed sector.
The inequality \cref{eq:cor-condition} is therefore the assumption that the suppressed sector is, in this sense, more harmful:
it contains a higher fraction of errors that would cause a logical fault even under ideal decoding.

Since LDD multiplies the suppressed-sector contribution by the suppression factor $\pDD<1$ while leaving the unsuppressed sector unchanged,
it preferentially rescales (down-weights) the sector with the higher uncorrectable-error density.
When decoder failures are sufficiently rare ($\pQEC$ small), the logical infidelity is dominated by these ideal-decoding considerations,
and this bias in what becomes suppressed yields the strict advantage $\FHyb>\FQEC$.

\subsection{The limit of small \texorpdfstring{$p$}{p}}
\label{ss:QEC-asymptotics}
The general fidelity expressions we derived above are amenable to analysis in the limit of small error probability $p$.
In this section we perform this analysis and establish the resulting hierarchy of fidelities of the different protocols.

\subsubsection{QEC-only}
\label{ss:QEC-only-asymptotics}

\begin{mylemma}
  \label{lem:F4.asymptotics}
  Let $\alpha \geq 1$ be the minimal weight of an uncorrectable error (with respect to the fixed decoding map $D$), and let $a$ be the number of such errors of weight $\alpha$.
Then asymptotically, as $p \to 0$, \cref{eq:1-F4.general.b} can be written as
  \begin{equation}
    \label{eq:F4.asymptotics}
    1 - \FQEC = a (p/3)^{\alpha} + O(p^{\alpha + 1}) + \pQEC \bigl(b (p/3) + O(p^2)\bigr).
  \end{equation}
  When $\alpha \geq 2$ and $k\ge 1$, the coefficient $b$ satisfies
  \begin{subequations}
  \begin{align}
    \label{eq:b-ineq1}
    b &= (1 - 4^{-k}) \bigl[3n - W_1\bigl(\pmS\bigr)\bigr] \\
    \label{eq:b-ineq2}
    &\geq 3k + 15 / 4 .
  \end{align}
  \end{subequations}
\end{mylemma}

\begin{proof}
  For small $p$, we have $z=\frac{p}{3-3p}= \frac{p}{3}\bigl(1+O(p)\bigr)$.
  Write the WEPs from \cref{eq:1-F4.general.b} as in \cref{eq:W}.
  Let
  \begin{equation}
    \alpha = \min\bigl\{w\ge 0  \big|  W_{w}\bigl(\tmP_{n}\setminus\tmcEc\bigr)\ne0\bigr\}\ge 1
  \end{equation}
  be the smallest weight of an uncorrectable error, and let $a\coloneqq W_{\alpha}(\tmP_{n}\setminus\tmcEc)$ be the number of such weight-$\alpha$ Paulis.
  Keeping only the leading-order term,
  \begin{equation}
  \label{eq:WP-E}
  \uC = W(\tmP_{n}\setminus\tmcEc;z)
  = az^{\alpha}+O(z^{\alpha+1}) = a(p/3)^{\alpha}+O(p^{\alpha+1}).
  \end{equation}

  Next, note that $\all=W(\tmP_n;z)=(1+3z)^n$, so
  \begin{equation}
    \frac{1}{\all}=(1+3z)^{-n}=(1-p)^n=1+O(p).
  \end{equation}
  Hence the prefactor $1/\all$ in \cref{eq:1-F4.general.b} does not change the leading $p^{\alpha}$ scaling.

  Since $\C - 4^{-k} \D$ is a polynomial in $z$ with no constant term, we can write
  \begin{equation}
    \label{eq:C-4kD-is-bz}
    \C - 4^{-k} \D = b z + O(z^2)
  \end{equation}
  for some coefficient $b$.
  Substituting these observations into \cref{eq:1-F4.general.b} yields \cref{eq:F4.asymptotics}.

  Now assume $\alpha \geq 2$ and $k\ge 1$.
  Then all weight-$1$ Paulis are corrected to the trivial logical class,
  and the set $\tmcEc\setminus\pmS$ contains all weight-$1$ Paulis that are not stabilizers.
  There are $3n$ weight-$1$ Paulis in $\tmP_n$, so
  \begin{equation}
  \label{eq:WE-S}
  W_1(\tmcEc\setminus\pmS) = 3n - W_1(\pmS).
  \end{equation}
  Also, $\alpha\ge 2$ implies there are no weight-$1$ logical Paulis, so the weight-$1$ part of $\pmSL$ equals that of $\pmS$:
  \begin{equation}
  \label{eq:W1-SL}
  W_1(\pmSL) = W_1(\pmS) ,
  \end{equation}
  and therefore
  \begin{equation}
  \label{eq:WP-SL}
  W_1(\tmP_n\setminus\pmSL) = 3n - W_1(\pmS).
  \end{equation}

  Using \cref{eq:WE-S,eq:WP-SL}, we have
  \begin{subequations}
  \label{eq:WP-E2}
  \begin{align}
    \C &= W\bigl(\tmcEc\setminus\pmS;z\bigr)
    = \bigl(3n - W_1(\pmS)\bigr)z + O(z^{2}), \\
    \D &= W\bigl(\tmP_n\setminus\pmSL;z\bigr)
    = \bigl(3n - W_1(\pmS)\bigr)z + O(z^{2}) .
  \end{align}
  \end{subequations}
  Combining with \cref{eq:C-4kD-is-bz} gives \cref{eq:b-ineq1}.

  To show \cref{eq:b-ineq2}, first note that each weight-$1$ stabilizer acts on a distinct physical qubit and hence contributes an independent stabilizer constraint, so
  \begin{equation}
    W_1(\pmS) \le n-k.
  \end{equation}
  Therefore $3n-W_1(\pmS) \ge 3n-(n-k)=2n+k$.
  Since $\alpha\ge 2$ implies the code can correct all weight-$1$ errors, its distance satisfies $d\ge 3$, and by the quantum Singleton bound $n-k\ge 2(d-1)\ge 4$, i.e., $n\ge k+4$.
  Hence $2n+k \ge 2(k+4)+k = 3k+8$, so
  \begin{subequations}
  \begin{align}
    b&=(1-4^{-k})\bigl[3n-W_1(\pmS)\bigr]\\&\ge (1-4^{-k})(3k+8)
    \ge 3k+\frac{15}{4},
  \end{align}
  \end{subequations}
  where the last inequality holds because $4^{-k}(3k+8)\le 11/4$ for all integers $k\ge 1$.
\end{proof}

We can now compare the DD-phys and QEC-only strategies using
\cref{eq:F2.explicit-2,eq:F4.asymptotics}.
When both $\pDD$ and $\pQEC$ have fixed values in $(0, 1)$, whether $\FQEC$ is larger or smaller than $\FDD$, even in the limit $p \to 0$, depends on the code and on the values of $\pDD$ and $\pQEC$ (through $k$ and the coefficient $b$).
For example, for $\pDD = \pQEC$, we have $\FQEC < \FDD$ for sufficiently small $p$ for any nontrivial code when $\alpha \geq 2$ [this follows from $b > 3k$ by \cref{eq:b-ineq2}].
One can also check that the order-$p$ coefficient of $\uC + \C - 4^{-k} \D$ exceeds $3k$ even when $\alpha = 1$ by a computation similar to the one at the end of the proof of \cref{lem:F4.asymptotics}.

On the other hand, if we assume that (1) $\pQEC = o(1)$ as $p \to 0$, while $\pDD > 0$ is fixed; and (2) the decoder corrects all weight-$1$ errors, i.e., $\alpha \geq 2$, then
\begin{equation}
\label{eq:fqec>fdd}
\FQEC > \FDD \quad \text{for sufficiently small $p$},
\end{equation}
because $1-\FQEC = o(p)$, while $1-\FDD = k  p  \pDD + O(p^2)$.

\subsubsection{LDD-only}

\begin{mylemma}
  \label{lem:F3.asymptotics}
  Asymptotically, as $p \to 0$, \cref{eq:1-F3.general} can be written as
  \begin{equation}
  \label{eq:1-F3.asymptotics}
  1-\FLDD = \bigl(a (1 - \pDD) + b \pDD\bigr) \frac{p}{3} + O(p^2),
  \end{equation}
  where
  \begin{subequations}
  \label{eq:1-F3.asymptotics-ab-v1}
  \begin{align}
    a &= W_1(\tGp \setminus \pmS) - 4^{-k} W_1(\tGp \setminus \pmSL)
    \geq 0, \\
    b &= \bigl(3n-W_1(\pmS)\bigr) - 4^{-k} \bigl(3n - W_1(\pmSL)\bigr).
  \end{align}
  \end{subequations}
  For the trivial code ($n=k$, $\pmS=\{\pi(I)\}$, $\pmSL=\tmP_k$), one has $b = 3k$.
  For any nontrivial code ($n>k$ and $k\ge 1$), one has
  \begin{equation}
  \label{eq:1-F3.asymptotics-b-ineq}
  b \ge 3k + \frac{3}{2}.
  \end{equation}
  If the code distance satisfies $d\ge 2$ (equivalently, there are no weight-$1$ logical Paulis), then
  \begin{subequations}
  \label{eq:1-F3.asymptotics-ab-v2}
  \begin{align}
    a &= (1-4^{-k})  W_1(\tGp \setminus \pmS) \geq 0, \\
    b &= (1-4^{-k}) \bigl(3n - W_1(\pmS)\bigr).
  \end{align}
  \end{subequations}
\end{mylemma}

\begin{proof}
  We write the WEPs from \cref{eq:1-F3.general} as series in $z$ [as in \cref{eq:W}]
  and keep only the $z^1$ order in the numerator and the $z^0$ order in the denominator of \cref{eq:1-F3.general}.
  After substituting
  $z=\frac{p}{3-3p}=\frac{p}{3}+O(p^{2})$, only these terms are needed to evaluate $1-\FLDD$ up to $O(p^2)$.

  The denominator terms are
  \begin{subequations}
  \label{eq:S-uS-expand}
  \begin{align}
    \uS &= W(\tGp;z) = 1+W_1(\tGp)z+O(z^{2}),\\
    \symS &= W\bigl(\tmP_n\setminus\tGp;z\bigr)
    = \bigl[3n - W_1\bigl(\tGp\bigr)\bigr] z + O(z^{2}) ,
  \end{align}
  \end{subequations}
  while the numerator terms are
  \begin{subequations}
  \begin{align}
    \uSuSt &=W(\tGp\setminus\pmS;z) = W_1(\tGp\setminus\pmS)z+O(z^{2}),\\
    \uSD &=W(\tGp\setminus\pmSL;z) = W_1(\tGp\setminus\pmSL)z+O(z^{2}),\\
    \SuSt &=W\bigl(\tmP_n\setminus\tGp\setminus\pmS;z\bigr) \\
    & = \bigl[3n-W_1(\tGp \cup \pmS)\bigr]z +O(z^{2}),\\
    \SD &=W\bigl(\tmP_n\setminus\tGp\setminus\pmSL;z\bigr) \\
    &= \bigl[3n-W_1(\tGp \cup \pmSL)\bigr]z+O(z^{2}).
  \end{align}
  \end{subequations}
  Note that
  \begin{align}
    W_1(\tGp \cup \pmS) - W_1(\tGp \setminus \pmS) &= W_1(\pmS), \\
    W_1(\tGp \cup \pmSL) - W_1(\tGp \setminus \pmSL) &= W_1(\pmSL).
  \end{align}
  Inserting the expansions into the numerator of \cref{eq:1-F3.general}, we obtain
   \begin{align}
  \label{eq:1-F3.asymptotics-num}
   &{\rm numerator} = \notag\\
    & \quad z (1-\pDD) \bigl(W_1(\tGp \setminus \pmS)
    - 4^{-k} W_1(\tGp \setminus \pmSL)\bigr)\notag \\
    & \quad+ z \pDD \bigl(3n - W_1(\pmS)
    - 4^{-k} (3n - W_1(\pmSL))\bigr) \notag\\
     &\quad + O(z^{2}).
  \end{align}
  The denominator equals $1+O(z)$ and does not affect the linear coefficient in $z$.
  Comparing \cref{eq:1-F3.asymptotics-num} with \cref{eq:1-F3.asymptotics-ab-v1}, we get
  \begin{equation}
    1-\FLDD = \bigl(a (1 - \pDD) + b \pDD\bigr) z + O(z^2),
  \end{equation}
  and substituting $z=\frac{p}{3}+O(p^2)$ yields \cref{eq:1-F3.asymptotics}.

  To show \cref{eq:1-F3.asymptotics-b-ineq} for nontrivial codes, let
  \begin{equation}
    l_1 \coloneqq W_1(\pmS).
  \end{equation}

  \noindent{Case 1: $l_1=0$.}
  Then, since $W_1(\pmSL)\ge 0$, from \cref{eq:1-F3.asymptotics-ab-v1} we have
  \begin{subequations}
  \begin{align}
    b &= 3n - 4^{-k}\bigl(3n - W_1(\pmSL)\bigr)
    \ge 3n - 4^{-k}(3n)\\
    &    = (1-4^{-k}) 3n .
  \end{align}
  \end{subequations}
  For a nontrivial code $n\ge k+1$, and therefore
  \begin{equation}
    b  \ge  (1-4^{-k}) 3(k+1)
     =  3k + 3 - \frac{3(k+1)}{4^k}.
  \end{equation}
  For all integers $k\ge 1$ one has $4^k \ge 2(k+1)$, hence
  \begin{equation}
    3 - \frac{3(k+1)}{4^k}  \ge  3 - \frac{3}{2}  =  \frac{3}{2},
  \end{equation}
  which gives $b \ge 3k+\frac{3}{2}$.

  \smallskip
  \noindent{Case 2: $l_1>0$.}
  Any weight-$1$ stabilizer acts nontrivially on a distinct physical qubit, and such stabilizers are independent.
  Remove the $l_1$ qubits on which these weight-$1$ stabilizers act, along with the corresponding $l_1$ independent
  stabilizer generators; this yields a reduced code on $n' = n-l_1$ qubits encoding the same $k$ logical qubits, and for the reduced code one has $W_1(\pi(\mS'))=0$.
  Let $b'$ denote the coefficient $b$ for the reduced code.

  For each removed qubit, the quantity $3n-W_1(\pmS)$ decreases by $3-1=2$ (since $n$ decreases by $1$ and exactly one weight-$1$ stabilizer element is removed).
  Moreover, on such a qubit the only weight-$1$ element of $\pmSL$ is that stabilizer itself (any other nontrivial single-qubit Pauli would anticommute with it), so $W_1(\pmSL)$ also decreases by $1$ per removed qubit, and hence $3n-W_1(\pmSL)$ decreases by $2$ per removed qubit as well.
  Therefore
  \begin{equation}
    b = b' + \bigl(2 - 4^{-k}\cdot 2\bigr) l_1 = b' + (1-4^{-k}) 2l_1.
  \end{equation}
  If the reduced code is trivial, then $b'=3k$ and
  \begin{equation}
    b \ge 3k + (1-4^{-k}) 2l_1 \ge 3k + \frac{3}{2},
  \end{equation}
  since $l_1\ge 1$ and $k\ge 1$ imply $(1-4^{-k}) 2l_1 \ge (1-\tfrac14)\cdot 2 = \tfrac{3}{2}$.
  If the reduced code is nontrivial, then by Case 1 we have $b'\ge 3k+\frac{3}{2}$, hence $b>3k+\frac{3}{2}$.
  This proves \cref{eq:1-F3.asymptotics-b-ineq}.

  Finally, when $d\ge 2$ there are no weight-$1$ logical Paulis, so $W_1(\pmSL) = W_1(\pmS)$.
  Substituting this into \cref{eq:1-F3.asymptotics-ab-v1} yields the simplified expressions \cref{eq:1-F3.asymptotics-ab-v2}.
\end{proof}

Comparing \cref{eq:1-F3.asymptotics} with \cref{eq:F2.explicit-2}, we have
\begin{equation}
\FDD-\FLDD
= \frac{p}{3}\Bigl[\pDD (b - 3k) + a (1 - \pDD)\Bigr] + O(p^2).
\end{equation}
By \cref{lem:F3.asymptotics}, $a\ge 0$, and moreover $b=3k$ for the trivial code while for any nontrivial code one has $b\ge 3k+\frac{3}{2}$.
Therefore, for any fixed $\pDD>0$ and any nontrivial code the coefficient of $p$ is strictly positive, so
\begin{equation}
\label{eq:fdd>fldd}
\FDD > \FLDD \quad \text{for sufficiently small $p$},
\end{equation}
with equality only in the trivial-code case.

\subsubsection{LDD+QEC}
\label{ss:LDD+QEC-fid}
Let us compare the asymptotic behavior of $\FHyb$ with $\FDD, \FLDD, \FQEC$ when $p \to 0$, $\pQEC = 0$, and fixed $\pDD \in (0, 1)$.
We already determined in \cref{eq:fqec>fdd,eq:fdd>fldd}
that in this case $\FQEC > \FDD > \FLDD$ for any code correcting single-qubit errors.
The main tool for comparing $\FHyb$ with $\FQEC$ is \cref{thm:Hyb.vs.QEC}.

\begin{mytheorem}
  \label{thm:Hyb.vs.QEC-asymptotics}
  Assume $\pQEC = 0$, $\pDD \in [0, 1)$, $k \geq 1$, and $\abs{\tG} > 1$.
  Let $\alpha$ be the minimal weight of uncorrectable errors
  (i.e., the smallest $w$ such that $W_w(\tmP_n\setminus\tmcEc)\neq 0$),
  and let $\beta\ge \alpha$ be the minimal weight of suppressed uncorrectable errors
  (i.e., the smallest $w$ such that $W_w\bigl((\tmP_n\setminus\tGp)\setminus\tmcEc\bigr)\neq 0$).
  Then:
  \begin{enumerate}
    \item If $\beta = \alpha$ (equivalently, at least one minimal-weight uncorrectable error lies in the suppressed sector), then there exists $p_0 > 0$ such that for all $p \in (0, p_0)$ we have $\FHyb > \FQEC$.
    \item Suppose $\beta > \alpha$ and, among the weight-$\alpha$ uncorrectable errors,
      there exists at least one error $E'$ with nonzero syndrome (equivalently, $\syn(E')\neq 0$).
      Fix a decomposition $\tG=\{I,g\}\tG_1$ with $\tG_1\subsetneq \tG$ (such a decomposition always exists since $\tG$ is a nontrivial $\mathbb{F}_2$-vector space under multiplication).
      Then there exists a stabilizer element $S\in\pmS$ such that, for $\tG'=\{I,gS\}\tG_1$,
      the corresponding minimal suppressed-uncorrectable weight satisfies $\beta'=\alpha$.
      In particular, $\tG'$ satisfies the condition of part~1.
    \item Suppose the code distance $d$ is at least $2$.
      Then $\alpha \leq \ceil{d / 2} < d$
      and the precondition $\syn(E') \neq 0$ in part~2 holds for any weight-$\alpha$ uncorrectable error.
    \item Suppose $\beta > \alpha$ and $\tG = \pi(\mL)$.
      Then the precondition $\syn(E') \neq 0$ in part~2 holds for any weight-$\alpha$ uncorrectable error.
  \end{enumerate}
\end{mytheorem}

\begin{proof}
  We start by showing that $\beta$ in the statement of the theorem is always well defined,
  i.e., that $\bigl(\tmP_n\setminus\tGp\bigr)\setminus\tmcEc \neq \varnothing$.
Fix a deterministic Pauli decoder $D(\sigma)$.
An error $E'\in\tmP_n$ is correctable iff $D(\syn(E'))E'\in \pi(\mS)$.
Thus, for each syndrome $\sigma$, the set of correctable errors is exactly the coset
$D(\sigma) \pi(\mS)$, which has $|\pi(\mS)|=2^{n-k}$ elements.
Since there are $2^{n-k}$ syndromes, $|\tmcEc|=2^{n-k}\cdot 2^{n-k}=4^{n-k}$ and hence
\[
|\tmP_n\setminus\tmcEc| \;=\; 4^n-4^{n-k} \;=\; 4^n(1-4^{-k}).
\]
For a nontrivial code ($k\ge 1$), this satisfies $|\tmP_n\setminus\tmcEc|>4^n/2$.
On the other hand, since $\tG$ is nontrivial, $\tGp\subsetneq\tmP_n$ and in fact $|\tGp|\le 4^n/2$
(it is contained in the commutant of any fixed nonidentity element of $\tG$).
Therefore $(\tmP_n\setminus\tmcEc)\setminus\tGp$ is nonempty, and hence
$\beta$ is well defined.

  We analyze the WEP ratios in the condition
  \cref{eq:F1.vs.F4} from \cref{thm:Hyb.vs.QEC} in the limit $z \to 0$ (equivalently $p\to 0$).
  \cref{eq:S-uS-expand} gives
  \begin{subequations}
  \begin{align}
    \uS &= 1+O(z),
    \quad
    \symS = c z + O(z^2)\\
    c &\coloneqq W_1(\tmP_n\setminus\tGp)=3n-W_1(\tGp)>0.
  \end{align}
  \end{subequations}

  Let
  \begin{equation}
    A \coloneqq W_\alpha\bigl(\tGp\setminus\tmcEc\bigr)\ge 0,
    \quad
    B \coloneqq W_\beta\bigl(\tmP_n\setminus\tGp\setminus\tmcEc\bigr)>0.
  \end{equation}
  Then
  \begin{equation}
    \uSuC = A z^\alpha + O(z^{\alpha+1}),
    \quad
    \SuC = B z^\beta + O(z^{\beta+1}),
  \end{equation}
  and therefore
  \begin{equation}
    \frac{\uSuC}{\uS} = A z^\alpha + O(z^{\alpha+1}),
    \quad
    \frac{\SuC}{\symS} = \frac{B}{c}  z^{\beta-1} + O(z^{\beta}).
  \end{equation}

  If $\beta=\alpha$, then $\frac{\SuC}{\symS}=\Theta(z^{\alpha-1})$ while $\frac{\uSuC}{\uS}=O(z^\alpha)$, so for sufficiently small $z>0$ we have
  $\frac{\uSuC}{\uS} < \frac{\SuC}{\symS}$, and hence $\FHyb>\FQEC$ by \cref{thm:Hyb.vs.QEC}.
This proves part~1.

  For part~2, assume $\beta>\alpha$
  and choose a weight-$\alpha$ uncorrectable error $E'$ with nonzero syndrome.
  Since $\beta>\alpha$, all weight-$\alpha$ suppressed-uncorrectable coefficients vanish,
  so such an $E'$ must lie in the unsuppressed sector $\tGp$
  and therefore commutes with every element of $\tG$, in particular with $g$.
  Nonzero syndrome means there exists a stabilizer element $S\in\pmS$
  such that $E'$ anticommutes with $S$.
  Then $E'$ anticommutes with $gS$ (because it commutes with $g$ and anticommutes with $S$),
  so $E'\notin (\tG')^\perp$.
  Therefore $E'$ becomes a suppressed uncorrectable error of weight $\alpha$ for $\tG'$,
  implying $\beta'=\alpha$.

  For part~3, note that uncorrectable errors $E'$ satisfying $\syn(E') = 0$ are known
  as nontrivial (i.e., non-stabilizer) logical operators,
  and recall that the code distance $d$ is defined as the minimum weight
  of such operators.
  Let $L' \in \tmP_n$ be a weight-$d$ nontrivial logical operator.
  Then $L' = E_1' E_2'$ for $E_1', E_2'$ of weights $\floor{d/2}$ and $\ceil{d/2}$, respectively.
Since $\syn(L')=0$, we have $\syn(E_1')=\syn(E_2')$.
Let $\sigma\coloneqq \syn(E_1')=\syn(E_2')$.
If both $E_1'$ and $E_2'$ were correctable, then $D(\sigma)E_1',D(\sigma)E_2'\in \pi(\mS)$.
Writing $E_i'=D(\sigma)^{-1}S_i$ with $S_i\in\pi(\mS)$, we obtain
\begin{equation}
L' \;=\; E_1'E_2' \;=\; D(\sigma)^{-2} S_1S_2 \;\in\; \pi(\mS),
\end{equation}
(using $D(\sigma)^2=I$ in the phase-stripped Pauli group),
which contradicts the fact that $L'$ is a nontrivial logical operator.
Therefore at least one of $E_1',E_2'$ is uncorrectable.
By definition, $\alpha$ is the minimal weight of an uncorrectable error, hence
\begin{equation}
\alpha \le \max\bigl(\wt(E_1'),\wt(E_2')\bigr)=\ceil{d/2}.
\end{equation}
Since $d\ge 2$, we have $\alpha \le \ceil{d/2}<d$.
Finally, let $E'$ be any weight-$\alpha$ uncorrectable error.
If $\syn(E')=0$, then $E'$ is a nontrivial logical operator of weight $\alpha<d$,
contradicting the definition of the code distance $d$.
Hence every weight-$\alpha$ uncorrectable error satisfies $\syn(E')\neq 0$, proving the precondition in part~2.

  To prove part~4, assume to the contrary that $\tG = \pi(\mL)$, $\beta > \alpha$,
  but $\syn(E') = 0$ for some weight-$\alpha$ uncorrectable error $E'$.
  Since $\syn(E') = 0$, $E'$ is a logical operator.
  Since $\beta>\alpha$, we have $W_\alpha(\tmP_n\setminus\tGp\setminus\tmcEc)=0$.
Hence every weight-$\alpha$ uncorrectable error lies in $\tGp$ and therefore commutes with all elements of $\tG$.
With $\tG=\pi(\mL)$, this means that the logical class of $E'$ commutes with all logical Paulis.
The only logical Pauli that commutes with all logical Paulis is the identity,
so $E'$ acts as the identity on the codespace.
Together with $\syn(E')=0$, this implies $E'\in\pi(\mS)$, contradicting that $E'$ is uncorrectable.
\end{proof}

We can now describe an explicit ``stabilizer-dressing'' procedure to modify $\tG$ when it does not satisfy the condition of part~1 of \cref{thm:Hyb.vs.QEC-asymptotics}, provided that there exists a minimal-weight uncorrectable error with nonzero syndrome:
(1) pick any generator list for $\tG$ and let $g$ be one generator;
(2) pick any uncorrectable error $E'$ of minimal weight $\alpha$ with $\syn(E')\neq 0$;
(3) find a stabilizer generator (or any stabilizer element) $S\in\pmS$ that anticommutes with $E'$;
(4) replace $g$ by $gS$ and rebuild the decoupling group.
The resulting group $\tG'$ has $\beta'=\alpha$ and therefore satisfies the condition of part~1, implying $\FHyb > \FQEC$ for sufficiently small $p$ (with $\pQEC=0$).

Note that if all weight-$\alpha$ uncorrectable errors have zero syndrome, i.e., are undetected logical errors, then multiplying DD generators by stabilizers cannot change their commutation relations with those errors; achieving $\beta'=\alpha$ then requires changing the logical decoupling group itself rather than only its stabilizer dressing.

\subsubsection{LDD+QEC vs QEC-only: examples when \texorpdfstring{$\beta>\alpha$}{β>α}}

A clean sufficient condition for an asymptotic hybrid-protocol advantage is given by part~1 of \cref{thm:Hyb.vs.QEC-asymptotics}:
if the leading uncorrectable weight in the unsuppressed sector and the suppressed sector coincide ($\beta=\alpha$),
then for sufficiently small $p$ the hybrid protocol satisfies $\FHyb>\FQEC$ (for any fixed $0\le \pDD<1$).
An interesting question is what happens when this condition is violated.
When $\beta > \alpha$,
the theorem is generally inconclusive: the hybrid and QEC-only protocols share the same leading $O(p^\alpha)$ scaling,
and the ordering can depend on subleading terms that are sensitive to how uncorrectable errors are distributed between the suppressed and unsuppressed sectors.

We give two examples below that are chosen to illustrate two qualitatively different mechanisms by which $\beta=\alpha$ can fail.
Our first example exhibits a structural mismatch between the chosen LDD group and the recovery map, so that the earliest suppressed uncorrectable events occur at a different weight than the earliest unsuppressed uncorrectable events; this places the hybrid-vs-QEC comparison in the delicate $\beta > \alpha$ regime where an asymptotic advantage is not automatic.
Our second example shows that $\beta > \alpha$ can also be an artifact of the
recovery/decoder choice: for simplicity, we use canonical syndrome representatives
and a decoder that is convenient to describe analytically, but such a choice can be suboptimal once the effective noise has been modified by LDD. In particular, a decoder
tailored to the original physical noise need not be optimal for the LDD-modified channel, and a suboptimal choice can mask an intrinsic hybrid-protocol advantage; the recovery/decoder choice should be adapted to the channel being optimized~\cite{ReimpellWerner:2005aa,FletcherShorWin:2007aa,Kosut:2008lq}.

Part~2 of \cref{thm:Hyb.vs.QEC-asymptotics} provides a concrete ``way out''
when the sufficient condition $\beta=\alpha$ fails: under its hypotheses, one may dress an LDD generator by a stabilizer, which changes the suppressed/unsuppressed partition without changing the logical action of the decoupling cycle,
and can enforce $\beta'=\alpha$ so that part~1 applies.
In the examples below we keep the recovery map fixed and use the exact criterion of \cref{thm:Hyb.vs.QEC} to determine whether $\FHyb$ exceeds $\FQEC$ when $\beta>\alpha$.

Thus, the examples should be read not only as demonstrations that the $\beta > \alpha$ regime can arise in natural and analytically tractable settings, but also as illustrations of why decoder choice can matter in this regime.

\paragraph*{Example 1: $\FHyb < \FQEC$.}
Consider any code $\mC$ with distance $d=7$ and a suboptimal choice of decoding map $D$:
such a code admits a decoder capable of correcting all errors of weight up to $3$,
but let us assume that $D$ fails to correct one (and exactly one) of those errors, call it $E_0$, with weight $\wt(E_0) = 2$.
Let $\tG$ be any decoupling group commuting with $E_0$ and suppressing at least some uncorrectable weight-$4$ errors.
In this case $\alpha = 2$ and $\beta = 4$, so $\beta > \alpha$
and the condition in part~1 of \cref{thm:Hyb.vs.QEC-asymptotics} is violated.
Thus, \cref{thm:Hyb.vs.QEC-asymptotics} does not directly resolve the comparison of $\FHyb$ and $\FQEC$, and we instead rely on \cref{thm:Hyb.vs.QEC}.
We have
\begin{equation}
\frac{\uSuC}{\uS} = z^2 + O(z^3),
\end{equation}
since there is exactly one weight-$2$ uncorrectable error and it is unsuppressed, and
\begin{equation}
\frac{\SuC}{\symS} = \frac{b}{c}  z^3 + O(z^4),
\end{equation}
where $b$ is the number of suppressed uncorrectable weight-$4$ errors and
$c = W_1(\tmP_n\setminus\tGp) = 3n - W_1(\tGp)$ is the number of suppressed weight-$1$ Paulis (by assumption $b,c>0$).
Thus, for sufficiently small $z$ (equivalently, $p$), $\frac{\uSuC}{\uS} > \frac{\SuC}{\symS}$ and, applying \cref{thm:Hyb.vs.QEC}, we conclude that $\FHyb < \FQEC$.

\paragraph*{Example 2: $\FHyb > \FQEC$.}
Consider the following set of stabilizer generators and logical operators that define a $[[13, 1, 3]]$ code:
\begin{equation}
  \begin{aligned}
    S_{0} & = {\rm \texttt{IIIIZIZXXZIXI}} \\
    S_{1} & = {\rm \texttt{IIIXIIIYZYYYI}} \\
    S_{2} & = {\rm \texttt{IIIZYIYYIXIIY}} \\
    S_{3} & = {\rm \texttt{IXXXYZYIZZYZI}} \\
    S_{4} & = {\rm \texttt{IXXXZIXZXZZIX}} \\
    S_{5} & = {\rm \texttt{IXYZZXZXXXXZX}} \\
    S_{6} & = {\rm \texttt{IZZZIIZZYIIXY}} \\
    S_{7} & = {\rm \texttt{XXIZZIYZIIYIY}} \\
    S_{8} & = {\rm \texttt{XYZYYYXXYYZIY}} \\
    S_{9} & = {\rm \texttt{YXZZIYZZXZZZZ}} \\
    S_{10} & = {\rm \texttt{ZIZZXIZZZZXYI}} \\
    S_{11} & = {\rm \texttt{ZYYZIXIZZYIYY}} \\
    \bar{Z} & = {\rm \texttt{IIIYYYYZIXYXY}} \\
    \bar{X} & = {\rm \texttt{XXXXXXXXXIIII}}
  \end{aligned}
\end{equation}
We pick $\tG = \pi(\{I, \bar{X}, \bar{Z}, \bar{Y}\})$, where $\bar{Y} = i \bar{X} \bar{Z}$.
We pick a decoding map $D$ that is able to correct all weight-$1$ errors
and the following set of weight-$2$ errors on the first $3$ qubits:
$\{\texttt{IXZ}, \texttt{IZX}, \texttt{ZIX}, \texttt{XIZ}, \texttt{XZI}, \texttt{ZXI}\}$.
This forces it to fail to correct the weight-$2$ errors
$\{\texttt{ZYI}, \texttt{ZIY}, \texttt{IZY}, \texttt{YZI}, \texttt{YIZ}, \texttt{IYZ}\}$,
because they differ from the previous set by $\texttt{ZZZ}$ (on the first $3$ qubits),
which is a nontrivial logical operator
(equivalent to $\bar{Z}$ up to a stabilizer and phase).
The errors $\texttt{ZZI}$, $\texttt{ZIZ}$, and $\texttt{IZZ}$ are also uncorrectable because they produce the same syndrome as $\texttt{IIZ}$, $\texttt{IZI}$, and $\texttt{ZII}$, respectively.
All uncorrectable weight-$2$ errors commute with $\tG$, so $\alpha = 2$.
There are suppressed uncorrectable weight-$3$ errors (e.g., $\texttt{ZZZ}$ on the first $3$ qubits), so $\beta = 3 > \alpha$.
Thus, the condition in part~1 of \cref{thm:Hyb.vs.QEC-asymptotics} is violated and we again rely on \cref{thm:Hyb.vs.QEC} to compare $\FHyb$ and $\FQEC$.

Computing the relevant WEPs using a computer program (see \cref{sss:Methodology}), we find
\begin{subequations}
\begin{align}
  \uSuC &= 9 z^2 + 843 z^3 + O(z^4), \\
  \uS   &= 1 + 8 z + 186 z^2 + O(z^3), \\
  \SuC  &= 2544 z^3 + O(z^4), \\
  \symS &= 31 z + 516 z^2 + O(z^3).
\end{align}
\end{subequations}
In this case, $\uSuC/\uS = 9 z^2 + O(z^3)$ and $\SuC/\symS = (2544/31) z^2 + O(z^3)$, so $\uSuC/\uS < \SuC/\symS$ for sufficiently small $z$ and, from \cref{thm:Hyb.vs.QEC}, we conclude that $\FHyb > \FQEC$.
Thus, in this case, despite $\beta > \alpha$, $\FHyb > \FQEC$ for small $z$.
Indeed, evaluating $1 - \FHyb$ and $1 - \FQEC$ directly using \cref{eq:1-F1.general,eq:1-F4.general.b}, we find
\begin{subequations}
\begin{align}
  \label{eq:1-F1-example2}
  1 - \FHyb &= 9 z^2 + (771 + 2265 \pDD) z^3 + O(z^4), \\
  \label{eq:1-F4-example2}
  1 - \FQEC &= 9 z^2 + 3036 z^3 + O(z^4).
\end{align}
\end{subequations}
Alternatively, one could obtain \cref{eq:1-F4-example2} by substituting $\pDD = 1$ into \cref{eq:1-F1-example2}.
The conclusion is that, in this example, $\FHyb$ is larger than $\FQEC$ by $2265 (1 - \pDD) z^3 + O(z^4)$.

\section{Error Detection}
\label{sec:QED}

\subsection{Setup}
The QEC framework of \cref{sec:theorem-2-ec} extends naturally to QED.
Here we specialize the common setup of \cref{sec:setup} to the case in which no recovery map is applied: instead, we accept a run iff the reported syndrome is trivial.
This postselected-detection viewpoint has a long history in fault-tolerant schemes based on error-detecting codes~\cite{Knill:2004postselected,Knill:2005realistically}.
The relevant observables are therefore (i) the acceptance probability $P_A$ and (ii) the logical fidelity conditioned on acceptance, denoted again by $F$ throughout this section.
The waiting-interval noise and the DD/LDD action are the same as in the QEC setting.
Detection imperfections are modeled by assuming that the syndrome is reported correctly with probability $1-\pQED$ and is otherwise replaced by a uniformly random syndrome among the $2^{n-k}$ possibilities, so the trivial outcome occurs with probability $2^{-(n-k)}$ in the faulty-readout branch.

As in the QEC case, the distinction here is operational rather than a property of the code distance: the same stabilizer code may be used either for correction or for detection, depending only on how the syndrome information is processed.
As in \cref{tab:tags} for the QEC case, \cref{tab:tags-QED} lists the WEPs that appear in this section.

\begin{table}
  \footnotesize
  \begin{center}
    \begin{tabular}{|c|c|c|}
      \hline
      tag & weight enumerator & set description \\ \hline
      \all & $W\bigl(\tmP_n;z\bigr) = (1+3z)^n$ & all \\ \hline
      \symS & $W\bigl(\tmP_n\setminus \tGp;z\bigr)$ & suppressed \\ \hline
      \uS & $W\bigl(\tGp;z\bigr)$ & unsuppressed \\ \hline
      \St & $W\bigl(\pmS;z\bigr)$ & stabilizers \\ \hline
      \symL & $W\bigl(\pmSL \setminus \pmS;z\bigr)$ & nontrivial logical \\ \hline
      \StL & $W\bigl(\pmSL;z\bigr)$ & all undetected (zero-syndrome) \\ \hline
      \D & $W\bigl(\tmP_n \setminus \pmSL;z\bigr)$ & detected (nonzero syndrome) \\ \hline
      \uSt & $W\bigl(\tmP_n \setminus \pmS;z\bigr)$ & non-stabilizers \\ \hline
      \SSt & $W\bigl(\pmS \setminus \tGp; z\bigr)$ & suppressed stabilizers \\ \hline
      \uSSt & $W\bigl(\pmS \cap \tGp; z\bigr)$ & unsuppressed stabilizers \\ \hline
      \SL & $W\bigl(\pmSL \setminus \pmS \setminus \tGp; z\bigr)$ & suppressed non-\St{} logical \\ \hline
      \uSL & $W\bigl(\tGp \cap \pmSL \setminus \pmS; z\bigr)$ & unsuppressed non-\St{} logical \\ \hline
      \SStL & $W\bigl(\pmSL \setminus \tGp; z\bigr)$ & suppressed all undetected \\ \hline
      \uSStL & $W\bigl(\pmSL \cap \tGp; z\bigr)$ & unsuppressed all undetected \\ \hline
      \SD & $W\bigl(\tmP_n \setminus \pmSL \setminus \tGp; z\bigr)$ & suppressed detected \\ \hline
      \uSD & $W\bigl(\tGp \cap \tmP_n \setminus \pmSL; z\bigr)$ & unsuppressed detected \\ \hline
      \SuSt & $W\bigl(\tmP_n \setminus \pmS \setminus \tGp; z\bigr)$ & suppressed non-stabilizers \\ \hline
      \uSuSt & $W\bigl(\tGp \setminus \pmS; z\bigr)$ & unsuppressed non-stabilizers \\ \hline
    \end{tabular}
  \end{center}
  \caption{List of all WEPs that appear in the QED conditional-fidelity and acceptance-probability expressions.
  Each error set (middle column) is described in terms of its relation to DD and QED (right column)
  and is given a corresponding acronym tag (left column) aligned with the set description.
  We always evaluate each WEP at $z=p/(3-3p)$.}
  \label{tab:tags-QED}
\end{table}

We again compare four strategies: DD-phys, QED-only, LDD-only, and the hybrid protocol LDD+QED.
Here LDD-only means that we use the same encoded block as in QED-only and LDD+QED and the same DD group $\tG$ as in LDD+QED,
but omit postselection, so $P_A=1$ by definition; any out-of-codespace component is mapped back to the codespace by the same random-projection convention used earlier.
To avoid confusion with $\FHyb$ from the QEC section, we denote the fidelity of the hybrid LDD+QED protocol by $\FHybD$.

\begin{mylemma}[Hybrid LDD+QED: conditional fidelity and acceptance probability]
\label{lem:QED.Hyb.general}
Consider the hybrid protocol in the error-detection setting described above: an encoded block undergoes (L)DD, followed by stabilizer syndrome measurement, and we accept iff the measured syndrome is trivial.
Assume the effective Pauli noise model with parameter $z=p/(3-3p)$
and the QED readout model in which ideal readout is used with probability $1-\pQED$
and the reported syndrome is otherwise uniformly random over the $2^{n-k}$ possible syndromes.
If a detected (nonzero-syndrome) error is nevertheless accepted due to readout failure, we model the resulting logical operation as a uniformly random logical Pauli (equivalently, the output is maximally mixed on the codespace), so the probability of the trivial logical class is $4^{-k}$.

Then the entanglement infidelity conditioned on acceptance is
\begin{multline}
  \label{eq:QED.Hyb.1-F}
  1 - \FHybD = \\
  Q_{{\rm Hyb}}^{-1}
  \bigl(
  (1-\pQED(1-2^{k-n})) (\uSL + \pDD \SL) \\
  + (1-4^{-k}) 2^{k-n} \pQED (\uSD + \pDD \SD) \bigr),
\end{multline}
and the acceptance probability is
\begin{equation}
  \label{eq:QED.Hyb.PA}
  P_{A,{\rm Hyb}} = \frac{Q_{{\rm Hyb}}}{\uS + \pDD \symS},
\end{equation}
where
\begin{multline}
  \label{eq:QED.Hyb.Q}
  Q_{{\rm Hyb}} =
  \bigl(
  (1-\pQED) (\uSStL + \pDD \SStL) \\
  + 2^{k-n} \pQED (\uS + \pDD \symS) \bigr) .
\end{multline}
All WEP tags are as in \cref{tab:tags-QED}.
\end{mylemma}

\begin{proof}
Let $E'\in\tmP_n$ denote the effective Pauli error after the waiting/(L)DD step.
As in the QEC setting, under our Pauli model the (unnormalized) probability weight of a subset
$A\subseteq\tmP_n$ is proportional to its WEP $W(A;z)$, with the additional DD factor $\pDD$ for suppressed errors.
Hence the total weight of all errors is $\uS + \pDD \symS$.

Acceptance probability:
There are $2^{n-k}$ stabilizer syndromes, so if the readout is uniformly random then the probability to report the
trivial syndrome is $2^{-(n-k)}$.
We consider the two readout branches.
\begin{enumerate}
\item Ideal readout (probability $1-\pQED$).
Acceptance occurs iff the true syndrome is trivial, i.e., iff $E'\in\pmSL$.
The corresponding acceptance weight is
$\uSStL + \pDD \SStL$.

\item Random readout (probability $\pQED$).
Independently of $E'$, the protocol accepts with probability $2^{k-n}$, so the corresponding acceptance weight is
$2^{k-n}(\uS+\pDD \symS)$.
\end{enumerate}
Multiplying by the respective branch probabilities and adding yields \cref{eq:QED.Hyb.Q}.
Dividing by the total weight $\uS+\pDD \symS$ gives \cref{eq:QED.Hyb.PA}.

Conditional logical infidelity:
We now count accepted events that produce a nontrivial logical Pauli on the $k$-qubit
logical subsystem, and then condition on acceptance.
Again we split into the two readout branches.

\begin{enumerate}
\item Ideal readout (probability $1-\pQED$).
Conditioned on acceptance, we necessarily have $E'\in\pmSL$.
Such an accepted error produces a logical fault iff it is a non-stabilizer logical Pauli,
i.e., iff $E'\in \pmSL\setminus \pmS$, whose WEP contribution is $\uSL+\pDD \SL$.
Thus the ideal-readout branch contributes the logical-fault weight
\begin{equation}
\label{eq:ideal-QED}
  (1-\pQED)(\uSL+\pDD \SL).
\end{equation}

\item Random readout (probability $\pQED$).
Acceptance occurs with probability $2^{k-n}$ irrespective of $E'$.
\begin{enumerate}
\item[(i)] If $E'\in\pmSL\setminus\pmS$ (non-stabilizer logical), then whenever we accept we keep a state
with a nontrivial logical Pauli applied, so this contributes
\begin{equation}
\label{eq:random-QED-i}
  2^{k-n}\pQED(\uSL+\pDD \SL).
\end{equation}

\item[(ii)] If $E'\notin\pmSL$ (detected error), then an acceptance event corresponds to a readout failure.
By assumption, in this case the induced logical Pauli is uniform over the $4^k$ logical classes, so the probability
of the trivial class is $4^{-k}$ and the probability of a logical fault is $1-4^{-k}$.
The accepted detected errors have WEP weight $\uSD+\pDD \SD$, hence this contributes
\begin{equation}
 \label{eq:random-QED-ii}
 (1-4^{-k}) 2^{k-n}\pQED(\uSD+\pDD \SD).
\end{equation}
\end{enumerate}
\end{enumerate}

Combining \cref{eq:ideal-QED,eq:random-QED-i} yields the total coefficient
$(1-\pQED) + 2^{k-n}\pQED  =  1-\pQED(1-2^{k-n})$
in front of $(\uSL+\pDD \SL)$.
Adding \cref{eq:random-QED-ii}
gives the numerator in
\cref{eq:QED.Hyb.1-F}.
Conditioning on acceptance divides by the acceptance weight $Q_{\rm Hyb}$,
completing the proof of \cref{eq:QED.Hyb.1-F}.
\end{proof}

\begin{mycorollary}[QED-only]
\label{cor:QED.only}
Setting $\pDD=1$ in \cref{lem:QED.Hyb.general} yields the QED-only expressions
\begin{multline}
  \label{eq:QED.QED.1-F}
  1 - \FQED = Q_{{\rm QED}}^{-1}
  \bigl(
  (1-\pQED(1-2^{k-n})) \symL \\
  + (1-4^{-k}) 2^{k-n} \pQED \D
  \bigr)
\end{multline}
and
\begin{equation}
  \label{eq:QED.QED.PA}
  P_{A,{\rm QED}} = \frac{Q_{{\rm QED}}}{\all},
\end{equation}
where
\begin{equation}
  \label{eq:QED.QED.Q}
  Q_{{\rm QED}} = (1-\pQED) \StL + 2^{k-n} \pQED \all.
\end{equation}
\end{mycorollary}

\begin{proof}
When $\pDD=1$, the suppressed/unsuppressed split recombines as follows:
$\uS+\symS=\all$,
$\uSStL+\SStL=\StL$,
$\uSL+\SL=\symL$, and $\uSD+\SD=\D$.
Substituting these identities and $\pDD=1$ into \cref{eq:QED.Hyb.Q,eq:QED.Hyb.1-F,eq:QED.Hyb.PA}
gives \cref{eq:QED.QED.Q,eq:QED.QED.1-F,eq:QED.QED.PA}.
\end{proof}

\begin{mycorollary}[LDD-only]
\label{cor:QED.LDD.only}
In the LDD-only strategy (no postselection), one has $P_{A,{\rm LDD}}=1$ and
\begin{multline}
  \label{eq:QED.LDD.1-F}
  1 - \FLDD = \\
  \frac{ \uSuSt + \pDD \SuSt -4^{-k} (\uSD + \pDD \SD) }{\uS + \pDD \symS}.
\end{multline}
\end{mycorollary}

\begin{proof}
By definition of LDD-only in this section, the output is always mapped back to the codespace by replacing any
out-of-codespace component with the maximally mixed codespace state.
Thus stabilizer errors contribute no logical fault, non-stabilizer logical errors contribute with full weight,
and detected errors contribute a logical fault with probability $1-4^{-k}$.
Using the decompositions
$\uSuSt=\uSL+\uSD$ and $\SuSt=\SL+\SD$,
the logical-fault weight is
$(\uSL+\pDD \SL) + (1-4^{-k})(\uSD+\pDD \SD)
= (\uSuSt+\pDD \SuSt) - 4^{-k}(\uSD+\pDD \SD)$,
and dividing by the total weight $\uS+\pDD \symS$ yields \cref{eq:QED.LDD.1-F}.
\end{proof}

Note that LDD-only is inequivalent to setting $\pQED=1$ (or any other value) in the hybrid LDD+QED expressions,
because $\pQED$ models imperfect syndrome readout and postselection, whereas LDD-only omits the postselection step entirely.

\begin{mycorollary}[DD-phys]
\label{cor:QED.DD.only}
For DD-phys, i.e., DD applied to $k$ physical qubits (no encoding, no postselection), the fidelity is
\begin{equation}
  \label{eq:QED.DD.F}
  \FDD = \bigl(1 + \pDD \bigl((1+3z)^k-1\bigr)\bigr)^{-1},
\end{equation}
and $P_{A,{\rm DD}}=1$.
\end{mycorollary}

\begin{proof}
This is identical to the DD-only case in the QEC setting, \cref{eq:F2.explicit-2}, since it depends only on the
single-qubit Pauli error model and not on whether one later performs correction or detection.
\end{proof}

\subsection{Modified protocol based on partial order}
\label{ss:QED-partial-order}
We define a partial order on pairs $(F, P_A)$ by writing $(F_1, P_{A,1}) \succeq (F_2, P_{A,2})$ iff $F_1 \geq F_2$ and
$P_{A,1} (F_1 - 4^{-k}) \geq P_{A,2} (F_2 - 4^{-k})$, where $F-4^{-k}$ is the fidelity excess above the maximally mixed $k$-qubit state.
These conditions ensure that scenario~1 is at least as good as scenario~2 in the following operational sense: we can simulate scenario~2 using scenario~1 by optionally discarding accepted runs (to reduce acceptance) and, if $P_{A,1}<P_{A,2}$, by ``padding'' some of the rejected runs by outputting and accepting a maximally mixed $k$-qubit state (whose fidelity is $4^{-k}$).

Concretely, if $P_{A,1}<P_{A,2}$, then when scenario~1 rejects, we additionally accept the maximally mixed state with conditional probability
$\frac{P_{A,2}-P_{A,1}}{1-P_{A,1}}$, so that the overall acceptance becomes $P_{A,2}$; this padding increases acceptance while contributing zero fidelity excess, and the inequality
$P_{A,1} (F_1 - 4^{-k}) \geq P_{A,2} (F_2 - 4^{-k})$ ensures the resulting conditional fidelity is still at least $F_2$.
If instead $P_{A,1}\geq P_{A,2}$, we can reduce the acceptance probability by uniformly discarding accepted runs; the condition $F_1\geq F_2$ then ensures that the conditional fidelity remains at least $F_2$.

\subsection{Comparison of LDD+QED and QED-only}
\label{ss:QED-compare}

We now compare the hybrid LDD+QED protocol to QED-only under the assumption $\pQED = 0$.
As in \cref{thm:Hyb.vs.QEC}, using \cref{lem:fraction-comparison} we have the following theorem.
\begin{mytheorem}
  \label{thm:QED.Hyb.vs.QED}
  Assume $\pQED = 0$, $\pDD \in [0, 1)$, and $\SStL \neq 0$. Then
  \begin{equation}
    \label{eq:QED.Hyb.vs.QED.a}
    \FHybD \geq \FQED \iff \frac{\uSL}{\uSStL} \leq \frac{\SL}{\SStL}.
  \end{equation}
  Under the partial order of \cref{ss:QED-partial-order}, $\bigl(\FHybD,P_{A,{\rm Hyb}}\bigr)\succeq\bigl(\FQED,P_{A,{\rm QED}}\bigr)$ iff, in addition to the right-hand side of \cref{eq:QED.Hyb.vs.QED.a}, we also have
  \begin{equation}
    \label{eq:QED.Hyb.vs.QED.b}
    \frac{\uSSt - 4^{-k} \uSStL}{\uS} \geq \frac{\SSt - 4^{-k} \SStL}{\symS}.
  \end{equation}
  When $\SStL = 0$, $\tG$ contains only stabilizers, $\FHybD = \FQED$,
  and $\bigl(\FHybD,P_{A,{\rm Hyb}}\bigr)\succeq\bigl(\FQED,P_{A,{\rm QED}}\bigr)$.
  The corresponding statements with strict inequalities also hold.
\end{mytheorem}

\begin{proof}
From \cref{eq:QED.Hyb.Q,eq:QED.Hyb.1-F,eq:QED.Hyb.PA}, setting $\pQED=0$ gives
\begin{subequations}
\begin{align}
  Q_{\rm Hyb} &=\uSStL+\pDD \SStL,\\
  1-\FHybD&=\frac{\uSL+\pDD \SL}{\uSStL+\pDD \SStL},\\
  P_{A,{\rm Hyb}}&=\frac{\uSStL+\pDD \SStL}{\uS+\pDD \symS}.
\end{align}
\end{subequations}
Likewise, QED-only is obtained by setting $\pDD=1$ in the hybrid expressions (equivalently, no suppression),
so that
\begin{equation}
  1-\FQED=\frac{\uSL+\SL}{\uSStL+\SStL},
  \quad
  P_{A,{\rm QED}}=\frac{\uSStL+\SStL}{\uS+\symS}.
\end{equation}

\paragraph{Fidelity comparison.}
Define
\begin{equation}
  f(t)\coloneqq \frac{\uSL+t \SL}{\uSStL+t \SStL},
  \quad t\in[0,1],
\end{equation}
so that $1-\FHybD=f(\pDD)$ and $1-\FQED=f(1)$.
Then
\begin{equation}
 \FHybD\ge \FQED \iff f(\pDD)\le f(1).
\end{equation}
Applying \cref{lem:fraction-comparison} with
$A=\uSL$, $B=\SL$, $C=\uSStL$, $D=\SStL$, and $t=\pDD\in[0,1)$,
we obtain
\begin{equation}
  f(\pDD)\le f(1)
  \iff
  \frac{\uSL}{\uSStL}\le \frac{\SL}{\SStL},
\end{equation}
which proves \cref{eq:QED.Hyb.vs.QED.a}.
The corresponding statement with strict inequalities follows
from the strict form in \cref{lem:fraction-comparison}.

\paragraph{Partial order comparison.}
Under the partial order of \cref{ss:QED-partial-order}, in addition to $\FHybD\ge \FQED$ we require
\begin{equation}
  \label{eq:partial-order-extra}
  P_{A,{\rm Hyb}}\bigl(\FHybD-4^{-k}\bigr) \ge  P_{A,{\rm QED}}\bigl(\FQED-4^{-k}\bigr).
\end{equation}
Define
\begin{equation}
  \label{eq:def-interp-PA-F}
  P_A(t) \coloneqq \frac{\uSStL+t \SStL}{\uS+t \symS},
  \quad
  F(t) \coloneqq 1-\frac{\uSL+t \SL}{\uSStL+t \SStL},
\end{equation}
where $t\in[0,1]$ is an interpolating parameter.
Under the assumption $\pQED=0$, these are exactly the acceptance probability and conditional fidelity
for the hybrid protocol with parameter $t$.
In particular,
\begin{subequations}
  \label{eq:specialize-t}
\begin{align}
  (P_A(\pDD),F(\pDD))&=(P_{A,{\rm Hyb}},\FHybD),\\
  (P_A(1),F(1))&=(P_{A,{\rm QED}},\FQED).
\end{align}
\end{subequations}

Next, since $\pmSL=\pmS\sqcup\bigl(\pmSL\setminus\pmS\bigr)$ and the suppressed/unsuppressed split is by $\tGp$,
the corresponding WEPs satisfy
\begin{equation}
  \label{eq:WEP-partitions}
  \uSStL=\uSSt+\uSL,
  \quad
  \SStL=\SSt+\SL.
\end{equation}
Substituting \cref{eq:WEP-partitions} into \cref{eq:def-interp-PA-F} gives
\begin{equation}
  \label{eq:F-of-t-detected-form}
  F(t)
  =1-\frac{\uSL+t \SL}{\uSStL+t \SStL}
  =\frac{\uSSt+t \SSt}{\uSStL+t \SStL}.
\end{equation}

Now consider the quantity that appears in the partial order;
using \cref{eq:def-interp-PA-F,eq:F-of-t-detected-form}, we obtain
\begin{subequations}
\begin{align}
  &P_A(t)\bigl(F(t)-4^{-k}\bigr)\notag\\
  &\quad\quad=
  \frac{\uSStL+t \SStL}{\uS+t \symS}\left(\frac{\uSSt+t \SSt}{\uSStL+t \SStL}-4^{-k}\right)\\
  &\quad\quad=
  \frac{\uSSt+t \SSt}{\uS+t \symS}
  -4^{-k}\frac{\uSStL+t \SStL}{\uS+t \symS}\\
  \label{eq:PA-times-Fminus}
  &\quad\quad=
  \frac{\bigl(\uSSt-4^{-k}\uSStL\bigr)+t\bigl(\SSt-4^{-k}\SStL\bigr)}{\uS+t \symS}.
\end{align}
\end{subequations}
This identity expresses the partial-order quantity $P_A(t)\bigl(F(t)-4^{-k}\bigr)$ as a linear-fractional function of $t$,
so that we can apply \cref{lem:fraction-comparison}.

Define
\begin{equation}
  \label{eq:def-g}
  g(t) \coloneqq \frac{A+tB}{C+tD},
\end{equation}
where $A\coloneqq \uSSt-4^{-k}\uSStL$, $B\coloneqq \SSt-4^{-k}\SStL$, $C\coloneqq \uS$, and $D\coloneqq \symS$.
Then \cref{eq:PA-times-Fminus} is exactly $P_A(t)\bigl(F(t)-4^{-k}\bigr)=g(t)$.
Therefore, by \cref{eq:specialize-t}, the extra partial-order condition \cref{eq:partial-order-extra} becomes
\begin{equation}
  \label{eq:g-ineq}
  g(\pDD)\ge g(1).
\end{equation}
Applying \cref{lem:fraction-comparison} to $g(t)$ with $t=\pDD\in[0,1)$ yields
\begin{multline}
  g(\pDD)\ge g(1)
  \iff
  \frac{A}{C}\ge \frac{B}{D}\\
  \iff
  \frac{\uSSt-4^{-k}\uSStL}{\uS}\ge \frac{\SSt-4^{-k}\SStL}{\symS},
\end{multline}
which is exactly \cref{eq:QED.Hyb.vs.QED.b}.
The corresponding strict-inequality statements follow
from the strict form of \cref{lem:fraction-comparison}.

Finally, combining \cref{eq:QED.Hyb.vs.QED.a} (for $\FHybD\ge \FQED$) with \cref{eq:QED.Hyb.vs.QED.b}
gives the stated equivalence for
$\bigl(\FHybD,P_{A,{\rm Hyb}}\bigr)\succeq\bigl(\FQED,P_{A,{\rm QED}}\bigr)$.
\end{proof}

We now come to our central LDD-QED result: a sufficient condition under which the hybrid protocol outperforms QED alone, analogous to \cref{cor:Hyb.vs.QEC}.

\begin{mycorollary}
  \label{cor:QED.Hyb.vs.QED}
  If $\pQED = 0$, $\pDD \in [0, 1)$, and $\tG = \pi(\mL)$,
  then $\FHybD > \FQED$ and
  $\bigl(\FHybD,P_{A,{\rm Hyb}}\bigr)\succ\bigl(\FQED,P_{A,{\rm QED}}\bigr)$
  for any $z > 0$.
\end{mycorollary}
\begin{proof}
  When $\tG = \pi(\mL)$, any nontrivial logical Pauli anticommutes with some element of $\tG$, while all stabilizers commute with $\tG$.
Hence $\uSL = 0$ and $\SSt = 0$, while $\uSStL = \uSSt = \St$.
Substituting these into \cref{eq:QED.Hyb.vs.QED.a,eq:QED.Hyb.vs.QED.b} shows that both inequalities hold (strictly for any $z>0$).
\end{proof}

As in \cref{cor:Hyb.vs.QEC}, by continuity, if there is a strict inequality for $\pQED=0$, then it also holds for sufficiently small $\pQED > 0$.

\subsection{The limit of small \texorpdfstring{$p$}{p}}
\label{ss:QED-asymptotics}
As in \cref{ss:QEC-asymptotics}, we can analyze the asymptotic behavior of $1-F$ and $P_A$ in the limit $z \to 0$ (which corresponds to $p \to 0$).
This leads to the following result:

\begin{widetext}
  \begin{mylemma}
    \label{lem:QED.QED.asymptotics}
    Let $d \geq 2$ be the distance of the code (i.e., the minimum weight of an element in $\pmSL\setminus\pmS$), and let $a \coloneqq W_d\bigl(\pmSL\setminus\pmS\bigr)$ be the number of weight-$d$ nontrivial logical Paulis (as counted in the WEP).
    Then, in the limit $z \to 0$ (with any dependence of $\pQED$ on $z$),
    the infidelity and rejection probability have the following asymptotic behavior:
    \begin{subequations}
    \begin{align}
    \label{eq:QED.infid.1}
      1 - \FQED &=
      \frac{ a z^d + (1-4^{-k}) 2^{k-n} (3n-W_1(\pmS)) \pQED z + O(z^{d+1}+\pQED z^2) }{1-(1-2^{k-n})\pQED}\\
    \label{eq:QED.infid.2}
     &= a z^d + (1-4^{-k}) 2^{k-n} (3n-W_1(\pmS)) \pQED z + O(z^{d+1} + \pQED z^2 + \pQED^2z),\\
    \label{eq:QED.reject}
      1 - P_{A,{\rm QED}} &= (1-2^{k-n}) \pQED + (3n - W_1(\pmS)) z + O(z^2 + z\pQED).
    \end{align}
    \end{subequations}
  \end{mylemma}
\end{widetext}

The expansion for $1-\FQED$ separates two conceptually different contributions.
The term $a z^d$ is the ``ideal QED'' contribution: conditioned on accepting the trivial syndrome,
the leading logical faults come from undetected logical Paulis of minimum weight $d$.
The linear-in-$z$ term proportional to $\pQED$ is instead induced by QED failure: when the syndrome-extraction step fails,
a detected physical error can be accepted, and for small $z$ the dominant such contribution comes from weight-$1$ detected Paulis.
The coefficient $3n-W_1(\pmS)$ is exactly the number of weight-$1$ Paulis that are not stabilizers (and hence would
be rejected under ideal detection).

Similarly, the rejection probability has a constant term $(1-2^{k-n})\pQED$ coming from QED failure even when $z=0$:
upon failure, we model the syndrome as random, so acceptance occurs only with probability $2^{k-n}$.
The linear term $(3n-W_1(\pmS))z$ is the physical-noise contribution to rejection under ideal detection, again dominated by
weight-$1$ detected Paulis.

Let us now prove \cref{lem:QED.QED.asymptotics}.
\begin{proof}
  We expand \cref{eq:QED.QED.Q,eq:QED.QED.1-F,eq:QED.QED.PA} in the limit $z\to 0$.
  First,
  \begin{equation}
  \label{eq:all-QED}
  \all=W(\tmP_n;z)=(1+3z)^n=1+3n z+O(z^2).
  \end{equation}
  Next, since $\pmSL$ contains the identity,
  \begin{equation}
    \StL \equiv W(\pmSL;z)=1+W_1(\pmSL) z+O(z^2).
  \end{equation}
  Also, by definition $\D \equiv W(\tmP_n\setminus \pmSL;z)=\all-\StL$,
  and since $d\geq 2$ implies $W_1(\pmSL)=W_1(\pmS)$, we have
  \begin{equation}
    \D = \bigl(3n-W_1(\pmS)\bigr) z + O(z^2).
  \end{equation}

  For $\symL\equiv W(\pmSL\setminus\pmS;z)$, the code distance assumption implies that
  $W_w(\pmSL\setminus\pmS)=0$ for all $w<d$, and $W_d(\pmSL\setminus\pmS)=a$.
  Hence
  \begin{equation}
    \symL = a z^d + O(z^{d+1}).
  \end{equation}

  Now consider the infidelity formula \cref{eq:QED.QED.1-F}.
  Using the above expansions and $d\ge 2$ [so that $\pQED z^d$ can be absorbed into the $O(\pQED z^2)$ remainder],
  the numerator becomes
  \begin{align}
    &(1-\pQED(1-2^{k-n})) \symL + (1-4^{-k})2^{k-n}\pQED \D \notag\\
    &\quad\quad=
    a z^d + (1-4^{-k})2^{k-n}\bigl(3n-W_1(\pmS)\bigr)\pQED z\notag\\
    &\quad\quad\quad+ O(z^{d+1}+\pQED z^2).
  \end{align}
  For the normalization factor \cref{eq:QED.QED.Q},
\begin{subequations}
\begin{align}
    Q_{{\rm QED}}^{-1}
    &=[(1-\pQED)\StL+2^{k-n}\pQED \all]^{-1}\\
    &=\bigl[1-(1-2^{k-n})\pQED\bigr]^{-1} + O(z).
  \end{align}
  \end{subequations}
  Multiplying by the numerator only produces additional contributions of order $O(z^{d+1})$ and $O(\pQED z^2)$, so we obtain \cref{eq:QED.infid.1}.

  To obtain \cref{eq:QED.infid.2}, expand
  \begin{equation}
    \bigl[1-(1-2^{k-n})\pQED\bigr]^{-1}
    = 1 + (1-2^{k-n})\pQED + O(\pQED^2),
  \end{equation}
  and note again that the cross-term $(1-2^{k-n})\pQED\cdot a z^d$ can be absorbed into $O(\pQED z^2)$ since $d\ge 2$.

Finally, for the acceptance probability \cref{eq:QED.QED.PA},
\begin{subequations}
 \begin{align}
P_{A,{\rm QED}} &=\frac{Q_{{\rm QED}}}{\all}\\
&=\frac{\bigl[1-(1-2^{k-n})\pQED\bigr] + W_1(\pmSL) z }{\all}\notag\\
&\quad\quad +\frac{O(z^2+z\pQED)}{\all},
 \end{align}
 \end{subequations}
 where, by \cref{eq:all-QED}, $1/\all = (1+3n z+O(z^2))^{-1}=1-3n z+O(z^2)$.
  Using $W_1(\pmSL)=W_1(\pmS)$ then yields
\cref{eq:QED.reject} as claimed.
\end{proof}

Now consider the other strategies and the resulting hierarchy in the low-$p$ (or low-$z$) regime.
First, the DD-only and LDD-only fidelities are exactly the same as in the QEC setting [i.e.,  \cref{eq:QED.DD.F,eq:QED.LDD.1-F}],
so their small-$z$ expansions are unchanged.
In particular, for any fixed $\pDD>0$ one has
$1-\FDD=\Theta(z)$,
while \cref{lem:F3.asymptotics} implies that for any nontrivial code the linear coefficient of $1-\FLDD$ is strictly larger than that of $1-\FDD$,
so the conclusion of \cref{eq:fdd>fldd} continues to hold in the QED setting:
$\FDD>\FLDD$ for all sufficiently small $p$ and any nontrivial code.

By contrast, \cref{lem:QED.QED.asymptotics} shows that for a distance-$d$ code with $d\ge 2$,
\begin{equation}
  1-\FQED = a z^d + O(\pQED z),
\end{equation}
so if $\pQED=o(1)$ as $z\to 0$ then $1-\FQED=o(z)$.
Since $1-\FDD=\Theta(z)$ for fixed $\pDD>0$, it follows that
\begin{equation}
  \FQED>\FDD>\FLDD
\end{equation}
for all sufficiently small $p$ and any nontrivial code with $d\ge 2$ and $\pQED=o(1)$.

Finally, the hybrid LDD+QED protocol reduces to QED-only when $\pDD=1$, and for $\pDD\in[0,1)$ it is compared to QED-only by
\cref{thm:QED.Hyb.vs.QED}: depending on the code and the chosen logical decoupling group, LDD can either improve or degrade the conditional
fidelity relative to QED-only.
In particular, whenever the fidelity condition in \cref{thm:QED.Hyb.vs.QED} holds (so that $\FHybD\ge \FQED$), we obtain the sharpened low-$p$ hierarchy
\begin{equation}
  \FHybD \ge \FQED > \FDD > \FLDD
\end{equation}
for all sufficiently small $p$ (under the same assumptions as above).

\begin{figure*}[t]
\includegraphics[width=\textwidth]{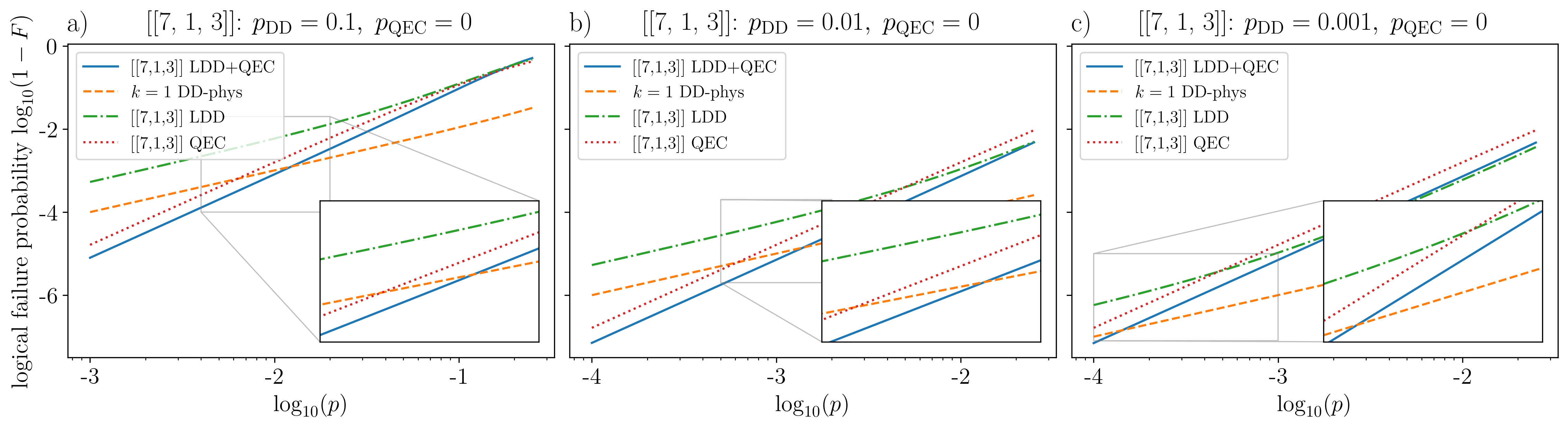} \caption{Logical failure probability for our four strategies using the $[[7,1,3]]$ code and the decoding map of \cref{ss:713-Decoding-Map}, shown as a function of the physical Pauli error probability $p$ for (a) $\pDD=0.1$, (b) $\pDD=0.01$, and (c) $\pDD = 0.001$, with perfect recovery ($\pQEC=0$).
For sufficiently small $p$, LDD+QEC achieves the lowest logical failure probability among the encoded strategies, in agreement with \cref{thm:Hyb.vs.QEC-asymptotics}.}
  \label{fig:7_1_3_comparison_Theorem2_setting}
\end{figure*}

\section{Numerical Calculations}
\label{sec:numerics}

This section numerically evaluates the closed-form fidelity expressions derived in \cref{sec:theorem-2-ec,sec:QED} and uses the resulting plots to (i) illustrate the theorem regimes and (ii) explore parameter settings beyond the small-$p$ assumptions that underlie parts of the asymptotic analysis.
Throughout this section we compare four memory strategies. In the QEC numerics these are DD-phys, LDD-only, QEC-only, and LDD+QEC; in the QED numerics, QEC-only and LDD+QEC are replaced by QED-only and LDD+QED. As noted in \cref{sec:QED}, LDD-only in the detection setting omits postselection and therefore is not obtained by setting $\pQED=1$ in the hybrid formulas.

These numerical calculations assume the same phenomenological effective error model as in the rest of this paper. In particular, they do not attempt to simulate any specific sequence of physical pulses or any microscopic physical noise model. Instead, they numerically compute the logical failure probability as a function of the model parameters, such as the QEC code, $p$, $\pDD$, and $\pQEC$.

\begin{figure*}[th]
  \centering
 \includegraphics[width=\textwidth]{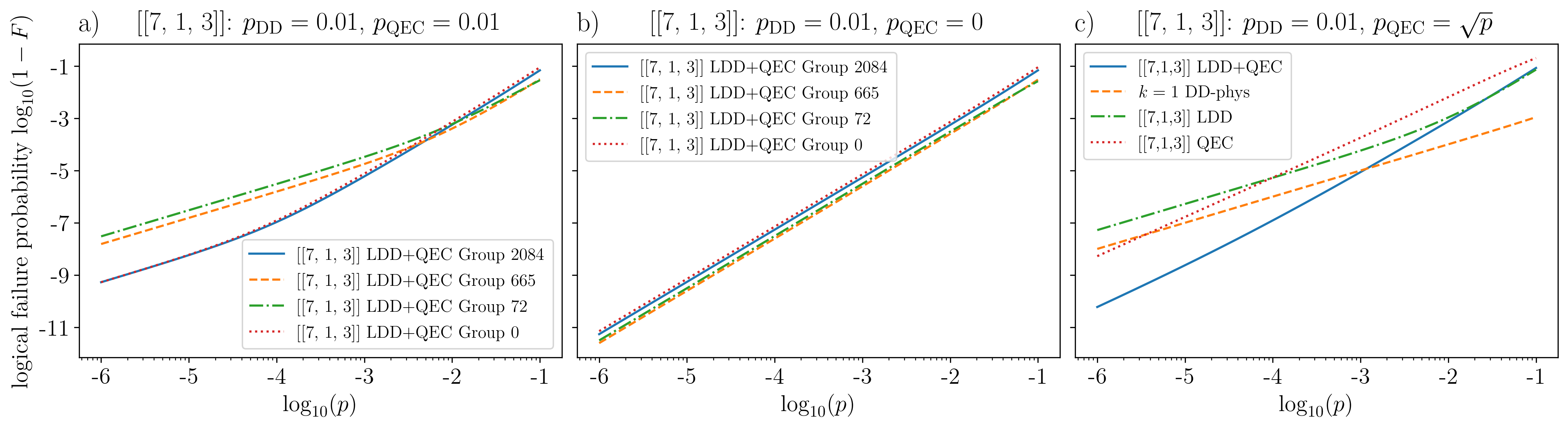}
 \caption{Panels (a) and (b): comparison of LDD+QEC performance for different choices of LDD generators using the $[[7,1,3]]$ code at fixed $\pDD=0.01$.  (a) Imperfect recovery $\pQEC=0.01$. (b) Perfect recovery $\pQEC=0$.
The generating sets shown are LDD group 2084, generated by $\langle \texttt{YXYXYXY},  \texttt{XZXZXZX}\rangle$; LDD group 665, generated by $\langle \texttt{XIIYYZZ},  \texttt{ZIIXXYY}\rangle$; LDD group 72, generated by $\langle \texttt{XXXIIII},  \texttt{ZZZIIII}\rangle$; and the standard choice, $\langle \texttt{XXXXXXX},  \texttt{ZZZZZZZ}\rangle$.
The identity of the best-performing LDD group depends on both $p$ and $\pQEC$, illustrating that LDD should be co-designed with the recovery map and the expected recovery imperfection level. (c) Protocol comparison for $\pQEC=\sqrt{p}$.
Although $\pQEC=\sqrt{p}$ is outside the perfect-recovery setting of \cref{thm:Hyb.vs.QEC}, the hybrid LDD+QEC protocol remains favorable in the small-$p$ regime for the parameters shown.}
  \label{fig:7_1_3_LDD+QEC_comparing_LDD_generating_sets}
\end{figure*}

\subsection{Methodology}
\label{sss:Methodology}

All plots are generated by direct numerical evaluation of the analytic formulas for the fidelities $\FDD$, $\FLDD$, $\FQEC$, and $\FHyb$ [see \cref{eq:F2.general,eq:1-F3.general,eq:1-F4.general.b,eq:1-F1.general}],
together with the corresponding QED conditional-fidelity formulas for $\FQED$ and $\FHybD$
[\cref{eq:QED.QED.1-F,eq:QED.Hyb.1-F}],
and, where relevant, the acceptance-probability formulas
[\cref{eq:QED.QED.PA,eq:QED.Hyb.PA}]
derived in \cref{sec:QED}.
We implement the relevant weight enumerator polynomials for the tagged Pauli subsets defined in \cref{tab:tags,tab:tags-QED} and evaluate the resulting polynomial expressions as functions of the physical error probability $p$, the DD suppression parameter $\pDD$, and the phenomenological recovery (or detection) failure rate $\pQEC$ (or $\pQED$).
The polynomial arithmetic is handled using NumPy's \texttt{Polynomial} class~\cite{numpy_poly_doc}.

To compare strategies in the $(\pDD,\pQEC)$ plane we use the logical failure probability
\begin{equation}
  \epsilon \coloneqq 1-F,
\end{equation}
and plot the relative improvement of the hybrid method over a comparator strategy,
\begin{equation}
  R \coloneqq \log_{10}\frac{\epsilon_{\mathrm{comp}}}{\epsilon_{\mathrm{hyb}}} .
  \label{eq:relative-advantage}
\end{equation}
Here $R>0$ indicates that LDD+QEC has lower logical failure probability than the comparator, while $R<0$ indicates the opposite.

Since most of the expressions we evaluate in our numerical calculations are ratios of polynomials with integer coefficients, the resulting values can, in principle, be computed and compared exactly for rational input parameters without any numerical approximations. For simplicity, however, we use standard double-precision floating-point arithmetic, which introduces small numerical errors. These errors can affect some claims of optimality and degeneracy. For example, when numerical calculations yield identical fidelities for different DD groups or decoder choices, the exact analytical fidelities (under the assumption of our model) may either be truly identical or differ by a small amount. As a reminder of this effect, we use the wording ``up to numerical precision'' when presenting the corresponding results.

\subsection{Steane \texorpdfstring{$[[7,1,3]]$}{[[7,1,3]]} code}
\label{sec:713}

We first use the $[[7,1,3]]$ Steane code~\cite{Steane:1997active,RyanAnderson:2021prx} with the recovery map specified in \cref{ss:713-Decoding-Map}.
This recovery corrects all single-qubit Pauli errors and additionally corrects a subset of weight-2 errors (42 nontrivial weight-2 Pauli errors for the chosen decoding map).
Unless stated otherwise, the ``standard'' LDD group is generated by the conventional logical operators $\langle X_L,Z_L\rangle=\langle \texttt{XXXXXXX},  \texttt{ZZZZZZZ}\rangle$.

\subsubsection{Ideal recovery regime \texorpdfstring{($\pQEC=0$)}{pQEC=0}}
\Cref{fig:7_1_3_comparison_Theorem2_setting} shows the logical failure probability as a function of the physical Pauli error probability $p$ for three suppression strengths $\pDD\in\{0.1,0.01,0.001\}$ in the ideal-recovery setting $\pQEC=0$.
Each panel compares four strategies: physical DD on a single unencoded qubit (DD-phys), LDD-only on the encoded block, QEC-only, and the hybrid LDD+QEC protocol.

Two qualitative features are apparent.
First, in the low-noise regime the encoded strategies (QEC-only and LDD+QEC) eventually outperform the uncorrected strategies (DD-phys and LDD-only) as $p$ decreases.
This is because with $\pQEC=0$ the recovery map removes all errors in the correctable set exactly, so both QEC-only and LDD+QEC fail only due to {uncorrectable} errors; consequently, their small-$p$ behavior is governed by the minimum weight of an uncorrectable error for the chosen recovery map, leading to a higher-order power-law scaling in $p$.
By contrast, DD-phys and LDD-only do not include an active correction step, so their logical failure probabilities remain dominated by lower-weight errors and scale as $O(p)$ (up to $\pDD$-dependent prefactors).
The resulting separation in slopes on the log-log plots explains why the encoded strategies dominate as $p\to 0$ in every panel.

Second, within the encoded strategies, the hybrid protocol consistently improves upon QEC-only in the plotted range.
Importantly, this ordering is not a small-$p$ statement: \cref{thm:Hyb.vs.QEC} provides an exact criterion at $\pQEC=0$ for each fixed $p$, given by \cref{eq:F1.vs.F4}.
Recall that this criterion means that LDD+QEC is advantageous precisely when the suppressed sector contains a larger fraction of uncorrectable errors than the unsuppressed sector, so that suppressing that sector preferentially targets the errors that survive perfect recovery.
In the present Steane-code example (with the specified recovery map and LDD choice), this mechanism manifests visually as the blue curve lying below the red curve throughout the displayed $p$ range.
In the small-$p$ limit, the same mechanism reduces to a constant-factor (prefactor) improvement: QEC-only and LDD+QEC share the same leading power of $p$ set by the minimum uncorrectable weight, but LDD reduces the coefficient by suppressing a subset of the uncorrectable contributions.

Finally, the comparison with DD-phys highlights the role of $\pDD$.
At larger physical error rates, DD-phys can yield a lower logical failure probability than the encoded strategies, because it directly suppresses the physical error channel while the finite-distance code and recovery map admit many higher-weight uncorrectable errors whose total probability grows rapidly with $p$.
Decreasing $\pDD$ strengthens DD suppression and therefore shifts the orange DD-phys curve downward, pushing the crossover with the encoded curves to smaller values of $p$ (and, in the formal limit $\pDD\to 0$ of perfect suppression in this effective model, favoring DD-phys over an increasingly wide $p$ range).
The LDD-only strategy remains uncompetitive in these plots because, without an error-correction step, encoding alone does not prevent low-weight errors from producing logical failure.

\subsubsection{Dependence on the LDD generating set}
The performance of LDD+QEC depends not only on the code and the recovery map, but also on the specific representatives chosen for the logical generators that define the LDD group.
For a single-logical-qubit stabilizer code, an LDD group is specified by a choice of $(X_L,Z_L)$ in the normalizer.
Multiplying either representative by an element of the stabilizer yields an equivalent logical action, but it can change the physical support of the applied pulses and therefore change which physical Pauli errors commute with the decoupling group (and hence which Pauli errors lie in the suppressed versus unsuppressed sectors in our effective model).
For the $[[7,1,3]]$ code, the stabilizer has size $2^{n-k}=64$, so there are $64^2=4096$ stabilizer-equivalent choices of $(X_L,Z_L)$ and thus many distinct candidate LDD groups to consider.

\Cref{fig:7_1_3_LDD+QEC_comparing_LDD_generating_sets} illustrates this dependence by comparing several representative LDD choices at fixed $\pDD=0.01$.
The plotted generating sets (defined in the caption) are selected to be optimal (up to numerical precision) in different noise regimes among the $4096$ possibilities.
The left and right panels use imperfect recovery ($\pQEC=0.01$ and $\pQEC=\sqrt{p}$, respectively), whereas the center panel uses perfect recovery ($\pQEC=0$).
The qualitative difference between the panels can be understood by identifying which error mechanisms dominate the logical failure probability in each setting.

In \cref{fig:7_1_3_LDD+QEC_comparing_LDD_generating_sets}(a), recovery is imperfect, so even correctable physical errors can lead to logical failure with probability proportional to $\pQEC$.
In particular, for small $p$ the dominant contribution can come from events in which a low-weight (typically weight-$1$) Pauli error occurs during the memory interval and the recovery step fails; schematically, these contributions scale like $\pQEC p$ and can dominate uncorrectable-error contributions that scale like $p^2$ when $p\ll \pQEC$.
In this low-$p$ regime, the best-performing LDD choices are therefore those that most strongly suppress the most likely low-weight errors and (to the extent possible) also suppress the subset of low-weight uncorrectable errors singled out by the particular decoding map.
This is the role played by LDD group 2084 (defined in \cref{tab:LDD-Groups}) in \cref{fig:7_1_3_LDD+QEC_comparing_LDD_generating_sets}(a).
As $p$ increases into an intermediate regime, multi-qubit errors become more probable and the logical failure probability becomes increasingly influenced by higher-weight uncorrectable contributions rather than by recovery failures on otherwise correctable errors.
Accordingly, the optimal LDD choice shifts toward groups that suppress a larger share of the uncorrectable (and typically higher-weight) errors relevant to the chosen recovery map, even if they do not prioritize suppressing every single-qubit error.
This shift is exemplified by the emergence of LDD group 665 as the best performer over much of the intermediate-$p$ range in \cref{fig:7_1_3_LDD+QEC_comparing_LDD_generating_sets}(a).
At still larger $p$, additional groups can become optimal within LDD+QEC (see \cref{tab:LDD-Groups}), although in that regime the encoded strategies may cease to be competitive with DD-phys (see \cref{fig:7_1_3_comparison_Theorem2_setting}), so the optimal choice of LDD group is less operationally relevant.

\begin{table}[ht]
\centering
\renewcommand{\arraystretch}{1.2}
\setlength{\tabcolsep}{7pt}
\begin{tabular}{|c c c|}
\hline
\multicolumn{3}{|c|}{\textit{Low-$p$ regime}} \\
\hline
\textit{Index} & $\mathbf{X_L}$ & $\mathbf{Z_L}$ \\
\hline
2084 & \texttt{YXYXYXY} & \texttt{XZXZXZX} \\
2308 & \texttt{ZXZXZXZ} & \texttt{YZYZYZY} \\
2605 & \texttt{YXYYXYX} & \texttt{XZXXZXZ} \\
2885 & \texttt{ZXZZXZX} & \texttt{YZYYZYZ} \\
3126 & \texttt{YYXXYYX} & \texttt{XXZZXXZ} \\
3462 & \texttt{ZZXXZZX} & \texttt{YYZZYYZ} \\
3647 & \texttt{YYXYXXY} & \texttt{XXZXZZX} \\
4039 & \texttt{ZZXZXXZ} & \texttt{YYZYZZY} \\
\hline
\multicolumn{3}{|c|}{\textit{Intermediate-$p$ regime}} \\
\hline
\textit{Index} & $\mathbf{X_L}$ & $\mathbf{Z_L}$ \\
\hline
665 & \texttt{XIIYYZZ} & \texttt{ZIIXXYY} \\
721 & \texttt{XIIZZYY} & \texttt{ZIIYYXX} \\
\hline
\multicolumn{3}{|c|}{\textit{High-$p$ regime}} \\
\hline
\textit{Index} & $\mathbf{X_L}$ & $\mathbf{Z_L}$ \\
\hline
72 & \texttt{XXXIIII} & \texttt{ZZZIIII} \\
\hline
\end{tabular}
\caption{Exhaustive list of LDD generating sets that attain the minimum logical failure probability (up to numerical precision) of LDD+QEC in different $p$ regimes for the imperfect-recovery setting $\pQEC>0$ in \cref{fig:7_1_3_LDD+QEC_comparing_LDD_generating_sets}(a).}
\label{tab:LDD-Groups}
\end{table}

\Cref{fig:7_1_3_LDD+QEC_comparing_LDD_generating_sets}(b) concerns the opposite limit: perfect recovery ($\pQEC=0$).
In this case all correctable errors contribute zero logical failure probability, and the dominant contributions come from uncorrectable errors (beginning at the minimum uncorrectable weight for the chosen recovery map).
Consequently, LDD groups that primarily suppress errors that are already correctable provide little benefit, whereas groups that preferentially suppress the uncorrectable sector yield the largest improvement.
This explains why a single LDD choice (here, group 665) can dominate over nearly the entire $p$ range in \cref{fig:7_1_3_LDD+QEC_comparing_LDD_generating_sets}(b).
In other words, the best LDD generators depend on $\pQEC$ because $\pQEC$ changes which error classes are most responsible for logical failure.

Finally, \cref{tab:LDD-Groups} summarizes this behavior for the imperfect-recovery setting by listing all LDD choices that attain the minimum logical failure probability (up to numerical precision) within the low-, intermediate-, and high-$p$ regimes of \cref{fig:7_1_3_LDD+QEC_comparing_LDD_generating_sets}(a).
In our scan over the $4096$ stabilizer-equivalent possibilities, these regimes typically contain small families of degenerate (up to numerical precision) optima; the table reports an exhaustive list of those choices.
Overall, these results reinforce that the best LDD group is the one that suppresses the errors that actually dominate logical failure for the chosen recovery map and the given (or expected) recovery imperfection level.

\subsubsection{Imperfect recovery}
Realistic implementations of QEC are imperfect, and even phenomenologically small recovery failure rates can qualitatively change which error mechanisms dominate the logical failure probability.
When $\pQEC>0$, an error that would otherwise be correctable can still lead to logical failure if the recovery step fails.
Consequently, the logical failure probability can receive contributions of order $\pQEC p$ from single-qubit errors followed by recovery failure, in addition to the purely uncorrectable-error contributions that scale as $p^\alpha$ in the perfect-recovery model.
This motivates exploring hybrid performance outside the idealized $\pQEC=0$ setting of \cref{thm:Hyb.vs.QEC}, and more broadly, understanding how the interplay between $(p,\pDD,\pQEC)$ shapes the relative ordering of strategies.

As a representative example, \cref{fig:7_1_3_LDD+QEC_comparing_LDD_generating_sets}(c) sets $\pQEC=\sqrt{p}$ to deliberately model a regime in which recovery faults vanish with $p$ but do not become negligible fast enough for the perfect-recovery asymptotics to apply.
In particular, for the present decoding map the perfect-recovery logical failure probability is governed by uncorrectable errors beginning at weight $\alpha$ (leading to $O(p^\alpha)$ behavior), whereas the $\pQEC p$ contribution scales as $p^{3/2}$ under $\pQEC=\sqrt{p}$ and can dominate the small-$p$ behavior when $\alpha\ge 2$.
\Cref{fig:7_1_3_LDD+QEC_comparing_LDD_generating_sets}(c) shows that, even in this imperfect-recovery setting, LDD+QEC remains favorable in the low-noise regime for the parameters shown.
Intuitively, the hybrid protocol benefits from both mechanisms: QEC removes errors in the correctable set when recovery succeeds, while LDD reduces the rate at which errors (including low-weight ones that become harmful when recovery fails) occur in the first place.

\begin{figure*}[th]
  \centering
\includegraphics[width=\textwidth]{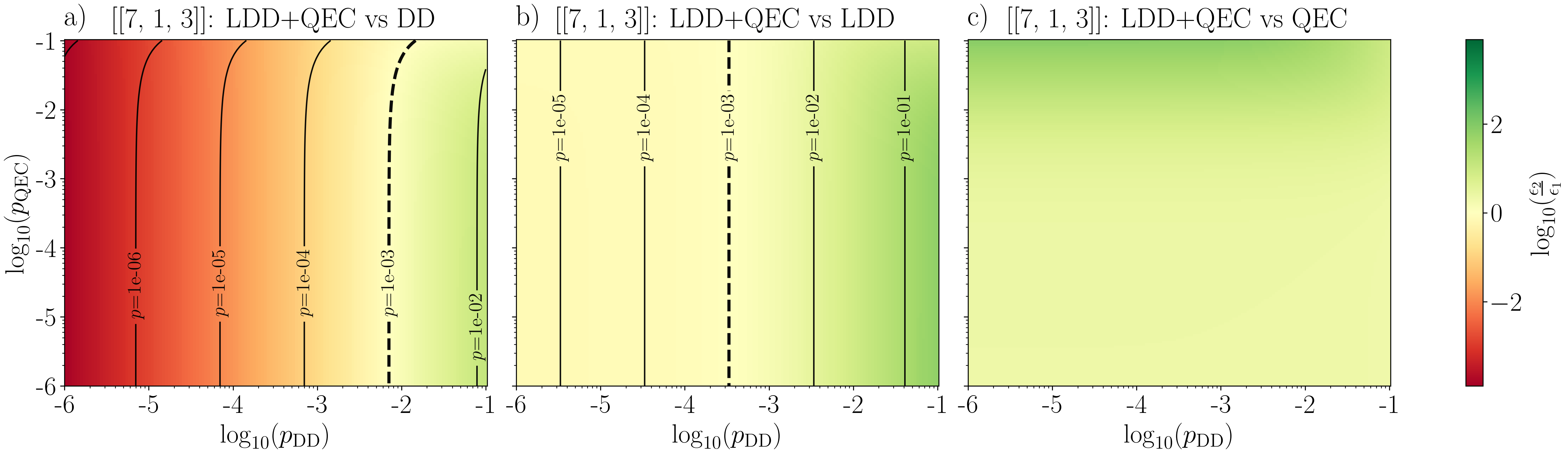}
\caption{Relative improvement $R=\log_{10}(\epsilon_{\mathrm{comp}}/\epsilon_{\mathrm{hyb}})$ for the $[[7,1,3]]$ code at fixed $p=10^{-3}$, comparing LDD+QEC against (a) DD-phys, (b) LDD-only, and (c) QEC-only.
Green (red) indicates parameter regions where LDD+QEC has lower (higher) logical failure probability than the comparator.
Contours mark the boundary $R=0$ for several values of $p$ (as labeled), illustrating how the advantage region evolves as $p$ decreases.
In (c), LDD+QEC dominates throughout the $(\pDD,\pQEC)$ plane for all values of $p \in [10^{-6},0.1]$.}
  \label{fig:7_1_3_relative_advantage_comparison}
\end{figure*}
To map where this advantage persists, \cref{fig:7_1_3_relative_advantage_comparison} sweeps the parameters $(\pDD,\pQEC)$ at fixed $p=10^{-3}$ and compares LDD+QEC pairwise against DD-phys, LDD-only, and QEC-only using the relative-improvement metric $R$ defined in \cref{eq:relative-advantage}.
Green regions ($R>0$) indicate parameter settings where the hybrid method has lower logical failure probability than the comparator, while red regions ($R<0$) indicate the opposite.
The contour lines mark the boundary $R=0$ for several other values of $p$ (as labeled) and therefore summarize how the advantage region (always to the right of the contour line) evolves as the physical noise strength decreases.

Two trends are particularly relevant.
First, at fixed $p$ there can be a region in which DD-phys outperforms encoded strategies when DD is sufficiently effective (small $\pDD$), reflecting the fact that strong physical suppression can be more valuable than a finite-distance encoding when the physical error rate is not yet in the regime where the code distance yields a large separation in logical failure.
Second, as $p$ decreases, the $R=0$ contours shift toward smaller $\pDD$ and/or smaller $\pQEC$, indicating that the region in which LDD+QEC outperforms the comparator expands in the low-noise regime.
This behavior is consistent with the design principle emphasized throughout this work: hybrid-protocol benefit becomes most pronounced when logical decoupling suppresses the error classes that dominate logical failure for the chosen recovery map, including the low-weight errors that become relevant when recovery is imperfect.

\subsection{\texorpdfstring{$[[13,1,3]]$}{[[13,1,3]]} code}
\label{sec:1313}

We repeat the same numerical analysis for a $[[13,1,3]]$ stabilizer code~\cite{YuBierbrauerDongChenOh:2013aa} to emphasize that the benefit of combining LDD with QEC depends on the compatibility between (i) the code and recovery map (which determine which Pauli errors are correctable \textit{vs.}\ uncorrectable) and (ii) the chosen LDD generating set (which determines which Pauli errors are placed in the suppressed sector of the effective noise model).
In the notation of \cref{thm:Hyb.vs.QEC-asymptotics}, let $\alpha$ denote the minimum weight of an uncorrectable Pauli error for the chosen recovery map, and let $\beta$ denote the minimum weight among uncorrectable Pauli errors that are suppressed by LDD (i.e., anticommute with at least one element of the logical decoupling group $\tG$).
In the plots below we take $\tG$ to be generated by the logical operators $\langle \texttt{IIIYYYYZIXYXY},  \texttt{XXXXXXXXXIIII}\rangle$.
For this choice (together with the fixed recovery map used in our numerics), the minimum-weight uncorrectable errors commute with $\tG$ and therefore remain unsuppressed, so $\beta>\alpha$.
As we discuss next, this directly impacts the small-$p$ relationship between LDD+QEC and QEC-only.

\subsubsection{Ideal recovery regime \texorpdfstring{($\pQEC=0$)}{pQEC=0}}
\begin{figure*}[th]
  \centering
 \includegraphics[width=\textwidth]{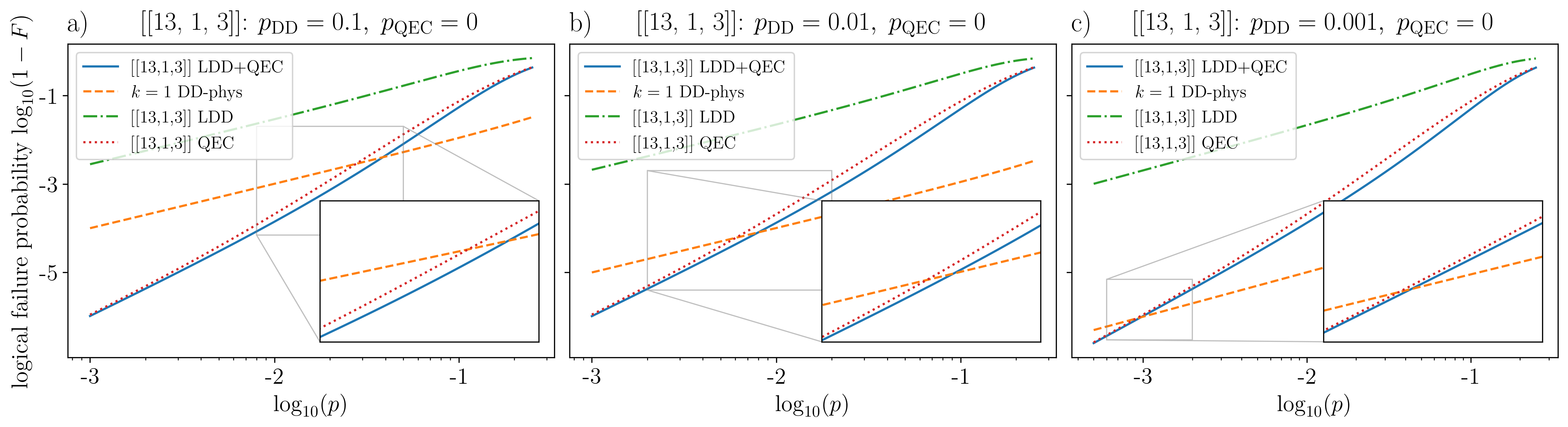}
 \caption{Logical failure probability for our protocols using a $[[13,1,3]]$ code with perfect recovery ($\pQEC=0$) and (a) $\pDD=0.1$, (b) $\pDD=0.01$, and (c) $\pDD = 0.001$.
For the LDD generators used here, we have $\beta>\alpha$ in the notation of \cref{thm:Hyb.vs.QEC-asymptotics} (no minimal-weight uncorrectable error lies in the suppressed sector), i.e., the sufficient low-$p$ condition $\beta=\alpha$ is not satisfied, and correspondingly the LDD+QEC and QEC-only curves approach one another as $p\to 0$.
In all other respects, the protocol ordering and overall scaling behavior are qualitatively similar to those of the Steane code shown in \cref{fig:7_1_3_comparison_Theorem2_setting}.}
  \label{fig:13_1_3_comparison_Theorem2_setting}
\end{figure*}
\Cref{fig:13_1_3_comparison_Theorem2_setting} shows the logical failure probability as a function of $p$ for representative suppression strengths $\pDD\in\{0.1,0.01,0.001\}$ with perfect recovery ($\pQEC=0$).
As in the Steane-code example, both QEC-only and LDD+QEC exhibit the higher-order small-$p$ scaling associated with a finite-distance code and ideal recovery: in this regime logical failure is driven by uncorrectable errors rather than by correctable single-qubit errors.
However, in contrast to the $[[7,1,3]]$ case, the LDD+QEC and QEC-only curves approach one another as $p\to 0$.

This convergence is expected given that $\beta>\alpha$ for the chosen LDD generators, as mentioned above.
Since no minimum-weight uncorrectable error lies in the suppressed sector, LDD cannot reduce the leading uncorrectable contribution that sets the QEC-only logical failure rate at small $p$; it can only suppress higher-weight uncorrectable errors and therefore modifies subleading corrections.
As a result, the hybrid protocol's improvement is most visible at intermediate $p$, before the asymptotic small-$p$ regime is fully dominated by the weight-$\alpha$ uncorrectable errors.

Varying $\pDD$ changes the overall strength of suppression and therefore shifts the crossovers with DD-phys and LDD-only, but it does not alter the qualitative small-$p$ conclusion in this example: when $\beta>\alpha$, decreasing $\pDD$ cannot produce a leading-order separation between LDD+QEC and QEC-only because the dominant minimum-weight uncorrectable errors remain unsuppressed for this LDD choice.
This $[[13,1,3]]$ case therefore provides a concrete illustration of the co-design message: hybrid advantage is not automatic and can depend critically on whether the chosen LDD generators suppress the errors that control logical failure for the specific code and recovery map.

\subsubsection{Imperfect recovery}
\begin{figure}[t]
  \hspace{0cm}{\includegraphics[width=\columnwidth]{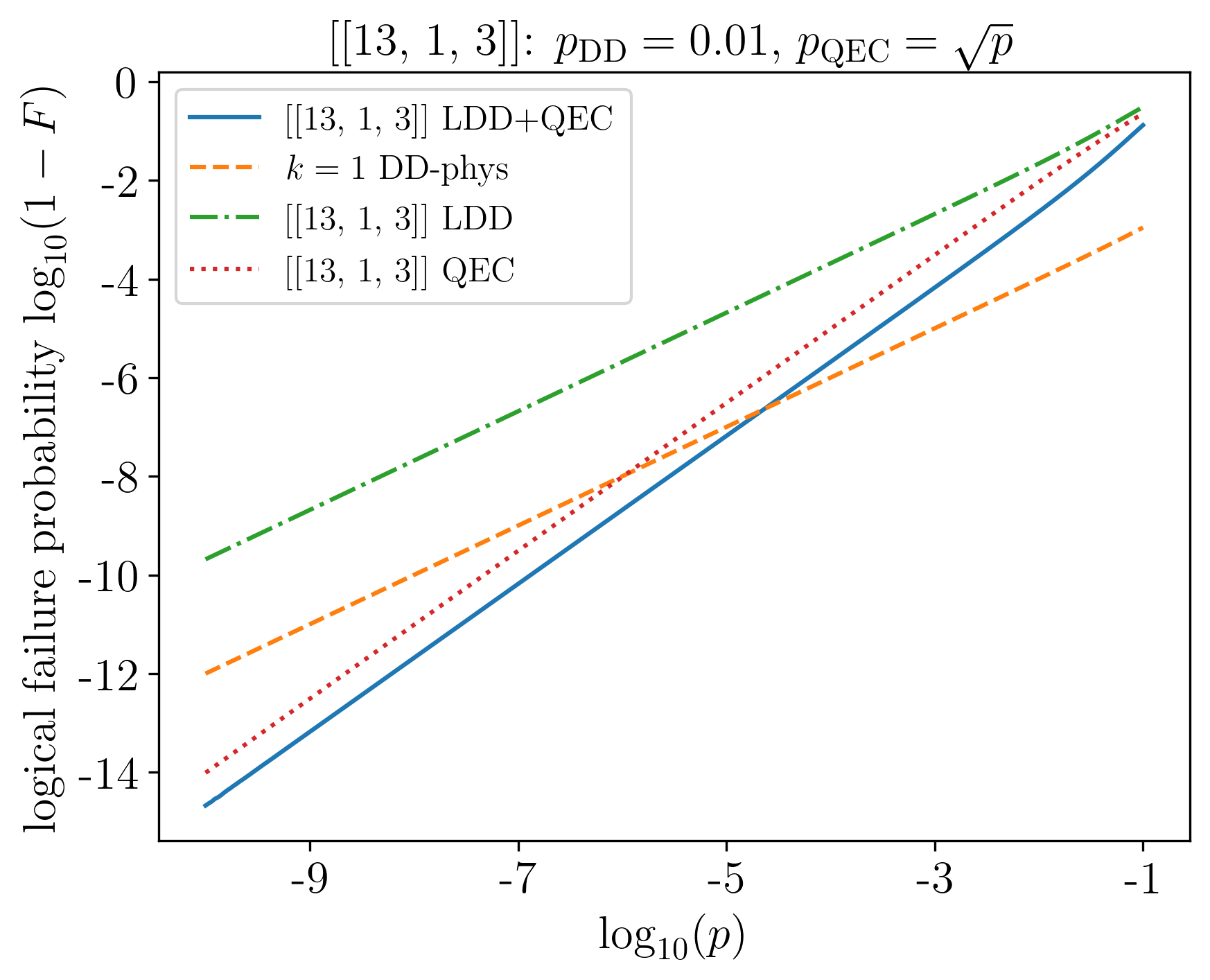}}
  \caption{Logical failure probability for a $[[13,1,3]]$ code with $\pQEC=\sqrt{p}$, comparing physical DD, LDD-only, QEC-only, and the hybrid LDD+QEC protocol.
The results are qualitatively similar to those for the Steane code [see \cref{fig:7_1_3_LDD+QEC_comparing_LDD_generating_sets}(c)].}
  \label{fig:13_1_3_w_pQEC_equal_sqrt_p}
\end{figure}
As we did for the Steane code in \cref{fig:7_1_3_LDD+QEC_comparing_LDD_generating_sets}(c), we now examine the regime in which $\pQEC$ is not negligible relative to $p$.
\Cref{fig:13_1_3_w_pQEC_equal_sqrt_p} illustrates this regime by again setting $\pQEC=\sqrt{p}$.
Recall that in this case, the $\pQEC p$ contribution scales as $p^{3/2}$ and can compete with (or dominate over) the uncorrectable-error contribution $p^\alpha$ when $\alpha\ge 2$.
In this imperfect-recovery setting, LDD can improve performance not only by suppressing uncorrectable errors, but also by reducing the rate of low-weight errors whose impact is amplified by recovery failures.
Consequently, LDD+QEC is again favorable in the low-$p$ regime even though the sufficient condition in \cref{thm:Hyb.vs.QEC-asymptotics} does not hold for the chosen LDD generators.
This setting illustrates that imperfect recovery can change which errors dominate logical failure and can qualitatively modify the comparison between strategies, restoring the hybrid protocol's advantage for sufficiently small $p$.

\begin{figure*}[th]
  \centering
\includegraphics[width=\textwidth]{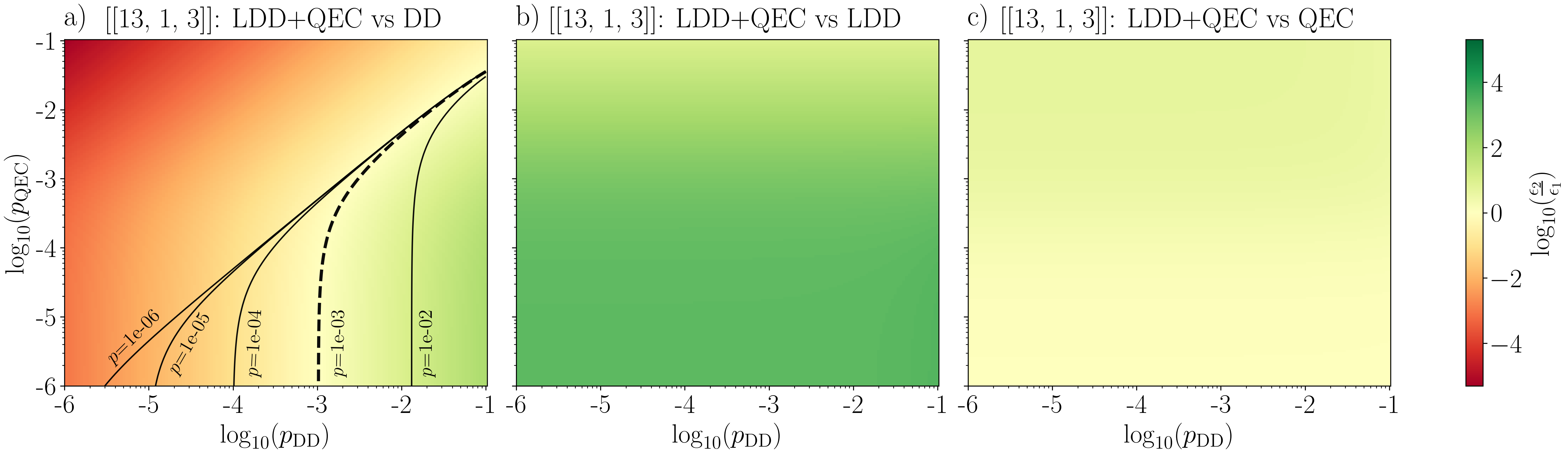}
\caption{Relative improvement metric $R$ for the $[[13,1,3]]$ code, as in \cref{fig:7_1_3_relative_advantage_comparison}.
In contrast to the $[[7,1,3]]$ example, the LDD+QEC advantage region relative to physical DD approaches a limiting boundary as $p$ decreases, reflecting sensitivity to the choice of code and LDD generators. Panels (b) and (c) remain positive throughout the plotted domain, indicating that LDD+QEC outperforms LDD-only and QEC-only for these parameters.}
  \label{fig:13_1_3_relative_advantage_comparison}
\end{figure*}
In analogy to \cref{fig:7_1_3_relative_advantage_comparison}, the relative-improvement metric $R$ is plotted in \cref{fig:13_1_3_relative_advantage_comparison} in the $(\pDD,\pQEC)$ plane at fixed $p=10^{-3}$; recall that $R>0$ indicates that LDD+QEC has lower logical failure probability than the comparator protocol.
In the plotted domain, LDD+QEC dominates both QEC-only and LDD-only across the full $(\pDD,\pQEC)$ range shown, indicating that adding LDD to the encoded cycle is beneficial once recovery imperfections are included for these parameters.
The comparison with DD-phys, however, exhibits a qualitatively different feature: the hybrid advantage region is separated from a DD-phys-dominant region by a boundary that persists as $p$ decreases, rather than sweeping monotonically across the entire plane as in the Steane code case [\cref{fig:7_1_3_relative_advantage_comparison}].
This boundary is captured by the contour lines (which mark $R=0$ for several values of $p$): as $p$ is reduced, the contours move but approach an apparent limiting curve, indicating that there remains a portion of the $(\pDD,\pQEC)$ plane where DD-phys outperforms the encoded hybrid strategy.
In contrast to the $[[7,1,3]]$ example, this behavior signals that for the present $[[13,1,3]]$ code, recovery map, and LDD generators, strengthening DD suppression (smaller $\pDD$) does not translate into an ever-expanding region of hybrid-protocol dominance relative to DD-phys.

The boundary observed when comparing LDD+QEC to DD-phys can be understood by inspecting the low-weight structure of the commutant $\tGp$.
Recall that in our model, DD/LDD suppresses precisely those Pauli errors that \textit{anticommute} with the decoupling group (their probabilities are multiplied by $\pDD$, followed by renormalization), whereas errors in $\tGp$ are left unsuppressed (up to the overall renormalization).
For the Steane-code choice $\tG=\langle X^{\otimes 7},Z^{\otimes 7}\rangle$, every single-qubit Pauli error anticommutes with at least one generator, so all weight-$1$ Pauli errors are suppressed by LDD.
In contrast, for the $[[13,1,3]]$ choice $\tG=\langle \texttt{IIIYYYYZIXYXY},  \texttt{XXXXXXXXXIIII}\rangle$, the commutant contains several weight-$1$ Paulis (e.g., $X_1,X_2,X_3,X_9,X_{10},X_{12},Y_{11},Y_{13}\in\tGp$), which therefore remain essentially unsuppressed by LDD.

This matters in the imperfect-recovery regime.
When $\pQEC>0$, a weight-$1$ error that produces a nontrivial syndrome can contribute to logical failure through a recovery fault: conditioned on a nontrivial syndrome, a decoder failure occurs with probability $\pQEC$ and induces a uniformly random logical Pauli, so the probability of a nontrivial logical fault is $1-4^{-k}$ (equal to $3/4$ for $k=1$).
Consequently, any \textit{unsuppressed} weight-$1$ errors in $\tGp$ generate an $O(\pQEC p)$ contribution to the hybrid logical failure probability that is only weakly dependent on $\pDD$.
By comparison, DD-phys on a single qubit suppresses all nontrivial single-qubit Pauli errors (in this effective model), leading to a logical failure probability that scales as $O(\pDD p)$ at small $p$.
Therefore, for sufficiently small $\pDD$, DD-phys can outperform LDD+QEC in a region determined by the relative coefficients of these linear-in-$p$ contributions, yielding a limiting boundary in the $(\pDD,\pQEC)$ plane as $p$ decreases, rather than an advantage region that expands without bound.

Overall, these results reinforce that hybrid-protocol benefit is not automatic; it depends on whether the chosen LDD group suppresses the error classes that dominate logical failure for the specific code and recovery model, including how recovery imperfections determine which errors are most consequential.

\begin{figure*}[t]
\includegraphics[width=\textwidth]{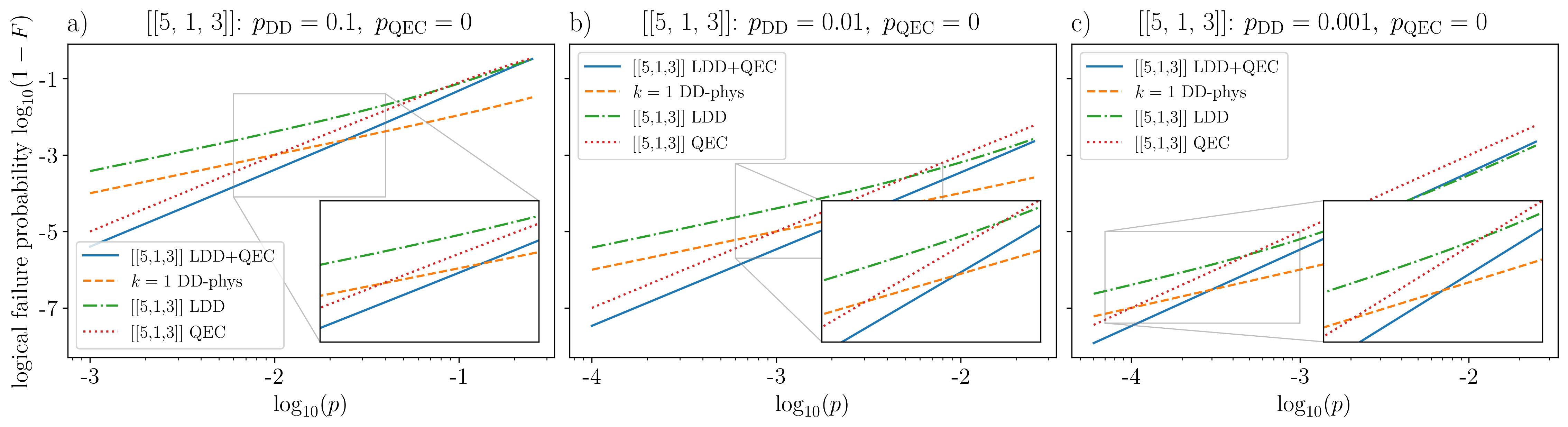} \caption{Logical failure probability for our four strategies using the $[[5,1,3]]$ code and the decoding map of \cref{ss:513-Decoding-Map}, shown as a function of the physical Pauli error probability $p$ for (a) $\pDD=0.1$, (b) $\pDD=0.01$, (c) $\pDD = 0.001$, with perfect recovery ($\pQEC=0$).
For sufficiently small $p$, LDD+QEC achieves the lowest logical failure probability among the encoded strategies, in agreement with \cref{thm:Hyb.vs.QEC-asymptotics}.}
  \label{fig:5_1_3_comparison_Theorem2_setting}
\end{figure*}

\subsection{\texorpdfstring{$[[5,1,3]]$}{[[5,1,3]]} code}

The $[[5,1,3]]$ code~\cite{BennettDiVincenzoSmolinWootters:1996aa,LaflammeMiquelPazZurek:1996aa,ChaoReichardt:2018aa} provides another distance-$3$ testbed for the LDD+QEC protocol.  We take the logical decoupling group to be generated by the standard logical Paulis,
$\langle X_L, Z_L \rangle = \langle XXXXX,  ZZZZZ\rangle$.
For the recovery map used in our numerics (\cref{ss:513-Decoding-Map}), the low-$p$ behavior is consistent with the sufficient-condition regime of \cref{thm:Hyb.vs.QEC-asymptotics}, so the hybrid and QEC-only curves separate in the expected direction as $p\to 0$.  We examine this behavior both in the ideal-recovery regime and in an imperfect-recovery setting.

\subsubsection{Ideal recovery regime \texorpdfstring{($\pQEC=0$)}{pQEC=0}}

\Cref{fig:5_1_3_comparison_Theorem2_setting} illustrates the ideal-recovery regime for the $[[5,1,3]]$ code: we take $\pDD\in\{0.1,0.01,0.001\}$ and set $\pQEC=0$.  As $p\to 0$, the finite distance of the code implies that QEC-based strategies exhibit higher-order scaling than unencoded DD.  In particular, consistent with \cref{thm:Hyb.vs.QEC-asymptotics}, we observe that the hybrid protocol achieves strictly lower logical failure probability than QEC-only in the sufficiently small-$p$ regime (for each fixed $\pDD<1$).  The crossover scale in $p$ at which this separation becomes visible depends on $\pDD$, because $\pDD$ controls the strength of suppression within the LDD-modified sector.

\subsubsection{Imperfect recovery regime \texorpdfstring{($\pQEC>0$)}{pQEC>0}}

We next move beyond the ideal-recovery regime by allowing recovery faults.  As a representative example, \cref{fig:5_1_3_w_pQEC_equal_sqrt_p} sets $\pQEC=\sqrt{p}$ at fixed $\pDD=0.01$.  This choice lies outside the hypothesis $\pQEC=0$ of \cref{thm:Hyb.vs.QEC-asymptotics} and can also fall outside the scaling regime in which QEC-only retains its ideal $O(p^\alpha)$ behavior.  Nevertheless, in the low-noise regime shown, the hybrid protocol remains the best-performing encoded strategy, motivating a systematic sweep over $(\pDD,\pQEC)$.

\Cref{fig:5_1_3_relative_advantage_comparison} maps the pairwise performance using the relative-improvement metric $R$ defined in \cref{eq:relative-advantage}, at fixed $p=10^{-3}$ with contours indicating the sign-change boundary for several values of $p$.  As $p$ decreases, the hybrid-advantage region expands (the contours move toward smaller $\pDD$ and/or smaller $\pQEC$), behavior that is qualitatively similar to the $[[7,1,3]]$ example and contrasts with cases where the boundary approaches a limiting curve as $p\to 0$ (as in the $[[13,1,3]]$ example discussed above).

\begin{figure}[t]
  \hspace{0cm}{\includegraphics[width=\columnwidth]{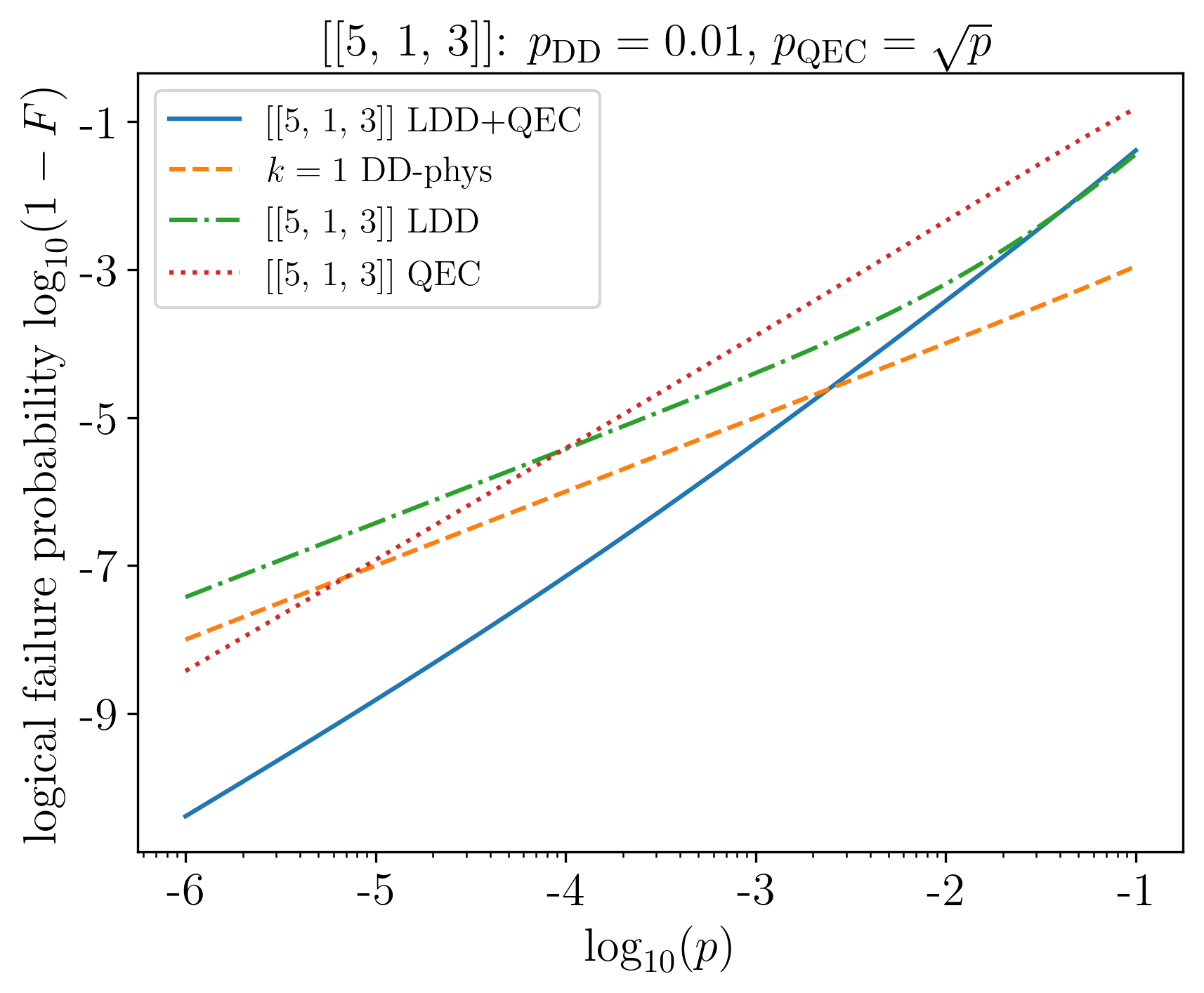}}
  \caption{Logical failure probability for a $[[5,1,3]]$ code with $\pQEC=\sqrt{p}$ at fixed $\pDD=0.01$, comparing physical DD, LDD-only, QEC-only, and the hybrid LDD+QEC protocol.
The hybrid protocol remains the best-performing encoded strategy in the low-noise regime shown, qualitatively similar to the corresponding $[[7,1,3]]$ imperfect-recovery comparison [see \cref{fig:7_1_3_LDD+QEC_comparing_LDD_generating_sets}(c)].}
  \label{fig:5_1_3_w_pQEC_equal_sqrt_p}
\end{figure}

\begin{figure*}[th]
  \centering
\includegraphics[width=\textwidth]{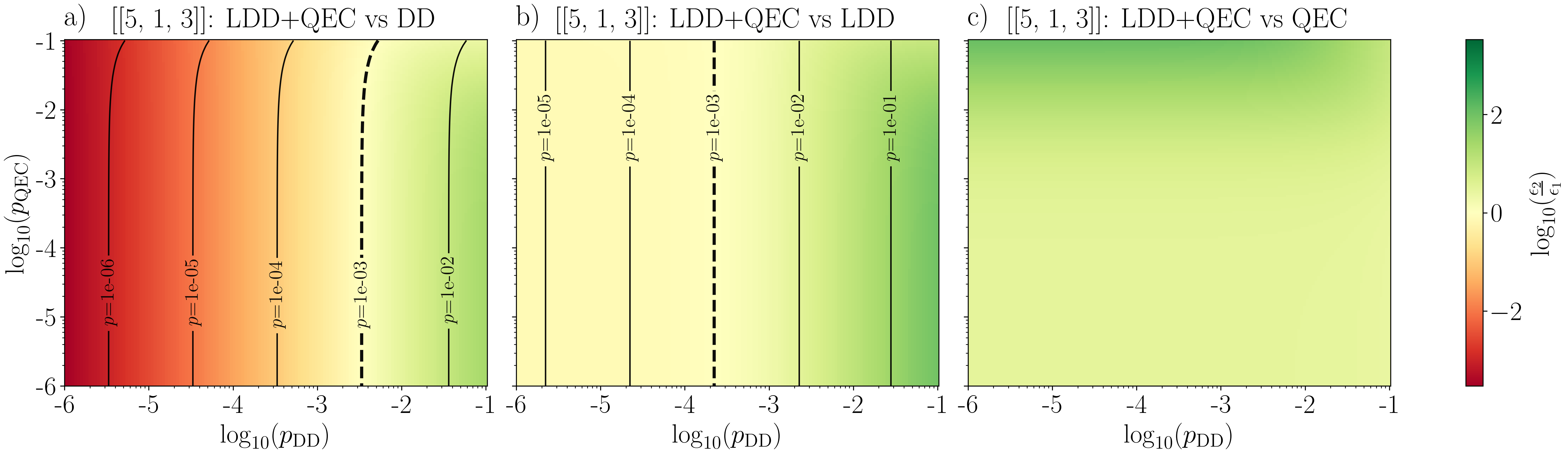}
\caption{Relative improvement metric $R$ for the $[[5,1,3]]$ code, as in \cref{fig:7_1_3_relative_advantage_comparison,fig:13_1_3_relative_advantage_comparison}.
Panels (a) and (b) show that, similarly to the $[[7,1,3]]$ example, the LDD+QEC advantage region relative to both physical DD and LDD-only continues to grow as $p$ decreases, indicating that the $[[5,1,3]]$ code is favorable for the LDD+QEC protocol within the present effective-noise and decoding model. Panel (c) remains positive throughout the plotted domain, indicating that LDD+QEC outperforms QEC-only for these parameters.}
  \label{fig:5_1_3_relative_advantage_comparison}
\end{figure*}

\begin{figure*}[t]
  \centering
\includegraphics[width=\textwidth]{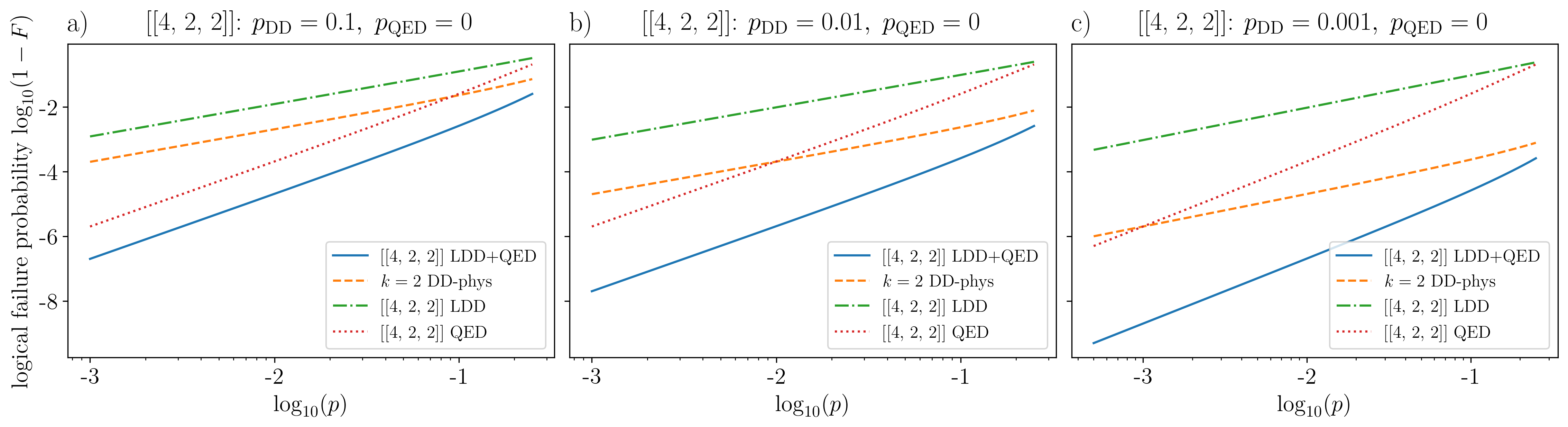}
\caption{Conditional logical failure probability for our four strategies using the $[[4,2,2]]$ code, shown as a function of the physical Pauli error probability $p$ for (a) $\pDD=0.1$, (b) $\pDD=0.01$, (c) $\pDD=0.001$, with perfect error detection ($\pQED=0$).
For sufficiently small $p$, LDD+QED achieves the lowest logical failure probability among the encoded strategies, in agreement with \cref{thm:QED.Hyb.vs.QED}.}
  \label{fig:4_2_2_comparison_Theorem2_setting}
\end{figure*}

\subsection{\texorpdfstring{$[[4,2,2]]$}{[[4,2,2]]} code (QED case)}
\label{sec:422}

We now turn to the QED setting and use the $[[4,2,2]]$ code as a testbed, motivated both by prior four-qubit error-detection experiments~\cite{CorcolesMagesanSrinivasanCrossSteffenGambettaChow:2015aa,Pokharel2024npj} and by the LDD+QED results of Ref.~\cite{vezvaee2025demonstrationhighfidelityentangledlogical}. Its stabilizer generators are $S=\langle XXXX, ZZZZ\rangle$, and we use the logical generators $\langle X_{L_1}, X_{L_2}, Z_{L_1}, Z_{L_2}\rangle=\langle XIIX, IIXX, IIZZ, ZIIZ\rangle$ to define the LDD sequence throughout.
We examine both the ideal- and imperfect-detection regimes in order to assess how this code interacts with LDD+QED within and beyond the hypotheses of \cref{thm:QED.Hyb.vs.QED}.

\subsubsection{Ideal detection regime \texorpdfstring{($\pQED=0$)}{pQED=0}}

\cref{fig:4_2_2_comparison_Theorem2_setting} compares the conditional logical failure probability (i.e., $1-F$ conditioned on acceptance) for the $[[4,2,2]]$ code in the perfect-detection setting $\pQED=0$, for $\pDD\in\{0.1,0.01,0.001\}$. Varying $\pDD$ primarily shifts the crossover locations between DD-based and QED-based strategies by changing the strength of suppression in the LDD-modified sector. In the low-noise regime, the numerical curves show that LDD+QED improves upon QED-only, consistent with the analytical comparison criterion of \cref{thm:QED.Hyb.vs.QED}.
In the plotted range, it also attains the smallest logical failure probability among the four strategies.

Because we are in the QED setting, one must also track the acceptance probability $P_A$. For DD-phys and LDD-only, $P_A=1$ because there is no postselection. For QED-only and LDD+QED, the expected number of trials required to obtain an accepted run is $1/P_A$. This acceptance-vs-restart trade-off is also central in more general schemes that combine postselection with active correction~\cite{PrabhuReichardt:2024aa}. \Cref{fig:4_2_2_probability_of_acceptance_comparison} shows that $P_{A,\mathrm{Hyb}}$ and $P_{A,\mathrm{QED}}$ both approach $1$ as $p\to 0$, with $P_{A,\mathrm{Hyb}}$ approaching $1$ more rapidly; the dependence on $\pDD$ is comparatively weak. Thus the shot-overhead penalty of postselection becomes mild in the low-noise regime and is smaller for the hybrid protocol than for QED-only in the plotted parameter range.

\begin{figure*}[t]
  \centering
  \includegraphics[width=\textwidth]{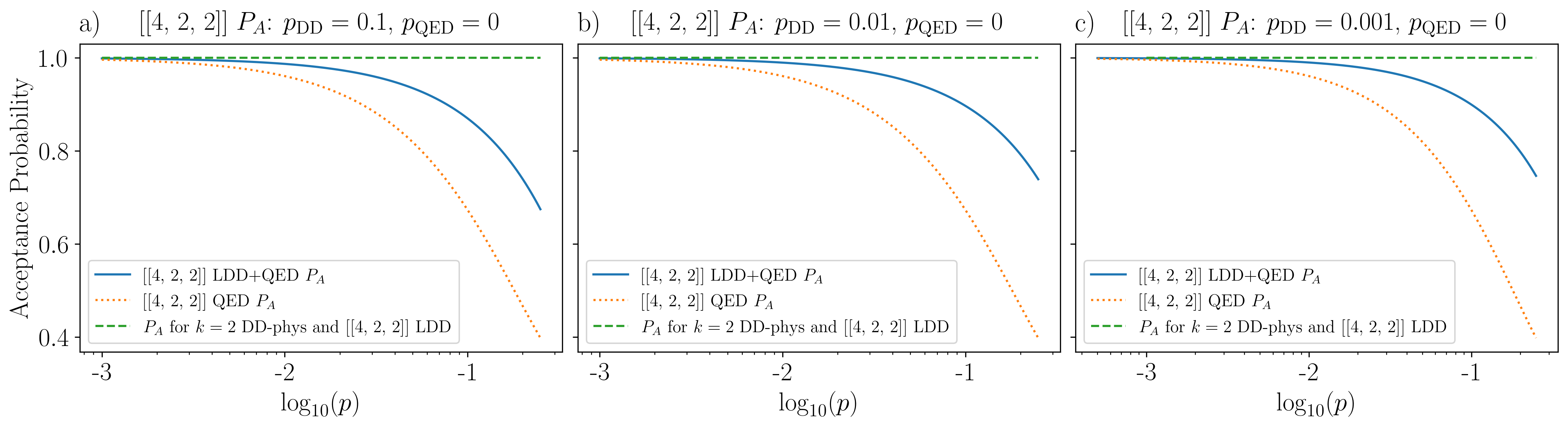}
  \caption{
    Acceptance probability $P_A$ for the four strategies using the $[[4,2,2]]$ code, shown as a function of the physical Pauli error probability $p$ for (a) $\pDD=0.1$, (b) $\pDD=0.01$, (c) $\pDD=0.001$, with perfect error detection ($\pQED=0$).
    The non-postselected strategies have $P_A=1$ by definition, while the QED-based strategies have $P_A<1$ due to postselection on a trivial syndrome.
  The expected number of trials required to obtain an accepted run is $1/P_A$, and this overhead approaches unity as $p\to 0$.}
  \label{fig:4_2_2_probability_of_acceptance_comparison}
\end{figure*}

\subsubsection{Imperfect detection regime \texorpdfstring{($\pQED>0$)}{pQED>0}}
We next relax the perfect-detection assumption and allow $\pQED>0$.
As a representative example, \cref{fig:4_2_2_w_pQED_equal_sqrt_p} sets $\pQED=\sqrt{p}$ at fixed $\pDD=0.01$.
This lies outside the hypotheses of \cref{thm:QED.Hyb.vs.QED}.
In the plotted range, the hybrid protocol remains favorable in the low-noise regime, although DD-phys becomes competitive again at larger $p$.

\cref{fig:4_2_2_relative_advantage_comparison} maps the relative-improvement metric $R$ in the $(\pDD,\pQED)$ plane at fixed $p=10^{-3}$.
Panels (b) and (c) show that LDD+QED outperforms LDD-only and QED-only throughout the plotted domain. The comparison with DD-phys in panel (a) is qualitatively different: the advantage boundary is approximately diagonal and does not sweep across the entire plane as $p$ decreases.
This behavior is consistent with the fact that imperfect detection can reintroduce leading-order ($O(p)$) failure mechanisms through missed nontrivial syndromes, so that both DD-phys and the hybrid protocol can share the same leading scaling in $p$ and the ordering is then controlled by relative prefactors.

\begin{figure}[t]
  \centering
  \includegraphics[width=\columnwidth]{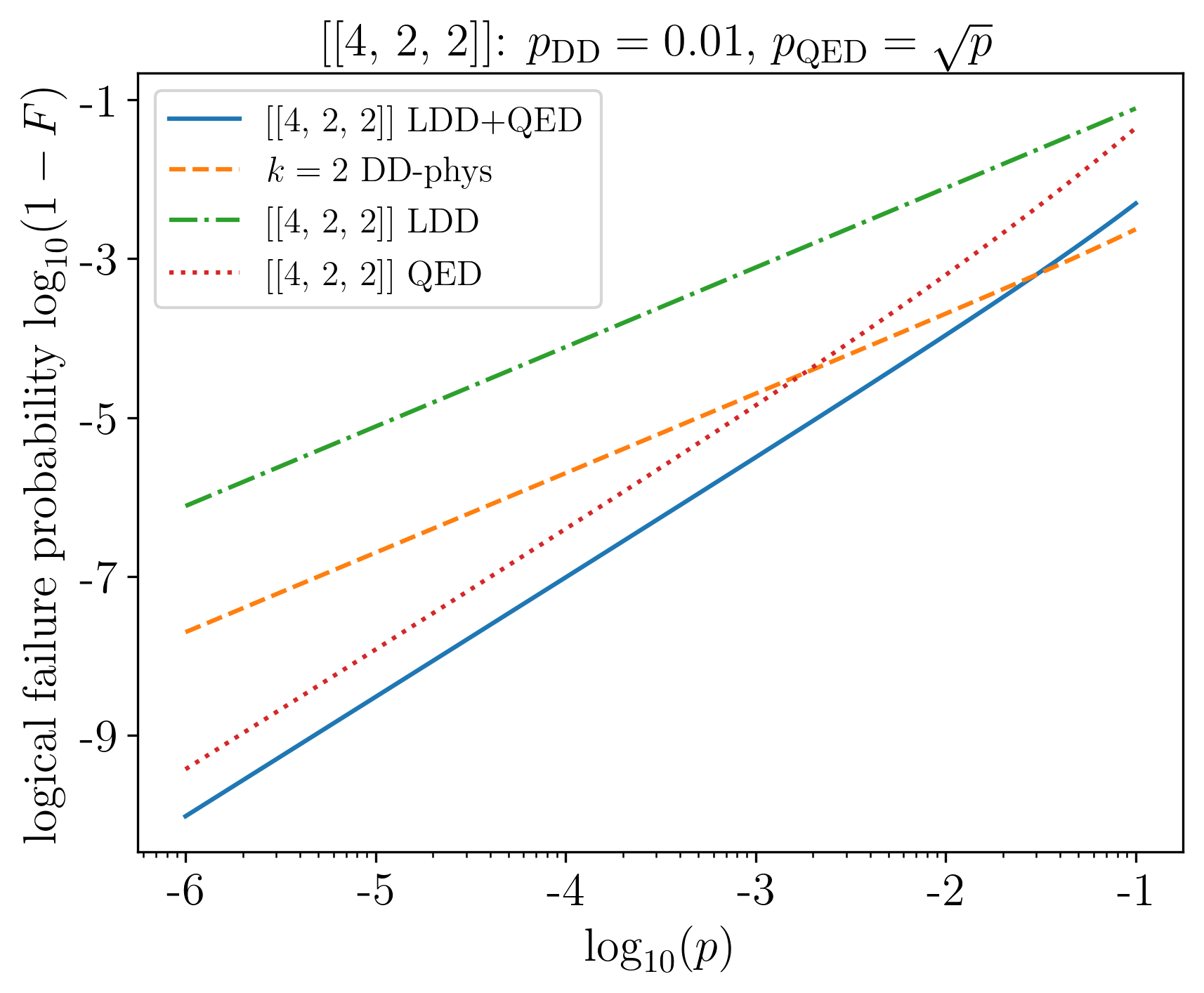}
  \caption{Logical failure probability for a $[[4,2,2]]$ code with $\pQED=\sqrt{p}$ at fixed $\pDD=0.01$, comparing physical DD, LDD-only, QED-only, and the hybrid LDD+QED protocol.
The hybrid protocol is favorable in the low-noise regime shown, while DD-phys becomes competitive again at larger $p$.}
  \label{fig:4_2_2_w_pQED_equal_sqrt_p}
\end{figure}

\begin{figure*}[th]
  \centering
\includegraphics[width=\textwidth]{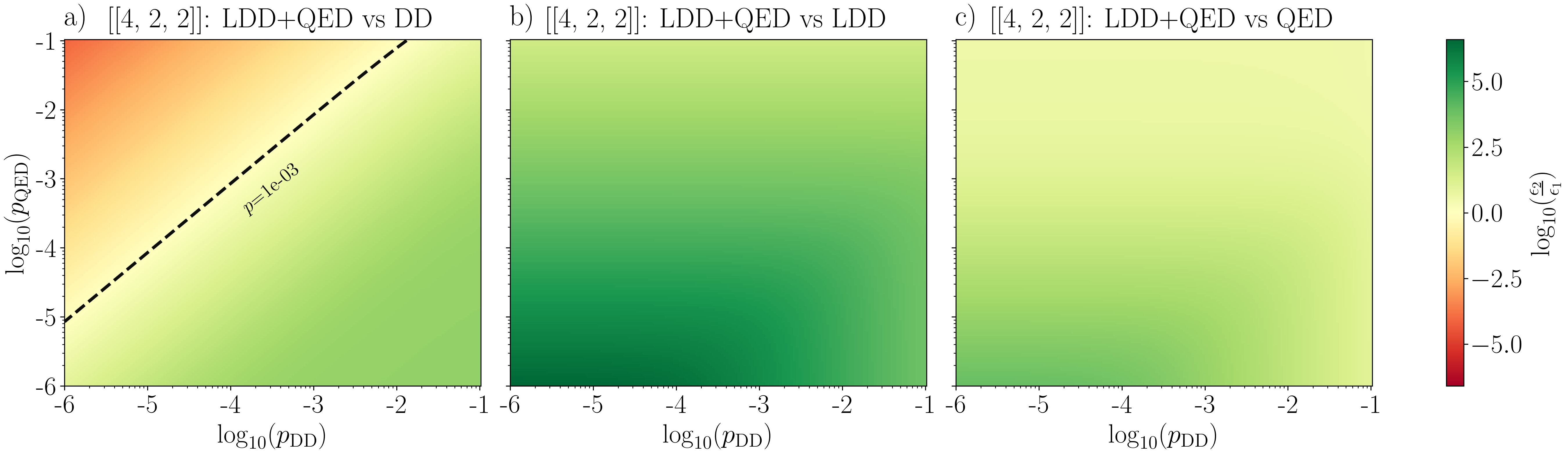}
\caption{Relative improvement metric $R=\log_{10}(\epsilon_{\mathrm{comp}}/\epsilon_{\mathrm{hyb}})$ for the $[[4,2,2]]$ code, comparing LDD+QED against (a) DD-phys, (b) LDD-only, and (c) QED-only. The contour in panel (a) marks the sign-change boundary $R=0$ for $p=10^{-3}$. Panels (b) and (c) remain positive throughout the plotted domain, indicating that LDD+QED outperforms LDD-only and QED-only for these parameters.
 }
\label{fig:4_2_2_relative_advantage_comparison}
\end{figure*}

\begin{figure*}[th]
  \centering
\includegraphics[width=\textwidth]{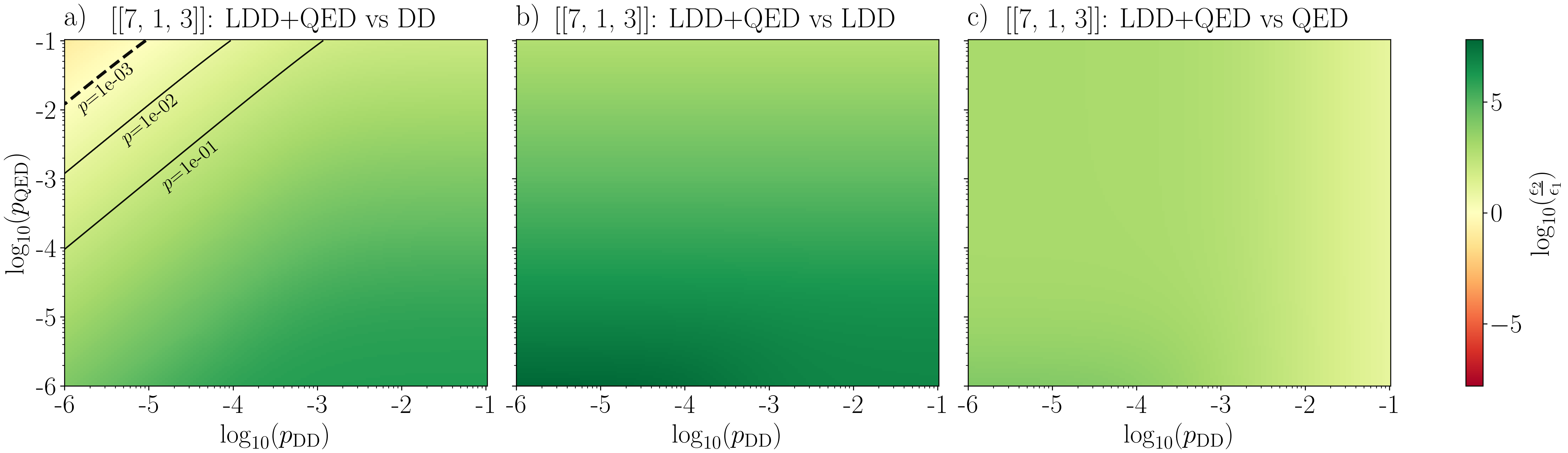}
\caption{
Relative improvement metric $R=\log_{10}(\epsilon_{\mathrm{comp}}/\epsilon_{\mathrm{hyb}})$ for the $[[7,1,3]]$ code in the QED setting, comparing LDD+QED against (a) DD-phys, (b) LDD-only, and (c) QED-only.
The contours in panel (a) mark the sign-change boundary $R=0$ for several values of $p$, while panels (b) and (c) remain positive throughout the plotted domain.}
\label{fig:7_1_3_relative_LDD-QED_advantage_comparison}
\end{figure*}

\subsection{\texorpdfstring{$[[7,1,3]]$}{[[7,1,3]]} code in the QED setting}
\label{ss:713-LDD+QED-Setting}

To test whether the saturation seen in the $[[4,2,2]]$ example is generic, we also evaluate the $[[7,1,3]]$ Steane code in the QED setting, using the same Steane-code LDD choice as in \cref{sec:713}. As in the preceding QED discussion, the comparison is made through the relative-improvement metric $R$ defined in \cref{eq:relative-advantage}, i.e., through the conditional logical failure probability; acceptance-probability considerations are the same as for the other postselected QED strategies and are not replotted here.

As shown in \cref{fig:7_1_3_relative_LDD-QED_advantage_comparison}, the comparison with DD-phys is qualitatively different from the $[[4,2,2]]$ case: as $p$ decreases, the LDD+QED advantage region continues to expand rather than approaching a limiting boundary.
Panels (b) and (c) remain favorable throughout the plotted domain, so for these parameters the hybrid protocol also outperforms LDD-only and QED-only.
This provides additional numerical evidence that the QED hybrid-advantage region is code- and LDD-dependent rather than universal, and complements \cref{thm:QED.Hyb.vs.QED} by probing the imperfect-detection regime $\pQED>0$.

\section{Summary}
\label{sec:summary}

This work develops a unified weight-enumerator-polynomial-based framework for analyzing how logical dynamical decoupling combines with both QEC and QED in a single memory cycle.
For QEC, the framework yields exact logical-fidelity expressions for DD-phys, QEC-only, LDD-only, and LDD+QEC.
For QED, it yields exact formulas for both the logical fidelity conditioned on acceptance and the acceptance probability of QED-only and LDD+QED, together with the corresponding DD-phys and LDD-only baselines.

In the perfect-recovery QEC setting, \cref{thm:Hyb.vs.QEC} shows that LDD+QEC outperforms QEC-only precisely when the suppressed sector contains a larger fraction of uncorrectable errors than the unsuppressed sector.
The low-$p$ analysis sharpens this into a practical design rule: if the LDD group suppresses at least one minimum-weight uncorrectable error, then the hybrid protocol wins for sufficiently small $p$, and stabilizer dressing can often be used to enforce this condition.
In the QED setting, the relevant trade-off is two-dimensional, involving both the logical fidelity conditioned on acceptance and the acceptance probability $P_A$; \cref{thm:QED.Hyb.vs.QED} gives the corresponding ideal-readout comparison criterion for LDD+QED versus QED-only.

In \cref{sec:numerics} we evaluate the analytic formulas for representative QEC and QED examples to illustrate the theory and to probe regimes beyond the clean asymptotic assumptions.
For the $[[7,1,3]]$ Steane code with the decoding map specified in \cref{ss:713-Decoding-Map}, we find that in the ideal-recovery regime ($\pQEC=0$) the hybrid LDD+QEC strategy outperforms QEC-only, LDD-only, and physical DD as $p\to 0$ for the tested suppression values, consistent with \cref{thm:Hyb.vs.QEC-asymptotics}.
When $\pQEC$ is allowed to vary (including settings that violate $\pQEC=o(p^{\alpha-1})$), heat-map comparisons in the $(\pDD,\pQEC)$ plane show that the region where LDD+QEC has lower logical failure probability than competing strategies expands as $p$ decreases, although very strong physical DD can remain competitive when $\pDD$ is extremely small.
The numerics additionally highlight that the best LDD generating set depends on both the physical noise strength and the QEC imperfection level: when QEC is imperfect, suppressing low-weight errors can be especially beneficial, whereas in the near-perfect-QEC regime the main gains come from suppressing the higher-weight uncorrectable errors.

The $[[5,1,3]]$ code shows qualitatively similar QEC behavior.
With the decoding map of \cref{ss:513-Decoding-Map}, the hybrid and QEC-only curves separate in the expected direction in the ideal-recovery regime, again consistent with the sufficient low-$p$ mechanism of \cref{thm:Hyb.vs.QEC-asymptotics}.
In imperfect-recovery comparisons, the hybrid protocol remains the best-performing encoded strategy in the low-noise regime shown, and the region of hybrid advantage in the $(\pDD,\pQEC)$ plane expands as $p$ decreases.
This indicates that the favorable hybrid ordering is not unique to the Steane code but also appears in a distinct distance-$3$ example with a different decoding map.

By contrast, the $[[13,1,3]]$ example is deliberately chosen so that the sufficient low-$p$ condition $\beta=\alpha$ from \cref{thm:Hyb.vs.QEC-asymptotics} fails: the selected LDD generators do not suppress any minimum-weight uncorrectable error.
In that case, the hybrid protocol can lose its asymptotic separation from QEC-only, and in the $(\pDD,\pQEC)$ plane its advantage relative to physical DD approaches a limiting boundary rather than expanding across the full plane as $p\to 0$.
This provides a concrete reminder that hybrid advantage is not automatic, but depends on how well the chosen LDD group is matched to the code and decoding map.

In the QED setting, the $[[4,2,2]]$ code illustrates both the benefit of combining LDD with postselected detection and the importance of tracking acceptance as well as conditional fidelity.
In the ideal-detection regime ($\pQED=0$), LDD+QED improves upon QED-only and, in the plotted range, attains the smallest conditional logical failure probability among the four compared strategies; at the same time, the acceptance probabilities of the QED-based strategies approach unity as $p\to 0$, with the hybrid protocol doing so more rapidly.
In the imperfect-detection regime, LDD+QED continues to outperform LDD-only and QED-only in conditional logical failure probability throughout the plotted $(\pDD,\pQED)$ domain, but its comparison with physical DD is more delicate: missed nontrivial syndromes reintroduce leading-order $O(p)$ failure mechanisms, so the advantage boundary relative to DD-phys appears to approach a limiting curve controlled by relative prefactors rather than sweeping across the full plane as $p\to 0$.
Thus, in the detection setting as well, whether hybrid advantage expands or saturates is code- and LDD-dependent.

These results can be summarized as the statement that DD and QEC/QED are ``better together''
when one co-designs the code, the decoding/detection map, and the LDD group so that suppression is concentrated on the errors that actually control logical failure and, in the detection setting, acceptance.

There are several natural directions to build on this work.
First, it would be valuable to connect microscopic noise models
(including biased noise and coherent/non-Pauli components)
to the effective parameters $(p,\pDD,\pQEC/\pQED)$ used here,
and to incorporate control-induced errors so that $\pDD$
reflects both suppression and pulse imperfections.
In particular, one can simulate pulse-level dynamics
for a specific hardware platform to do so.
Second, the WEP-based formulas provide a fast evaluation tool that can be used to optimize
over LDD generating sets and decoding choices for a given hardware noise profile.
Third, one can optimize fine-grained decisions abstracted away in our effective model
for a specific hardware implementation or simulation
and observe how they impact the values of the effective parameters $(p,\pDD,\pQEC/\pQED)$.
These include, e.g., the choice of pulse count, DD sequence (for a given DD group),
pulse shapes, and timing.
Finally, extending the analysis from a single memory cycle to multi-round fault-tolerant protocols (where syndrome extraction and control must be interleaved repeatedly) would clarify how the advantages identified here translate into end-to-end logical performance and resource tradeoffs in realistic architectures.

\begin{acknowledgments}
Research was sponsored by IARPA and the Army Research Office, under the Entangled Logical Qubits program, and was accomplished under Cooperative Agreement Number W911NF23-2-0216. The views and conclusions contained in this document are those of the authors and should not be interpreted as representing the official policies, either expressed or implied, of IARPA, the Army Research Office, or the U.S. Government. The U.S. Government is authorized to reproduce and distribute reprints for Government purposes notwithstanding any copyright notation herein. This material is also based upon work supported by, or in part by, the U.S. Army Research Laboratory and the U.S. Army Research Office under contract/grant number W911NF2310255.
\end{acknowledgments}

\appendix

\section{Alternative model for DD suppression}
\label{app:alt-DD-model}

The model of DD suppression used in the main text,
which we call here the ``renormalized'' (ren) model,
is one of many possible effective models
that one can attempt to fit to a given microscopic noise-and-control scenario.
For example, one may consider replacing the renormalization procedure we used
with an assumption that errors suppressed by DD or LDD are instead ``transferred'' to the identity sector, rather than redistributed over the remaining errors.
In such an ``identity-transfer'' (IT) model,
$\FHyb\ge \FQEC$ is built in by construction,
so the model provides an optimistic benchmark for the hybrid fidelity.
In fact, for fixed sector weights one finds $\FHyb^{\mathrm{ren}}\le \FHyb^{\rm IT}$. More generally, one can consider a whole family of redistribution rules,
and the renormalized sector-reweighting model is just one such choice.
It is more conservative than the IT model, but it is not a worst-case bound.
Some degree of pessimism is warranted because there is no \textit{a priori} reason for all probability mass removed from the DD-suppressed sector to flow to the identity sector.
Determining the redistribution rule appropriate to a given microscopic noise-and-control model would require a much more detailed first-principles derivation
than is attempted here and lies outside the scope of our effective-model analysis,
and we expect the details of the redistribution to be case-dependent.
This motivates, for present purposes, the use of a model that is analytically tractable and yields a genuinely nontrivial hybrid-vs-QEC comparison, namely the renormalized model.

We next present a formalization of these considerations.
First note that in the renormalized model, the possibility that LDD+QEC performs worse than QEC-only is already present in the ideal-recovery regime $\pQEC=0$ and is not caused by imperfect DD ($\pDD>0$).
DD partitions the Pauli errors into a DD-unsuppressed sector
\begin{equation}
  \mathrm{U}:=\tGp=\bigl\{E'\in\tmP_n: [E',g']=0\ \ \forall g'\in\tG\bigr\},
\end{equation}
and a DD-suppressed sector $\mathrm{S}:=\tmP_n\setminus\tGp$,
while the decoder partitions them into correctable and uncorrectable errors.
\Cref{thm:Hyb.vs.QEC} shows that the sign of the fidelity difference $\FHyb-\FQEC$
is determined entirely by how the uncorrectable set is distributed across these two DD sectors.
In particular, for any fixed $\pDD\in[0,1)$, the hybrid protocol helps precisely when the DD-suppressed sector contains a larger fraction of uncorrectable errors than the DD-unsuppressed sector, and hurts when the opposite is true.
Thus the role of $\pDD$ in the regime of \Cref{thm:Hyb.vs.QEC} is only to determine how strongly the DD-controlled interval reweights the two sectors; it does not determine the sign of $\FHyb-\FQEC$, and this is true even for $\pDD=0$.

Under the ``identity-transfer'' (IT) redistribution rule, by contrast, monotonic improvement is built in by construction. To see this, let us introduce a reference effective Pauli distribution
\begin{equation}
\eref:\tmP_n\to[0,1],
\qquad
\sum_{E'\in\tmP_n}\eref(E')=1,
\end{equation}
which in our model is the effective distribution corresponding to the case $\pDD=1$.
According to \cref{eq:error-model},
\begin{equation}
\eref(E')=(1-p)^{n-\wt(E')}\Bigl(\frac{p}{3}\Bigr)^{\wt(E')}.
\end{equation}
Define the reference weights of the unsuppressed and suppressed sectors by
\begin{equation}
w_{\mathrm{U,ref}}:=\sum_{E'\in \mathrm{U}}\eref(E'),
\qquad
w_{\mathrm{S,ref}}:=\sum_{E'\in \mathrm{S}}\eref(E'),
\end{equation}
so that
\begin{equation}
w_{\mathrm{U,ref}}+w_{\mathrm{S,ref}}=1.
\end{equation}

For notational simplicity, set $\lambda:=\pDD\in[0,1]$, and define the intermediate distribution
\begin{equation}
\label{eq:widetildeP}
\widetilde P_\lambda(E')=
\begin{cases}
\eref(E'), & E'\in \tGp,\\[3pt]
\lambda \eref(E'), & E'\in \tmP_n \setminus \tGp.
\end{cases}
\end{equation}
Its total mass is
\begin{equation}
Z_\lambda=\sum_{E'\in\tmP_n}\widetilde P_\lambda(E')
=
w_{\mathrm{U,ref}}+\lambda w_{\mathrm{S,ref}}
=
1-(1-\lambda)w_{\mathrm{S,ref}},
\end{equation}
so the removed mass is
\begin{equation}
m_\lambda:=1-Z_\lambda=(1-\lambda)w_{\mathrm{S,ref}}.
\end{equation}
The identity-transfer model is then
\begin{equation}
\label{eq:app-delta-model}
P_\lambda^{\rm IT}(E')
=
\widetilde P_\lambda(E')+m_\lambda \delta_{E',\Id},
\end{equation}
where $\delta_{E',\Id}$ is the Kronecker delta on $\tmP_n$. By contrast, the renormalized model used in the main text is the renormalized sector-reweighting model
\begin{equation}
\label{eq:app-ren-model}
P_\lambda^{\mathrm{ren}}(E')
=
\frac{\widetilde P_\lambda(E')}{Z_\lambda}
=
\frac{\widetilde P_\lambda(E')}{w_{\mathrm{U,ref}}+\lambda w_{\mathrm{S,ref}}}.
\end{equation}
Using $Z_\lambda=1-m_\lambda$, this can be rewritten as
\begin{equation}
\label{eq:app-ren2-model}
P_\lambda^{\mathrm{ren}}(E')
=
\widetilde P_\lambda(E')+m_\lambda P_\lambda^{\mathrm{ren}}(E'),
\end{equation}
which will be useful below.

In the correction setting, let $M$ denote the random phase-stripped Pauli error produced during the DD-controlled wait interval with parameter $\lambda=\pDD$ under the hybrid protocol with the identity-transfer rule. By definition,
\begin{equation}
\Pr(M=E')=P_\lambda^{\rm IT}(E'),
\qquad E'\in\tmP_n.
\end{equation}
Let $L$ denote the event that the full hybrid correction cycle ends in logical failure, so that
\begin{equation}
1-\FHyb^{\rm IT}=\Pr(L).
\end{equation}
Now define
\begin{equation}
b(E'):=\Pr(L\mid M=E'),
\end{equation}
i.e., $b(E')$ is the conditional probability that the remainder of the protocol
(syndrome measurement, recovery, and any recovery-fault mechanism) produces a logical failure, given that the memory-interval error was $E'$.
In particular, $b(E')$ depends only on the subsequent correction stage conditioned on $E'$, and not on the redistribution rule used to assign the probabilities of the different $E'$.
In the ideal-recovery case, $b(E')$ is simply the indicator function of whether $E'$ is uncorrectable, while for imperfect recovery it also incorporates the conditional failure probability coming from the recovery step. Applying the law of total probability with respect to the random variable $M$, we obtain
\begin{multline}
\label{eq:1-FhybIT}
1-\FHyb^{\rm IT}=\Pr(L)
=
\sum_{E'\in\tmP_n}\Pr(M=E') \Pr(L\mid M=E') \\
=
\sum_{E'\in\tmP_n}P_\lambda^{\rm IT}(E') b(E').
\end{multline}
Define
\begin{equation}
B_{\mathrm U}:=\sum_{E'\in \mathrm U}\eref(E') b(E'),
\qquad
B_{\mathrm S}:=\sum_{E'\in \mathrm S}\eref(E') b(E').
\end{equation}
Assuming $b(\Id)=0$ (this is, indeed, the case for the QEC error model considered in the main text; see \cref{ss:protocol-and-error-model}),
it follows by combining \cref{eq:widetildeP,eq:app-delta-model,eq:1-FhybIT} that
\begin{equation}
  1-\FHyb^{\mathrm{IT}}
  =
  B_{\mathrm{U}}+\lambda B_{\mathrm{S}}.
\end{equation}
For QEC-only, we have
\begin{equation}
  1-\FQEC=B_{\mathrm{U}}+B_{\mathrm{S}},
\end{equation}
and therefore
\begin{equation}
\FHyb^{\rm IT}\ge \FQEC
\qquad \forall\lambda\in[0,1].
\end{equation}
Thus, under the IT model, all regions in which the hybrid protocol is worse than QEC-only disappear.
The corresponding QED statement is analogous: after decomposing the acceptance and accepted-error weights into their unsuppressed and suppressed contributions,
one again finds that the IT model is monotone-improving relative to QED-only.
We omit the algebra here because it is parallel to the QEC calculation above.

Next, we compare the IT model directly to the renormalized model. Under the renormalized model, \cref{eq:1-FhybIT} applies again, but in the form
\begin{equation}
\label{eq:1-Fhybren}
1-\FHyb^{\rm ren} = \sum_{E'\in\tmP_n}P_\lambda^{\rm ren}(E') b(E'),
\end{equation}
so that combining this with \cref{eq:widetildeP,eq:app-ren-model} now yields
\begin{equation}
1-\FHyb
\equiv
1-\FHyb^{\mathrm{ren}}
=
\frac{B_{\mathrm U}+\lambda B_{\mathrm S}}{w_{\mathrm{U,ref}}+\lambda w_{\mathrm{S,ref}}}.
\end{equation}
Since
$w_{\mathrm{U,ref}}+\lambda w_{\mathrm{S,ref}}\le 1$,
we immediately obtain
\begin{equation}
\FHyb^{\mathrm{ren}}\le \FHyb^{\rm IT}.
\end{equation}
So, between these two redistribution models, the IT rule is an optimistic best case for the hybrid fidelity, while the renormalized model is the more conservative one.

To situate the renormalized model, consider the following general redistribution framework. Let $R$ be an arbitrary probability distribution on $\tmP_n$:
\begin{equation}
R:\tmP_n\to[0,1],
\qquad
\sum_{E'\in\tmP_n}R(E')=1.
\end{equation}
Generalizing the IT and renormalized models [\cref{eq:app-delta-model,eq:app-ren2-model}, respectively], we may then define a general redistributed effective model by
\begin{equation}
P_\lambda^{R}(E')
=
\widetilde P_\lambda(E')+m_\lambda R(E').
\end{equation}
In the correction setting, the corresponding hybrid logical-failure probability is
\begin{equation}
1-\FHyb^{R}
=
\sum_{E'\in\tmP_n}P_\lambda^{R}(E') b(E')
=
B_{\mathrm U}+\lambda B_{\mathrm S}+m_\lambda r_{\mathrm{bad}},
\end{equation}
where
\begin{equation}
r_{\mathrm{bad}}
:=
\sum_{E'\in\tmP_n}R(E') b(E').
\end{equation}
Since $r_{\mathrm{bad}}$ is a convex combination of the values of the conditional logical-failure probability $b(E')$, we have
\begin{equation}
0 = \min_{E'\in\tmP_n} b(E')
\le
 r_{\mathrm{bad}}
\le
\max_{E'\in\tmP_n} b(E') \le 1.
\end{equation}

In particular, under the same assumption that $b(\Id)=0$, the IT model corresponds to a best-case choice
\begin{equation}
R(E')=\delta_{E',\Id},
\qquad\Longrightarrow\qquad
r_{\mathrm{bad}}=0,
\end{equation}
so that
\begin{equation}
1-\FHyb^{\rm IT}=B_{\mathrm U}+\lambda B_{\mathrm S}.
\end{equation}
At the opposite extreme, if the redistributed mass is placed entirely on events
that maximize $b(E')$, then $r_{\mathrm{bad}}=\max_{E'}b(E') =: r_{\mathrm{bad}}^{\mathrm{worst}}$, which yields the most pessimistic redistribution allowed by this framework, and for which
\begin{equation}
1-\FHyb^{\mathrm{worst}} = B_{\mathrm U}+\lambda B_{\mathrm S}+m_\lambda r_{\mathrm{bad}}^{\mathrm{worst}} .
\end{equation}

The renormalized model is neither of these extremes.
Indeed, from \cref{eq:app-ren2-model}, the renormalized model corresponds to
\begin{equation}
R(E')=P_\lambda^{\mathrm{ren}}(E').
\end{equation}
Therefore
\begin{equation}
r_{\mathrm{bad}}^{\mathrm{worst}} \ge r_{\mathrm{bad}}^{\mathrm{ren}}
=
\sum_{E'\in\tmP_n}P_\lambda^{\mathrm{ren}}(E') b(E')
=
1-\FHyb^{\mathrm{ren}} .
\end{equation}

Thus, as claimed, the renormalized model is an intermediate choice of redistribution rule:
\begin{equation}
\FHyb^{\mathrm{worst}} \le \FHyb^{\mathrm{ren}}\le \FHyb^{\rm IT} ,
\end{equation}
i.e., it is situated between the best-case IT model and the most pessimistic possible redistribution. Moreover, it is a simple analytically tractable effective model in which the hybrid-vs-QEC comparison remains nontrivial.

Finally, the corresponding QED statement is analogous:
after decomposing the acceptance and accepted-error weights into their unsuppressed and suppressed contributions, we again find similar inequalities.
We omit the algebra here because it is parallel to the QEC calculation above.

\section{Equivalence of our fidelity measure to the entanglement fidelity}
\label{app:ent-fid}

Recall the standard definition of the entanglement fidelity $F_{\rm e}$ for a CPTP map $\Lambda(\rho)=\sum_{\alpha}K_{\alpha}\rho K_{\alpha}^{\dagger}$ acting on a $d$-dimensional system:
\begin{subequations}
\begin{align}
  F_{\rm e}(\Lambda)&\coloneqq
  \bra{\Phi}(\Lambda\otimes\operatorname{id})\bigl(\ketbra{\Phi}\bigr)\ket{\Phi}\\
  &= \frac{1}{d^2} \sum_\alpha \Bigl| \Tr \bigl(K_\alpha\bigr) \Bigr|^{2},
\end{align}
\end{subequations}
where $\ket{\Phi}=\frac{1}{\sqrt{d}}\sum_{j=1}^{d}\ket{j}\otimes\ket{j}$ is a maximally entangled state.
In our setting, $d=2^k$.

The set of Kraus operators for the logical Pauli channel \cref{eq:logical-Pauli-channel} can be taken as
$\{K_\alpha\} = \{\sqrt{p_{L}} L\}$ (choosing an arbitrary representative $L$ of each logical Pauli class acting on the code space).
Thus,
\begin{equation}
F_{\rm e}(\Lambda) =
\sum_{L}p_{L} \frac{|\Tr (L)|^{2}}{d^{2}}.
\end{equation}
For any nontrivial logical Pauli (i.e., any Pauli operator on $k$ qubits not proportional to the identity), we have $\Tr(L)=0$, while $\Tr(I)=d$.
Therefore,
\begin{equation}
F_{\rm e}(\Lambda)=p_{I} = F,
\end{equation}
where $F$ is defined in \cref{eq:F}.

\section{QEC-only}
\label{ss:QEC-only}

In this appendix, we give an independent derivation of the QEC-only fidelity expression.

The phase-stripped Pauli group $\tmP_n$ splits into the following four disjoint subsets.
For each subset, we state the logical outcome of QEC recovery on an error $E'\in\tmP_n$, referring to the notation in \cref{tab:tags}.
\begin{enumerate}
  \item $E'\in\pmS$: $E'$ is a stabilizer element, so $\syn(E')=0$.
  This never results in a logical error (independent of $\pQEC$).
  Because $D(0)=\pi(I)$, every stabilizer is trivially correctable, so $\pmS\subset\tmcEc$.
  \item $E'\in \tmcEc \setminus \pmS$: $E'$ is a correctable (non-stabilizer) error, so $\syn(E')\ne 0$, and its WEP contribution is $\C$.
  There are two cases:
  \begin{enumerate}
    \item With probability $1-\pQEC$ the decoder applies $D(\syn(E'))$ and the error is corrected to the trivial logical class.
    \item With probability $\pQEC$ the decoder fails and applies a uniformly random syndrome-consistent recovery; this yields a logical error with probability $1-4^{-k}$.
  \end{enumerate}
  \item $E'\in \pmSL \setminus \pmS$: $E'$ is an undetectable nontrivial logical operator, so $\syn(E')=0$, and its WEP contribution is $\symL$.
  This always produces a logical error (since decoder failure only affects nontrivial syndromes in our model).
  \item $E'\in \tmP_n \setminus (\tmcEc \cup \pmSL)$: $E'$ is an uncorrectable but detectable error, so $\syn(E')\ne 0$, and its WEP contribution is $\uCD$.
  There are two cases:
  \begin{enumerate}
    \item With probability $1-\pQEC$ the decoder applies $D(\syn(E'))$ and mis-corrects, resulting in a logical error.
    \item With probability $\pQEC$ the decoder fails and applies a uniformly random syndrome-consistent recovery, which yields a logical error with probability $1-4^{-k}$.
  \end{enumerate}
\end{enumerate}

\begin{mylemma}
  The infidelity under QEC is
  \begin{subequations}
  \begin{align}
    \label{eq:1-F4.general.b-1}
    1 - \FQEC &= \frac{1}{\all}\Big[\symL + \uCD + \pQEC\bigl[\C - 4^{-k}(\uCD + \C)\bigr]\Big] \\
    \label{eq:1-F4.general.b-2}
    &= \frac{1}{\all}\Big[\uC + \pQEC(\C - 4^{-k}\D)\Big] ,
  \end{align}
  \end{subequations}
  in agreement with \cref{eq:1-F4.general.b}.
\end{mylemma}

\begin{proof}
  The fidelity $\FQEC$ is the probability that after one round of the QEC-only strategy, the logical state is unchanged.
  Thus, $1-\FQEC$ is the probability of logical failure.
  Of the subsets enumerated above, only $\pmS$ (subset 1) and properly corrected errors in $\tmcEc \setminus \pmS$ (subset 2a) do not contribute to logical failures.
  Let us consider the remaining branches:
  \begin{enumerate}
    \item[3.] Nontrivial logical errors occur whether or not the decoder acts.
      Thus, this branch contributes with weight $\symL$.
    \item[4.(a)] The uncorrectable but detectable errors have weight $\uCD$ and are converted into logical errors with probability $1-\pQEC$.
    \item[2.(b)] The correctable errors have weight $\C$; the decoder fails with probability $\pQEC$ and then yields a logical error with probability $1-4^{-k}$.
    \item[4.(b)] The same happens with the uncorrectable but detectable errors, again with weight $\uCD$.
  \end{enumerate}

  Adding these four branches and dividing by the total phase-stripped Pauli group weight $W(\tmP_n;z)=\all$ yields
  \begin{equation}
  \label{eq:1-F4.general.a}
  \begin{split}
    1 - \FQEC &= \frac{1}{\all}
    \Big[\symL + (1-\pQEC)\uCD\\
    &\quad\quad + \pQEC(1-4^{-k}) \big( \C + \uCD \big)\Big] ,
  \end{split}
  \end{equation}
  which rearranges into \cref{eq:1-F4.general.b-1}.
\Cref{eq:1-F4.general.b-2} follows since the corresponding sets form disjoint unions, giving $\symL+\uCD=\uC$ and $\C+\uCD=\D$.
\end{proof}

\section{LDD-only}
\label{ss:LDD-only}

In this appendix, we give an independent derivation of the LDD-only fidelity expression.

We now have an $[[n,k,d]]$ code and a DD group $\tG$, which is a subgroup of $\tmP_n$ (in our comparisons we take $|\tG|=4^k$).
The recovery step is skipped, but the syndrome measurement still acts, which means that detected errors (i.e., errors with a nonzero syndrome) have a probability $4^{-k}$ of being mapped back to the correct logical class by a random projection.
Thus, relative to pure DD, the contribution of detected errors to logical failure is reduced by a factor $4^{-k}$.

\begin{mylemma}
  The infidelity under LDD is
  \begin{subequations}
    \begin{align}
      1 - \FLDD =\frac{1}{\uS+\pDD\symS}\bigl(&(\uSuSt -4^{-k}\uSD)\\
      +\pDD&(\SuSt-4^{-k}\SD)\bigr) ,
    \end{align}
  \end{subequations}
  in agreement with \cref{eq:1-F3.general}.
\end{mylemma}

\begin{proof}
  The denominator is unchanged relative to the DD-phys case (\cref{ss:DD-phys}): the weight of the total number of errors is the sum of the unsuppressed errors ($\uS$), which keep their probability, and the suppressed errors ($\symS$), which acquire a factor $\pDD$.

  In the numerator, we count the total weight of logical failures, split into unsuppressed-sector terms (no $\pDD$) and suppressed-sector terms (with $\pDD$).
  The term $\uSuSt$ accounts for unsuppressed non-stabilizer errors (all of which would cause a logical failure were it not for the random projection), while $\uSD$ accounts for unsuppressed detected errors, which are corrected to the trivial logical class with probability $4^{-k}$ after a random projection.
  Hence the net unsuppressed contribution is $\uSuSt-4^{-k}\uSD$.

  The suppressed-sector contribution is identical, but restricted to errors in $\tmP_n\setminus\tGp$, yielding $\SuSt-4^{-k}\SD$, and multiplied by the DD rescaling factor $\pDD$.
  Combining these contributions and dividing by the normalization $\uS+\pDD\symS$ gives the stated infidelity formula.
\end{proof}

\section{$[[7, 1, 3]]$ Decoding Map}
\label{ss:713-Decoding-Map}

\begin{table}[h]
\centering
\begin{tabular}{|c|c|c|c|}
\hline
Syndrome & Recovery Op. & Syndrome & Recovery Op. \\
\hline
000000 & \texttt{IIIIIII} & 101111 & \texttt{IIIIXIZ} \\
111000 & \texttt{IIIIIIX} & 101010 & \texttt{IIIIYIZ} \\
111111 & \texttt{IIIIIIY} & 110001 & \texttt{IIIYIIX} \\
000111 & \texttt{IIIIIIZ} & 111001 & \texttt{IIIZIIX} \\
011000 & \texttt{IIIIIXI} & 110111 & \texttt{IIIXIIY} \\
011011 & \texttt{IIIIIYI} & 111110 & \texttt{IIIZIIY} \\
000011 & \texttt{IIIIIZI} & 001111 & \texttt{IIIXIIZ} \\
101000 & \texttt{IIIIXII} & 001110 & \texttt{IIIYIIZ} \\
101101 & \texttt{IIIIYII} & 110101 & \texttt{IIIIYXI} \\
000101 & \texttt{IIIIZII} & 011101 & \texttt{IIIIZXI} \\
001000 & \texttt{IIIXIII} & 110011 & \texttt{IIIIXYI} \\
001001 & \texttt{IIIYIII} & 011110 & \texttt{IIIIZYI} \\
000001 & \texttt{IIIZIII} & 101011 & \texttt{IIIIXZI} \\
110000 & \texttt{IIXIIII} & 101110 & \texttt{IIIIYZI} \\
110110 & \texttt{IIYIIII} & 010001 & \texttt{IIIYIXI} \\
000110 & \texttt{IIZIIII} & 011001 & \texttt{IIIZIXI} \\
010000 & \texttt{IXIIIII} & 010011 & \texttt{IIIXIYI} \\
010010 & \texttt{IYIIIII} & 011010 & \texttt{IIIZIYI} \\
000010 & \texttt{IZIIIII} & 001011 & \texttt{IIIXIZI} \\
100000 & \texttt{XIIIIII} & 001010 & \texttt{IIIYIZI} \\
100100 & \texttt{YIIIIII} & 100001 & \texttt{IIIYXII} \\
000100 & \texttt{ZIIIIII} & 101001 & \texttt{IIIZXII} \\
100011 & \texttt{IIIIIYX} & 100101 & \texttt{IIIXYII} \\
111011 & \texttt{IIIIIZX} & 101100 & \texttt{IIIZYII} \\
100111 & \texttt{IIIIIXY} & 001101 & \texttt{IIIXZII} \\
111100 & \texttt{IIIIIZY} & 001100 & \texttt{IIIYZII} \\
011111 & \texttt{IIIIIXZ} & 100010 & \texttt{IYXIIII} \\
011100 & \texttt{IIIIIYZ} & 110010 & \texttt{IZXIIII} \\
010101 & \texttt{IIIIYIX} & 100110 & \texttt{IXYIIII} \\
111101 & \texttt{IIIIZIX} & 110100 & \texttt{IZYIIII} \\
010111 & \texttt{IIIIXIY} & 010110 & \texttt{IXZIIII} \\
111010 & \texttt{IIIIZIY} & 010100 & \texttt{IYZIIII} \\
\hline
\end{tabular}
\caption{Decoding map of the $[[7, 1, 3]]$ code.}
\label{tab:713}
\end{table}

For the $[[7,1,3]]$ Steane code, we use the standard CSS stabilizer generators, ordered as
\begin{equation}
\begin{aligned}
  S_1&=\texttt{ZIZIZIZ},\quad
  S_2=\texttt{IZZIIZZ},\quad
  S_3=\texttt{IIIZZZZ}\\
  S_4&=\texttt{XIXIXIX},\quad
  S_5=\texttt{IXXIIXX},\quad
  S_6=\texttt{IIIXXXX}.
  \end{aligned}
\end{equation}
We represent a syndrome $\sigma\in\mS^*$ as the $6$-bit string $s_1s_2s_3s_4s_5s_6$ with
$s_j=\sigma(S_j)\in\{0,1\}$, i.e., $s_j=1$ iff the error anticommutes with $S_j$
(equivalently, the measurement outcome of $S_j$ is $-1$).
With this convention, \cref{tab:713} specifies which recovery operation $D(\sigma)$ is applied for each syndrome.
The table lists all $|\mS^*|=2^{n-k}=64$ syndromes, including the trivial syndrome $000000$.
With this decoding map,
\begin{equation}
\label{eq:tmcEc}
  \tmcEc=\bigcup_{\sigma\in\mS^*} D(\sigma) \pmS,
\end{equation}
has size $|\mS^*| |\pmS|=2^{2(n-k)}=4096$, i.e., the decoder corrects exactly one stabilizer coset per syndrome. Each corrected error is stabilizer-equivalent to one of the $64$ weight-$0$ through weight-$2$ Pauli strings listed as recovery operations (one per syndrome).

The way the decoding map in \cref{tab:713} was chosen is the following. First, we wish to correct all weight-$1$ errors. Hence, if, for a given syndrome $\sigma$, there is $E' \in \tmP_n$ with $\syn(E') = \sigma$ and $\wt(E') \leq 1$, we pick $D(\sigma) = E'$. For the remaining $42$ syndromes $\sigma$, the minimal weight errors $E'$ with $\syn(E') = \sigma$ have weight $2$, and there $3$ such errors for each syndrome, so the choice of $D(\sigma)$ is somewhat arbitrary: for any such choice all $63$ $\texttt{XX}$, $\texttt{YY}$, and $\texttt{ZZ}$ errors will be uncorrectable (because they share a syndrome with a single-qubit error); out of the $126$ $\texttt{XZ}$, $\texttt{ZY}$, $\texttt{YX}$ errors, $42$ will be correctable and the remaining $84$ will be uncorrectable. We make the following choice.
First, we enumerate all weight-$2$ Pauli strings on $7$ qubits.  These weight-$2$ Pauli strings are enumerated as follows.  First, choose the support pair $Q = (q_1, q_2)$ lexicographically.  Then, assign $(\sigma_1, \sigma_2) \in \{X, Y, Z \}^{2}$ lexicographically.  Each $\sigma_j$ is placed on qubit $q_j$, and all other qubits receive $I$ support.  After enumerating all weight-$2$ Pauli strings using this rule, we iterate through the list.  For each weight-$2$ Pauli string, we check which syndrome is generated.  If no error with this syndrome is already in the recovery map, we add the error to the recovery map.  Otherwise, we continue and the error is deemed uncorrectable.

\section{[[5, 1, 3]] Decoding Map}
\label{ss:513-Decoding-Map}

\begin{table}[h]
\centering
\begin{tabular}{|c|c|c|c|}
\hline
Syndrome & Recovery Op. & Syndrome & Recovery Op. \\
\hline
0000 & \texttt{IIIII} & 0111 & \texttt{IIYII} \\
1100 & \texttt{IIIIX} & 0100 & \texttt{IIZII} \\
1110 & \texttt{IIIIY} & 0001 & \texttt{IXIII} \\
0010 & \texttt{IIIIZ} & 1011 & \texttt{IYIII} \\
0110 & \texttt{IIIXI} & 1010 & \texttt{IZIII} \\
1111 & \texttt{IIIYI} & 1000 & \texttt{XIIII} \\
1001 & \texttt{IIIZI} & 1101 & \texttt{YIIII} \\
0011 & \texttt{IIXII} & 0101 & \texttt{ZIIII} \\
\hline
\end{tabular}
\caption{Decoding map of the $[[5, 1, 3]]$ code.}
\label{tab:513}
\end{table}
For the $[[5,1,3]]$ code, we use the standard stabilizer generators, ordered as
\begin{equation}
  S_1 = \texttt{ZXIXZ},\quad
  S_2 = \texttt{XIXZZ},\quad
  S_3 = \texttt{IXZZX},\quad
  S_4 = \texttt{XZZXI}.
\end{equation}
We represent a syndrome $\sigma\in\mS^*$ as the $4$-bit string $s_1s_2s_3s_4$ with
$s_j=\sigma(S_j)\in\{0,1\}$.
With this convention, \cref{tab:513} specifies which recovery operation $D(\sigma)$ is applied for each syndrome.
The table lists all $|\mS^*|=2^{n-k}=16$ syndromes, including the trivial syndrome $0000$.
With this decoding map, $\tmcEc$ as defined in \cref{eq:tmcEc}
has size $|\mS^*| |\pmS|=2^{2(n-k)}=256$, i.e., the decoder again corrects exactly one stabilizer coset per syndrome.
Each corrected error is stabilizer-equivalent to one of the $16$ weight-$0$ through weight-$1$ Pauli strings listed in \cref{tab:513} as recovery operations (one per syndrome).

\end{document}